\documentclass[%
reprint,
superscriptaddress,
nofootinbib,
amsmath,amssymb,
aps,
prb,
]{revtex4-2}

\usepackage{amsthm}
\usepackage{physics}
\usepackage{xcolor}
\usepackage{graphicx}
\usepackage{bm}
\usepackage{hyperref}
\usepackage{amsmath}
\usepackage{tensor}
\usepackage{MnSymbol}

\hypersetup{
colorlinks=true,
linkcolor=blue,
urlcolor=cyan,
citecolor=red,
}
\usepackage[normalem]{ulem}
\usepackage{verbatim}
\usepackage{enumitem}

\renewcommand{\(}{\left(}
\renewcommand{\)}{\right)}
\renewcommand{\[}{\left[}
\renewcommand{\]}{\right]}
\renewcommand{\v}[1]{\boldsymbol{#1}} 

\newtheorem{theorem}{Theorem}
\newtheorem{corollary}{Corollary}[theorem]
\newtheorem{fact}[theorem]{Fact}

\theoremstyle{definition}
\newtheorem{definition}{Definition}[section]
\newtheorem*{remark}{Remark}

\begin{document}

\title{Emergent unitary designs for encoded qubits from coherent errors and syndrome measurements} 

\author{Zihan Cheng}
\affiliation{Department of Physics, University of Texas at Austin, Austin, TX 78712, USA}

\author{Eric Huang}
\affiliation{Joint Center for Quantum Information and Computer Science, NIST/University of Maryland, College Park, MD 20742, USA}

\author{Vedika Khemani}
\affiliation{Department of Physics, Stanford University, Stanford, CA 94305, USA}

\author{Michael J. Gullans}
\affiliation{Joint Center for Quantum Information and Computer Science, NIST/University of Maryland, College Park, MD 20742, USA}

\author{Matteo Ippoliti}
\affiliation{Department of Physics, University of Texas at Austin, Austin, TX 78712, USA}

\date{\today}

\begin{abstract}
Unitary $k$-designs are distributions of unitary gates that match the Haar distribution up to its $k$-th statistical moment. They are a crucial resource for randomized quantum protocols. However, their implementation on encoded logical qubits is nontrivial due to the need for magic gates, which can require a large resource overhead. In this work, we propose an efficient approach to generate unitary designs for encoded qubits in surface codes by applying local unitary rotations (``coherent errors'') on the physical qubits followed by syndrome measurement and error correction. We prove that under some conditions on the coherent errors (notably including all single-qubit unitaries) and on the error correcting code, this process induces a unitary transformation of the logical subspace. 
We numerically show that the ensemble of logical unitaries (indexed by the random syndrome outcomes) converges to a unitary design in the thermodynamic limit, provided the density or strength of coherent errors is above a finite threshold. This ``unitary design'' phase transition coincides with the code's coherent error threshold under optimal decoding. 
Furthermore, we propose a classical algorithm to simulate the protocol based on a ``staircase'' implementation of the surface code encoder and decoder circuits. This enables a mapping to a 1+1D monitored circuit, where we observe an entanglement phase transition (and thus a classical complexity phase transition of the decoding algorithm) coinciding with the aforementioned unitary design phase transition. 
Our results provide a practical way to realize unitary designs on encoded qubits, with applications including quantum state tomography and benchmarking in error correcting codes.
\end{abstract}

\maketitle

\tableofcontents


\section{Introduction \label{sec:intro} }

Quantum error correction is crucial to the realization of scalable quantum computations with real, noisy hardware. At the same time, encoded quantum information tends to be harder to manipulate. Quantum error correcting codes have some restricted set of logical operations that can be implemented ``transversally'', i.e., by separately rotating individual physical qubits---this is desirable for fault-tolerance, as transversal operations do not propagate errors. But these transversal logical operations are not universal for quantum computing~\cite{eastin_restrictions_2009}. For surface codes and other codes of practical interest, ``magic'' (i.e., non-Clifford) gates are among those that cannot be implemented transversally. Strategies to overcome this issue, such as magic state distillation, come with some overhead in terms of auxiliary qubits or circuit size~\cite{knill_fault-tolerant_2004,bravyi_universal_2005,hastings_distillation_2018,wills_constant-overhead_2024}. This motivates the search for alternative strategies that can implement generic unitary operations on encoded qubits without additional resources.

In this work we address this issue within the context of {\it random unitary gates}, which are a key component of various quantum information processing protocols like randomized benchmarking, randomized measurements for state- and process-learning, cryptography, random circuit sampling, quantum simulation, etc. Our strategy is based on the application of transversal physical operations (which may be seen as ``coherent errors'', though applied intentionally and known to the experimentalist) followed by syndrome extraction and error correction, as sketched in Fig.~\ref{fig:main_idea}(a). The Born-rule randomness inherent in quantum measurements automatically gives rise to a random distribution of operations on the encoded information; remarkably, under some criteria that we characterize, these operations are unitary, and the resulting ensemble can achieve a universal form in the thermodynamic limit (namely a unitary $k$-design, i.e., an approximation to the Haar distribution on the unitary group). 

This phenomenon is reminiscent of ``deep thermalization''~\cite{Choi2023preparing,Cotler2023emergent,ho_exact_2022,ippoliti_solvable_2022}, the emergence of universal distributions of post-measurement {\it states} on a spatially local subsystem from measuring its environment, but with two key differences: the distribution is over unitary operators instead of states, and the partition between measured and un-measured degrees of freedom is not based on real space, but rather on {\it syndrome} versus {\it logical} degrees of freedom. 

In this work we focus on the two-dimensional surface code, and find that the emergence of unitary designs in the thermodynamic limit occurs only above a critical strength (or density) of the coherent errors, as summarized in Fig.~\ref{fig:main_idea}(b).
Remarkably, this threshold for design formation is also a threshold for three other related but distinct phenomena:
(i) an {\it error correction} threshold for the code under coherent errors and optimal decoding;
(ii) an {\it entanglement} phase transition in an associated $(1+1)$-dimensional monitored dynamics;
(iii) a {\it computational complexity} phase transition for a classical decoder based on matrix product states that takes advantage of the aforementioned entanglement transition.
Our results add to other recently discovered examples of coding~\cite{Turkeshi2024error}, ``magic''~\cite{niroula_phase_2024}, entanglement~\cite{behrends_surface_2024} and delocalization~\cite{Venn_coherent_2023} transitions in similar setups involving coherent errors on encoded qubits. 


The rest of this paper is structured as follows. In Sec.~\ref{sec:background} we review essential background concepts including state and unitary designs, projected ensembles, and measurement-induced phase transitions. Sec.~\ref{sec:pe_unitary} contains our first main result, regarding projected ensembles of unitary operations arising from coherent errors and syndrome measurements in error correcting codes. 
In Sec.~\ref{sec:emergent_unitary_design} we numerically study the emergence of unitary designs from these ensembles, focusing on the surface code. We then introduce a mapping to $(1+1)$-dimensional monitored dynamics via staircase encoder/decoder circuits in Sec.~\ref{sec:staircase}, where we numerically uncover an entanglement transition and analytically connect it to the unitary design threshold. 
Finally we discuss our results and point to future directions in Sec.~\ref{sec:discussion}.

\begin{figure*}
    \centering
    \includegraphics[width=1\linewidth]{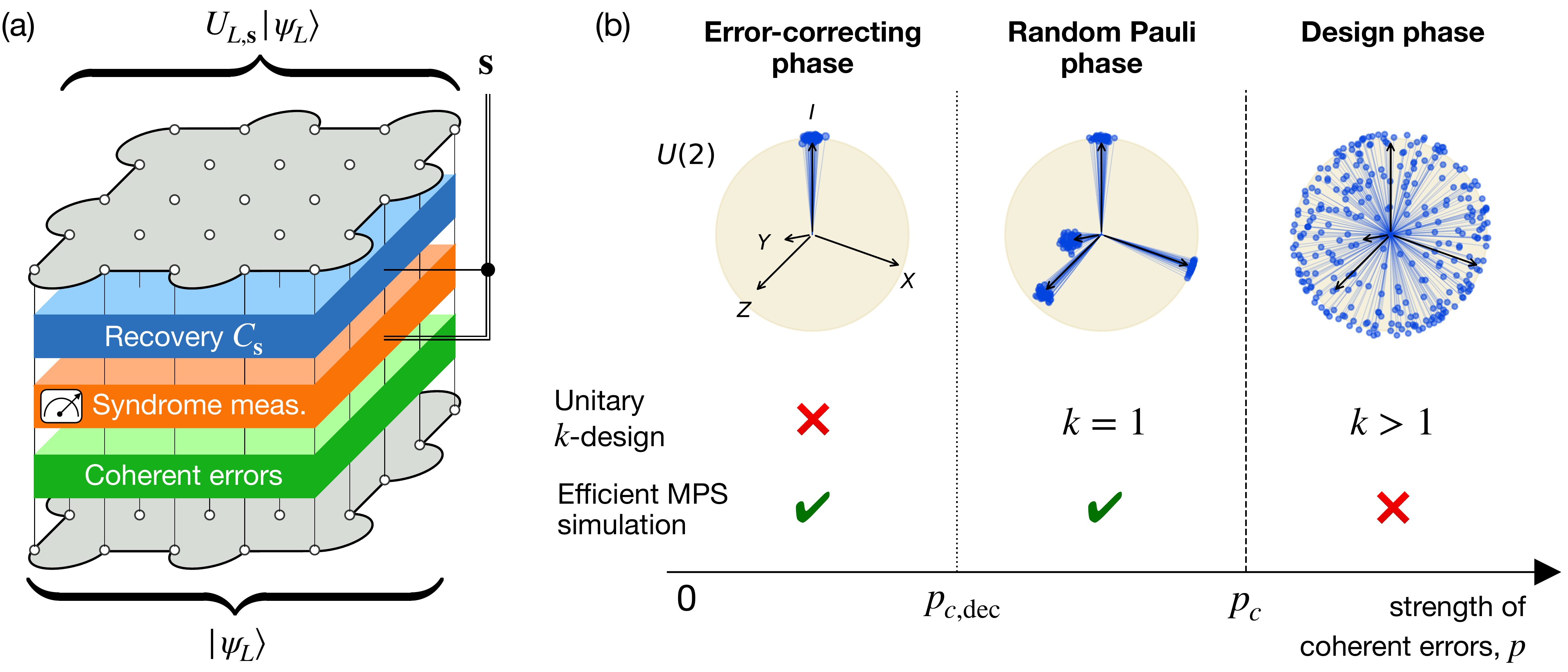}
    \caption{Main ideas of this work. (a) Intentional application of coherent errors to a quantum error-correcting code, followed by syndrome measurement and correction, yields (under the assumptions of Theorem~\ref{theo:logicalunitary}) a logical unitary operation that depends on the syndrome outcome $\mathbf s$, i.e., a projected ensemble of logical unitaries. 
    (b) Depending on the strength of coherent errors, the projected ensemble can take sharply different forms in the limit of large code size: an ``error-correcting phase'', where all unitaries in the ensemble are very close to the identity with high probability; a ``unitary design phase'', where the distribution of unitaries becomes highly uniform over the logical unitary group $U(2)$; and an intervening ``random Pauli phase'' where the distribution clusters around Pauli unitaries $I$, $X$, $Y$, $Z$ with equal probability. Qualitative sketches of these behaviors are shown at the top. The threshold $p_{c,{\rm dec}}$ between the error-correcting and random-Pauli phase depends on the choice of decoding algorithm, while the threshold $p_c$ of the unitary design phase is intrinsic (achieved by the optimal decoder). We find that it also corresponds to a transition in complexity for a MPS-based classical simulation algorithm.
    }
    \label{fig:main_idea}
\end{figure*}

\section{Background \label{sec:background}}

\subsection{State designs and unitary designs \label{sec:review_designs}}

Many protocols in quantum information science take advantage of the properties of random unitary operations or random quantum states~\cite{knill_randomized_2008,mcclean_barren_2018,huang_predicting_2020,elben_randomized_2023,mele_introduction_2024}.
However, Haar-random quantum states or unitaries have exponentially high complexity, making them impractical. This motivates the introduction of {\it designs}---distributions of quantum states or unitaries that capture only some statistical moments of the Haar distribution, but can be implemented efficiently~\cite{ambainis_quantum_2007,dankert_exact_2009,roberts_chaos_2017,hunter-jones_unitary_2019,schuster_random_2024}. Here we briefly review the definitions of quantum state and unitary designs and related concepts.

\begin{definition}[State ensemble, moment operator]
An ensemble of states is a probability distribution over the Hilbert space: $\mathcal{E} = \{p_i,\ket{\psi_i}\}$, where state $\ket{\psi_i}$ occurs with probability $p_i$. 
Given an ensemble of states $\mathcal{E} = \{ p_i,\ket{\psi_i}\}$, its $k$-th moment operator is
\begin{align}
    \rho^{(k)}_{\mathcal E} 
    = \mathbb{E}_{\psi \sim \mathcal{E}}[ (\ketbra{\psi})^{\otimes k} ]
    = \sum_i p_i (\ketbra{\psi_i})^{\otimes k}.
\end{align}
This is a mixed state on $k$ copies of the system.
\end{definition}

\begin{definition}[Unitary ensemble, twirling channel]
An ensemble of unitaries is a probability distribution over the unitary group $U(D)$ ($D$ is the Hilbert space dimension): $\mathcal{E} = \{p_i,U_i \}$, where the unitary operator $U_i$ occurs with probability $p_i$. 
Given an ensemble of unitaries $\mathcal{E} = \{ p_i,U_i\}$, its $k$-fold twirling channel is
\begin{align}
    \Phi^{(k)}_{\mathcal E}[\cdot] 
    = \mathbb{E}_{U \sim \mathcal{E}} \left[ U^{\otimes k}[\cdot] (U^\dagger)^{\otimes k} \right]
    = \sum_i p_i U_i^{\otimes k } [\cdot] (U_i^{\dagger})^{\otimes k}.
\end{align}
This is a channel on $k$ copies of the system.
\end{definition}

\begin{remark}
    More generally, for both definitions one could use a continuous rather than atomic probability measure, but that is not necessary for the purpose of this work.
\end{remark}

We are interested in ensembles of quantum states or unitaries that approximate certain statistical moments of the Haar distribution. This notion is made precise by the following definitions:
\begin{definition}[$\epsilon$-approximate $k$-design]
We say that an ensemble of states $\mathcal E$ forms an $\epsilon$-approximate quantum state $k$-design if 
\begin{align}
    \Delta^{(k)}_{\rm state} = \left\Vert\rho_{\mathcal{E}}^{(k)}-\rho^{(k)}_{\rm Haar}\right\Vert_1 \leq \epsilon,
\end{align}
where $k\geq 1$ is an integer, $\Vert A \Vert_{1} = \frac{1}{2} {\rm Tr}(\sqrt{A^\dagger A})$ is the trace norm, and $\rho^{(k)}_{\rm Haar} = \mathbb{E}_{\psi\sim{\rm Haar}} (\ketbra{\psi})^{\otimes k}$.
We say that an ensemble of unitaries $\mathcal{E}$ forms an $\varepsilon$-approximate unitary $k$-design if 
\begin{align}
    \Delta^{(k)} = \left\Vert\Phi^{(k)}_{\mathcal E} -\Phi^{(k)}_{\rm Haar}\right\Vert_\Diamond \leq \epsilon,
\end{align}
where 
\begin{equation} 
\Vert \Phi \Vert_\Diamond = \max_{X} \frac{ \Vert (\Phi \otimes \mathcal{I})(X) \Vert_1} {\Vert X\Vert_1}
\end{equation}
is the diamond norm and $\Phi^{(k)}_{\rm Haar}[\cdot] = \mathbb{E}_{U\sim {\rm Haar}}U^{\otimes k}[\cdot] (U^\dagger)^{\otimes k}$. Here $\mathcal{I}$ is the identity channel on a second copy of the Hilbert space, and $X$ is chosen among operators acting on both copies.
\end{definition}

\subsection{Projected ensembles \label{sec:review_pe}}

The projected ensemble, introduced by Refs.~\cite{Cotler2023emergent,Choi2023preparing}, is an ensemble of quantum states generated by projective measurements on part of a many-body system. The Born-rule randomness of measurement outcomes induces a probability distribution over wavefunctions of the subsystem left unmeasured; remarkably, this distribution can be a quantum state $k$-designs in many cases of interest. More generally, the emergence of universal distributions (including but not limited to $k$-designs) from the projected ensemble goes under the name of ``deep thermalization''~\cite{Choi2023preparing,Cotler2023emergent,ho_exact_2022,ippoliti_dynamical_2023,ippoliti_solvable_2022,bhore_deep_2023,mark_maximum_2024,chan_projected_2024}. Here we briefly review the construction of the projected ensemble for states; our work will generalize this construction to an ensemble of unitaries. 

Consider a pure quantum state $|\Psi \rangle_{AB} $ on a bipartite system $AB$ consisting of $N_A$ qubits in $A$ and $N_B$ qubits in $B$. Upon performing projective measurements on $B$, one obtains a measurement outcome (bitstring) $\v{z}\in\{0, 1\}^{N_B}$ on $B$ and a corresponding post-measurement state on $A$, $\ket{\psi_{\v z}}_A$. The projected ensemble can be defined as
\begin{align}
    \mathcal{E}_{\psi, A}=\{p(\v{z}), |\psi_{\v{z}}\rangle_A \},
\end{align}
where $\ket{\psi_{\v z}}_A = {}_B\!\braket{\v z}{\Psi}_{AB} / \sqrt{p(\v z)}$ is the (normalized) post-measurement state and $p(\v z)$ is the Born-rule probability of outcome $\v z\in \{0,1\}^{N_B}$. 

It has been shown that if $|\Psi\rangle_{AB}$ is a random many-body state, for sufficiently large $N_B$ the projected ensemble $\mathcal{E}_{\psi, A}$ forms an approximate state $k$-design with high probability~\cite{Cotler2023emergent}. 
Other universal distributions have been identified for Hamiltonian eigenstates at variable energy~\cite{Cotler2023emergent,mark_maximum_2024}, for fermionic or bosonic Gaussian states~\cite{lucas_generalized_2023,liu_deep_2024}, or for random states that respect a symmetry~\cite{chang_deep_2024}.

It is natural to ask whether a similar construction could be used to realize not just state designs, but also unitary designs. However a na\"ive attempt to generalize this construction runs into the issue that the operations generated by a projected ensemble are typically {\it not} unitary. To be precise, consider a unitary transformation $U_{AB}$ acting on a bipartite state $\ket{\psi}_A\otimes \ket{0}_B$, where $\ket{0}_B$ is some fixed state of $B$ while $\ket{\psi}_A$ is an arbitrary, variable input state on $A$; projective measurements on $B$ then give rise to an ensemble of operations on $A$ $\{K_{\v z} = {}_B\!\bra{\v z} U_{AB} \ket{0}_B\}$. These are Kraus operators that unravel a quantum channel $\Phi(\cdot) = {\rm Tr}_B [U_{AB} ([\cdot]_A \otimes \ketbra{0}_B) U_{AB}^\dagger]$ for which $U_{AB}$ is a Stinespring dilation.
Generically, these Kraus operators are non-unitary, but rather correspond to weak measurements of subsystem $A$. 
This is also manifested in the fact that the probability distribution over Kraus operators $K_{\v z}$ depends on the input state $\ket{\psi}$ on $A$, $p(\v z) = \bra{\psi} K_{\v z}^\dagger K_{\v z} \ket{\psi}$. So $\mathcal{E}$ is not an intrinsically-defined ensemble of transformations.
The two issues (unitarity of the transformations and independence of the probability distribution from the input state) are in fact equivalent:
\begin{fact}\label{fact:unitarity}
    Consider a set of Kraus operators $\{K_{\v{z}}\}$. The probability distribution $p(\v{z})=\Tr(K_{\v{z}}\rho K^\dagger_{\v{z}})$ is independent of the input state $\rho$ if and only if each $K_{\v{z}}$ is proportional to a unitary operator:
    \begin{align}
        K_{\v{z}}=\sqrt{p(\v{z})}U_{\v{z}}.
    \end{align}
\end{fact}
\begin{proof}
This is a well-known fact in quantum information theory, here we report a simple proof for the sake of completeness.
The \textit{if} statement is obvious. As for the {\it only if}, we start by defining the POVM $\{ E_{\v{z}}=K^{\dagger}_{\v{z}}K_{\v{z}}\} $, such that $p(\v{z})=\Tr(E_{\v{z}}\rho)$ for any state $\rho$. Let $\{\ket{\alpha}\}$ be an orthonormal basis of the Hilbert space; then
\begin{align}\label{eq:lemma1}
    E_{\v{z}, \alpha\alpha} = {\rm Tr}(E_{\v z} \ketbra{\alpha}) = p(\v{z}) \ \forall\, \alpha.
\end{align}
Furthermore, for any pair of basis vectors $|\alpha\rangle$, $|\beta\rangle$ $(\beta\neq \alpha)$, we can let $\rho=\frac{1}{2}(|\alpha\rangle+|\beta\rangle)(\langle \alpha|+\langle \beta |)$, which yields
\begin{align}\label{eq:lemma2}
    E_{\v{z}, \alpha\beta}+E_{\v{z}, \beta\alpha}=0.
\end{align}
The same reasoning for $\rho=\frac{1}{2}(|\alpha\rangle+i|\beta\rangle)(\langle \alpha|-i\langle \beta |)$ yields 
\begin{align}\label{eq:lemma3}
    E_{\v{z}, \alpha\beta}-E_{\v{z}, \beta\alpha}=0.
\end{align}
Combining Eqs.\eqref{eq:lemma1}-\eqref{eq:lemma3}, we have 
\begin{align}
E_{\v{z},\alpha\beta} = p(\v{z}) \delta_{\alpha\beta}
\end{align}
and thus $K_{\v{z}}=\sqrt{p(\v{z})}U_{\v{z}}$ for some unitaries $U_{\v z}$. 
\end{proof}

Past works have identified a setting where one can obtain projected ensembles of unitary transformations, $\mathcal{E} = \{ p_{\v z}, U_{\v z} \}$ where the $U_{\v z}$ are unitary and the $p(\v z)$ are intrinsically defined; the key ingredient is quantum error correction.
Specifically 
Ref.~\cite{Bravyi2018correcting} focuses on the surface code and identifies certain unitary ensembles (made only of $z$ rotations) that can be realized from syndrome measurements.
In this work, we build on that result to allow for the realization of projected ensembles of {\it general} unitaries on a single qubit, thus enabling the formation of unitary designs.

\subsection{Monitored circuits and the SEBD algorithm \label{sec:review_sebd}}

Recent years have seen a surge of interest in ``monitored'' quantum dynamics~\cite{fisher_random_2023,Potter2022entanglement}: time evolution where unitary interactions coexist with measurements. The interest stems from the discovery of sharp thresholds, or phase transitions, in the structure of quantum correlations in these dynamics or in their output states. While there are many variations of these phenomena, the simplest instance is represented by local unitary circuits interspersed with projective measurements, where the rate $p$ of measurements (probability that a qubit is measured after each gate) tunes a transition from an area-law entangled phase to a volume-law entangled one~\cite{Skinner2019measurement,li_quantum_2018,Choi2020quantum,Bao2020theory}. These phases are closely related to ideas of quantum error correction and decoding~\cite{Gullans2020dynamical,Gullans2020scalable,fan_self-organized_2021,li_statistical_2021,li_cross_2023,noel_measurement-induced_2022,hoke_measurement-induced_2023}.

A significant application of monitored quantum dynamics lies in the classical simulation of quantum dynamics---even in the absence of measurements. As an example, noise in the form of single-qubit dephasing or depolarizing channels can be unraveled into measurements; quantum trajectory methods then effectively simulate a monitored time evolution, and it may be possible to leverage the entanglement phase transition for efficient tensor network simulation (depending on spatial dimensionality)~\cite{vovk_entanglement-optimal_2022,Cheng2023efficient,chen_optimized_2023}.

Another way to leverage the entanglement phase transition for efficient classical simulation, with or without noise, applies to sampling experiments on shallow circuits~\cite{Napp2022efficient,Cheng2023efficient}. There, an approach dubbed ``space-evolving block decimation'' (SEBD) effectively trades one spatial dimension for a time dimension, turning the $D$-dimensional shallow unitary circuit into a deep monitored circuit in $D-1$ dimensions. Loosely speaking, by sweeping through the circuit along one spatial direction, the measurements occurring at the end of the original circuit (responsible for the sampling) are effectively spread out throughout the dynamics. 
Crudely associating a ``measurement rate'' $p\sim 1/T$ where $T$ is the shallow circuit's depth (each qubit participates in $T$ unitary gates and one measurement) yields a phase transition in the complexity of SEBD controlled by the depth, with efficient simulation below a finite critical depth $T_c$.

In this work we adapt the SEBD approach to the distinct setting of syndrome measurements on the surface code---a nontrivial generalization since single-qubit measurements are essential to the formulation of SEBD. To this end we will leverage ``staircase'' encoder and decoder circuits for the surface codes. 


\section{Projected ensemble of logical unitaries}\label{sec:pe_unitary}

In this section we study conditions under which syndrome measurements on an error correcting code give rise to unitary operations, building on and extending the results of Ref.~\cite{Bravyi2018correcting}. 

Based on the discussion in Sec.~\ref{sec:review_pe}, and specifically Fact~\ref{fact:unitarity}, a projected ensemble of {\it unitary} operators requires that the measurements reveal no information about the input state. This naturally points to the framework of quantum error correction, where syndrome measurements are carefully chosen so as to reveal no information about the logical input state. 
Consider a quantum error correcting code initialized in a logical state, then subject to some unitary evolution $U$ followed by syndrome measurement and correction. Generally $U$ takes the system outside the code subspace, while the syndrome measurement and correction returns it to the code subspace. The goal of error correction is to return the system to the same logical state it started from, but depending on $U$ this may fail, so the final logical state is generally different from the initial one. We therefore obtain an ensemble of linear transformations of the code space, and it is natural to ask under what conditions these transformations are unitary.

A setting where this was shown to be the case is the surface code under coherent $Z$ errors~\cite{Bravyi2018correcting,Venn2020error-correction}. Ref.~\cite{Bravyi2018correcting} shows that starting with a logical state $|\psi_L\rangle$ of the surface code, single-qubit $Z$-rotations $U=\prod_j\exp(i\eta_jZ_j)$ followed by syndrome measurement and correction result in an effective logical $Z$-rotation $U_{L,\v{s}}=\exp\(i\theta_{\v{s}}Z_L\)$ (indexed by the syndrome outcome $\v{s}$) provided the surface code has odd distance. 
Importantly, by Fact~\ref{fact:unitarity}, unitarity of $U_{L,\v{s}}$ implies that the syndrome probability distribution $p(\v{s})$ is independent of the initial logical state $|\psi_L\rangle$. As a result we have a well-defined projected ensemble of logical unitaries:
\begin{align}\label{eq:z_ensemble}
    \mathcal{E}_{U, L}=\{p(\v{s}), \exp\(i\theta_{\v{s}}Z_L\)\}.
\end{align}
However the proof technique employed in Ref.~\cite{Bravyi2018correcting} only applies to $Z$ rotations (or equivalently $X$ rotations, by the CSS duality of the surface code), which limits the projected ensemble to rotations about a single axis. Such ensembles clearly fail to form unitary designs for the $U(2)$ logical group (they may at most form designs for a $U(1)$ subgroup). 

Here we generalize the result above in two ways: by considering codes beyond the surface code, and by allowing physical operations beyond single-qubit $Z$ (or $X$) rotations. This allows the realization of general $U(2)$ unitaries on the logical qubit. 

We consider a general stabilizer code prepared in a logical state $|\psi_L\rangle$, and subject to a physical unitary $U$:
\begin{align}
    \vert\psi'\rangle = U\ket{\psi_L}.
\end{align}
The state $\ket{\psi'}$ is generally not a logical state ($U$ may be seen as a ``coherent error'' on the code).
Our goal is to find constraints on $U$ and on the error-correcting code such that the final logical state, obtained from $\ket{\psi'}$ by syndrome measurement and correction, is related to the input logical state $\ket{\psi_L}$ by a logical unitary operation that is {\it independent of the input state} (see Fig.~\ref{fig:main_idea}(a)). 

To address this question we must first introduce some notation. 
Let $\{g_i\}$ be a minimal set of stabilizer generators for the error-correcting code. We introduce the syndrome projectors $\Pi_{\v{s}}=\prod_i\frac{\mathbb{I}+(-1)^{s_i}g_i}{2}$. To each syndrome we associate a recovery operation $C_{\v s}$, i.e. a Pauli unitary that obeys $\Pi_{\v s} C_{\v s} = C_{\v s} \Pi_{\v 0}$ (it returns error states with syndrome $\v s$ to the code space, characterized by the trivial syndrome $\v 0$). Note that the $C_{\v s}$ are only defined modulo stabilizer and logical operations---the specific choice of recovery $\{C_{\v s}\}$ defines a decoder. Our statements here hold for any choice of $\{C_{\v s}\}$. 
With this notation, we can write the final state of the protocol as 
\begin{align}
    \ket{ \phi_{\v{s}, L}} = \frac{C_{\v s} \Pi_{\v s} U\ket{\psi_L}}{ \Vert C_{\v s} \Pi_{\v s} U\ket{\psi_L} \Vert },
\end{align}
occurring with probability $\bra{\psi_L}U^\dagger \Pi_{\v s} U\ket{\psi_L}$ (probability of measuring syndrome $\v s$). 
We aim to understand when the above reduces to an ensemble of unitary operations on the logical space.  
We answer this question by proving the following theorem:
\begin{theorem}\label{theo:logicalunitary}
    Consider a stabilizer code with one logical qubit and a unitary $U$ on the physical qubits.
    Assume the following three conditions hold: 
    \begin{itemize}
        \item[(1)] All the stabilizers have even Pauli weight;
        \item[(2)] The code is CSS, and its $X$ and $Z$ code distances $d_x$, $d_z$ are odd; 
        \item[(3)] The physical unitary $U$ commutes with the time-reversal transformation $\mathcal{T} = \mathcal{K} \prod_j (iY_j)$ (with $\mathcal{K}$ the complex conjugation);
    \end{itemize} 
    then
    \begin{equation}
        C_{\v s} \Pi_{\v s} U \Pi_{\v 0} = \sqrt{p(\v s)} U_{L, \v{s}} \Pi_{\v 0},
        \label{eq:thm1}
    \end{equation}
    where $p(\v{s})$ is a probability distribution over syndrome outcomes $\v s$ and $U_{L,\v{s}}$ are unitaries on the logical subspace.
\end{theorem}
\begin{proof}
    Here we only provide an outline of the proof; see Appendix~\ref{ap:theo1proof} for details.
    For any Pauli operator $P$, one has $\mathcal{T} P \mathcal{T}^{-1} = (-1)^{|P|} P$ with $|P|$ the Pauli weight (number of non-identity Pauli matrices) of $P$. Conditions (1) and (2) imply that all nontrivial logical operators have odd Pauli weight, so all the logical information is stored in time-reversal-odd operators; this remains true after application of $U$, by condition (3). But the syndrome measurements only reveal information about time-reversal-even operators (by condition 1 on the stabilizer generators), so they reveal no information about the logical state. Unitarity then follows by Fact~\ref{fact:unitarity}. 
\end{proof}
\begin{corollary}
For any logical input state $\ket{\psi_L}$, the syndrome probability distribution on $U\ket{\psi_L}$ is $p(\v s) = \frac{1}{2} {\rm Tr}(U \Pi_{\v 0} U^\dagger \Pi_{\v s})$, independent of the input state; the output state conditional on syndrome $\v s$ is
\begin{equation}
    C_{\v s}\Pi_{\v s} U\ket{\psi_L}\propto U_{L,\v{s}} \ket{\psi_L}.
\end{equation}
\end{corollary}

Theorem~\ref{theo:logicalunitary} asserts that, when conditions (1-3) on the error-correcting code and physical operation $U$ are met, one indeed obtains a projected ensemble of unitary operations on the logical subspace. We will next discuss when these conditions arise. 

Condition (1), on the Pauli weight of stabilizers, is met whenever one has a set of stabilizer generators $\{g_i\}$ where each $g_i$ has even Pauli weight\footnote{Since any two $g_i$, $g_j$ commute, $g_i g_j$ is a Hermitian Pauli operator; thus $\mathcal{T} g_i g_j \mathcal{T}^{-1} = (-1)^{|g_i g_j|}$, but also $\mathcal{T} g_i g_j \mathcal{T}^{-1} = \mathcal{T} g_i \mathcal{T}^{-1} \mathcal{T} g_j \mathcal{T}^{-1} = (-1)^{|g_i|+|g_j|} g_i g_j$. Thus $|g_i g_j| = |g_i| + |g_j|\mod 2$. So if every element in a generating set $\{g_i\}$ has even Pauli weight, then each element of the stabilizer group does.}. 
This is true in many codes of interest, for instance in the toric code (4-body generators), rotated surface code (4-body generators in the bulk and 2-body on the boundaries, Fig.~\ref{fig:surfacecode}), color code (6-body generators), etc. 

Condition (2) is straightforward and can typically be realized in families of codes such as the surface code (Fig.~\ref{fig:surfacecode}) or the color code by a choice of system size. We note that this condition can be relaxed, as our proof only uses a weaker property: that all nontrivial logical operators must have odd Pauli weight. This is true for instance in the 5-qubit perfect code, which is non-CSS. 

Condition (3) is most easily met by setting $U = \bigotimes_{i=1}^n u_i$ where $\{u_i\}$ are arbitrary single-qubit unitaries. It is easy to see that $u = e^{-i \theta \mathbf{n} \cdot \boldsymbol{\sigma} /2}$ (rotation by $\theta$ around the $\mathbf n$ axis) commutes with the time-reversal symmetry, since $\mathcal{T} \sigma^\alpha = - \sigma^{\alpha}\mathcal{T}$ for $\alpha = x,y,z$, and thus $\mathcal{T} i \mathbf{n}\cdot\boldsymbol{\sigma} = -i \mathcal{T} \mathbf{n}\cdot\boldsymbol{\sigma} = + i \mathbf{n}\cdot\boldsymbol{\sigma} \mathcal{T}$, using anti-unitarity of $\mathcal{T}$. 
More generally, the condition $[U,\mathcal{T}]=0$ identifies a subgroup\footnote{
This subgroup is isomorphic to the compact symplectic group $\textrm{USp}(2^{n-1})$, of dimension $2^{n-1}(2^n+1)$. This is defined as $\{ U\in U(2^n):\, U^T \Omega U = \Omega\}$ with $\Omega$ a standard symplectic form---a real antisymmetric matrix obeying $\Omega^2 = -I$ and $\Omega\Omega^T=\Omega^T\Omega=I$. One can rewrite $[U,\mathcal{T}] = 0$ as $U^T \Omega U = \Omega$ where $\Omega = \prod_{j=1}^n iY_j$ is real antisymmetric and obeys $\Omega^2 = (-1)^n I$ and $\Omega\Omega^T=\Omega^T\Omega=I$. Finally, conditions (1) and (2) imply that the code size $n$ is odd (e.g. the operator $\prod_{j=1}^n X_j$ commutes with all stabilizers and anticommutes with $Z_L$, so it is a nontrivial logical of weight $n$), giving $\Omega^2 = -I$.
} 
of $U(2^n)$, comprising all unitaries $U = e^{iH}$ where $H$ is a real linear combination of odd-weight Paulis. Thus Theorem~\ref{theo:logicalunitary} covers a large class of {\it many-body} coherent errors $U$. As an example one can consider local circuits of three-body gates such as $e^{-i\alpha ZZZ}$, interspersed by arbitrary single-qubit gates. 

We note that unitarity of the logical operation was proven for the special case of single-qubit $Z$ rotations on the odd-distance surface code in Ref.~\cite{Bravyi2018correcting}, and recently generalized to arbitrary single-qubit rotations in Ref.~\cite{darmawan_optimal_2024}. Theorem~\ref{theo:logicalunitary} generalizes the result to a much wider class of codes and coherent errors.

For the rest of the present work we focus on rotated surface codes with odd distance $d$, subject to arbitrary single-qubit physical rotations. The generalization to many-body physical unitaries is an interesting direction whose exploration we leave to future work. 

\begin{figure}
    \centering
    \includegraphics[width=\linewidth]{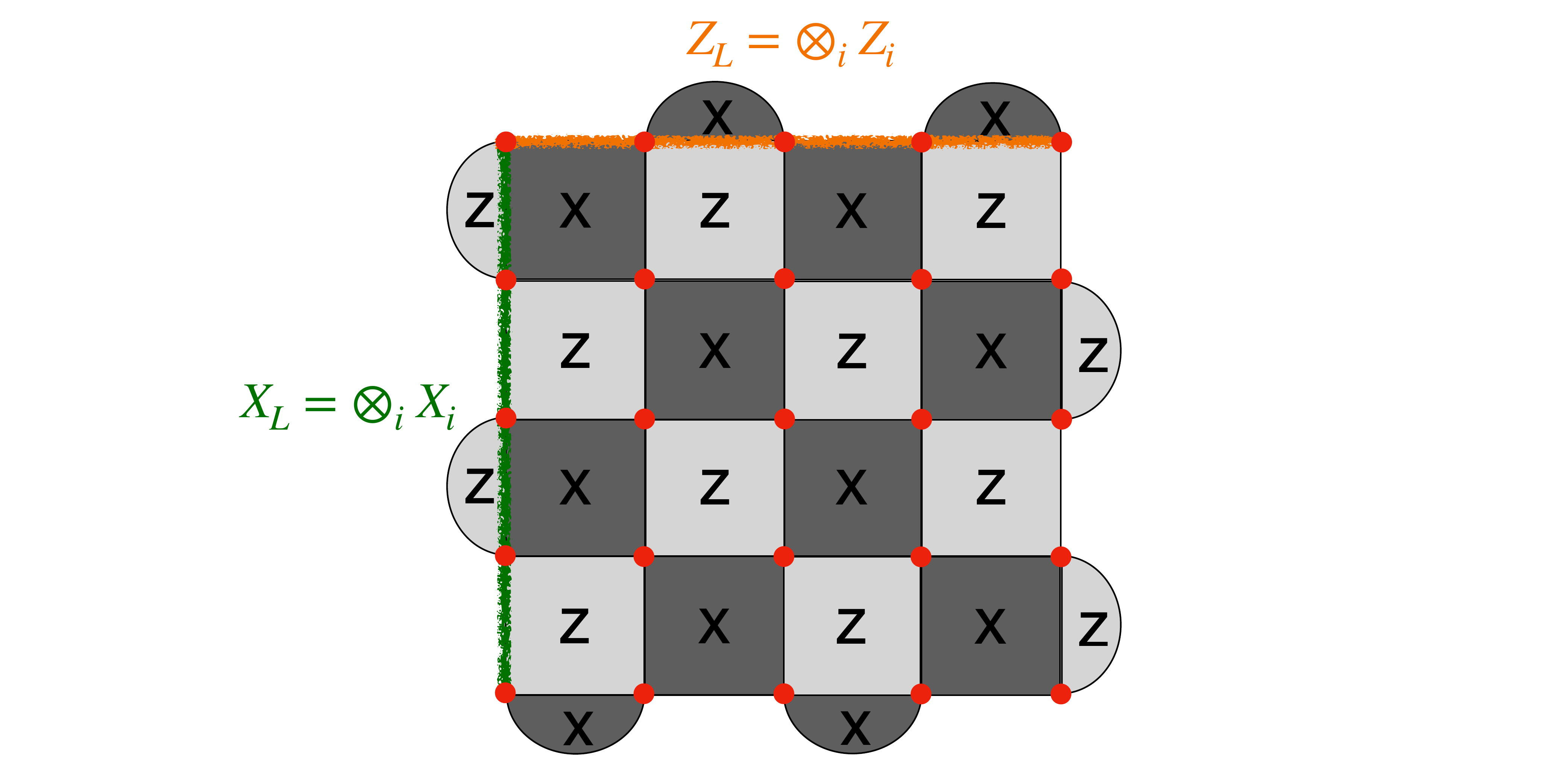}
    \caption{Rotated surface code with distance $d=5$. Physical qubits are denoted by red points. Logical operators $X_L$ and $Z_L$ are denoted by green and orange strings of $X$'s and $Z$'s respectively. We have $|X_L| = |Z_L| = 5$ and $|Y_L| = |i X_L Z_L| = 9$, so all logical operators have odd weight.}
    \label{fig:surfacecode}
\end{figure}


\section{Emergent logical unitary designs}\label{sec:emergent_unitary_design}

Theorem~\ref{theo:logicalunitary} allows the realization of projected ensembles of logical unitary operations 
\begin{align}\label{eq:projected_ensemble}
    \mathcal{E}_{U, L}=\{p(\v{s}), U_{\v{s}}\}
\end{align}
that in principle go beyond $Z$-rotations, thus extending the results of Ref.~\cite{Bravyi2018correcting} (see Eq.~\ref{eq:z_ensemble}). 
In analogy with the phenomenon of ``deep thermalization'' in projected ensembles of states (Sec.~\ref{sec:review_pe}), it is then natural to ask about the emergence of universal behavior in this ensemble, and specifically about the possibility of unitary designs.

When $U$ is weak enough, i.e. close enough to identity, the ``coherent error'' represented by $U$ is correctable, which means all logical unitaries in the projected ensemble approach the logical identity in the limit of large code distance, $d\to\infty$. Therefore logical unitary designs are only possible above the error threshold, i.e. for strong enough $U$. 
In the following we show numerically that unitary designs emerge as soon as coherent errors exceed the optimal coherent error threshold. The optimal error threshold thus also serves as the critical point for a ``unitary design phase transition''. 

Before proceeding, we first remark on the role of the decoder. Similar to error correction, the projected ensemble in Eq.~\ref{eq:projected_ensemble} depends on the choice of decoder, i.e., the assignment of a recovery operation $C_{\v s}$ to each syndrome $\v s \in \{0,1\}^{n-1}$. With different decoders, the logical unitaries may differ by a logical Pauli for some syndrome outcomes. Specifically, if a decoder yields corrections $\{C_{\v s}\}$ and another yields corrections $\{C'_{\v s} = \sigma_{\v s}C_{\v s}\}$ (with $\sigma_{\v s}$ logical Paulis), then the projected ensembles for the two decoders will be
\begin{equation}
    \mathcal{E}_{U,L} = \{p(\v s), U_{L,\v{s}}\}, \quad 
    \mathcal{E}_{U,L}' = \{p(\v s), \sigma_{\v s} U_{L,\v{s}}\}. 
\end{equation}
This ambiguity results in logical Pauli errors, which can induce a logical $1$-design for one decoder but not for another. Therefore, we study the formation of logical unitary ($k\geq2$)-designs separately from $1$-designs. While the former come from the intrinsic randomness of the projected ensemble, the latter are affected by the choice of decoder. 

To tease out the effect of the decoder choice from the intrinsic randomness in the projected ensemble and connect it to error correction, we first consider two extreme decoders---optimal decoder and random decoder. We then consider a practical decoder which is intermediate between the two: the minimum weight perfect matching (MWPM) decoder, which is a standard choice for the surface code~\cite{dennis_topological_2002,fowler_towards_2012}. In our application of the MWPM decoder we use no prior information about the noise model. 
The optimal decoder and random decoder are defined as follows:
\begin{itemize}
    \item The optimal decoder~\cite{Venn2020error-correction}, 
    \begin{align}
        U_{{\rm opt},\v{s}}=\sigma_{{\rm opt},\v{s}}U_{L, \v{s}},
    \end{align}
    with 
    \begin{align}
        \sigma_{{\rm opt},\v{s}}=\operatorname*{arg\,max}_{\sigma \in \{I, X, Y, Z\}}|\Tr(\sigma U_{L,\v{s}})|^2,
    \end{align}
    where $U_{L,\v{s}}$ is the logical unitary obtained from any other decoder, for example MWPM. This choice is relevant to questions about error correction, but not about unitary designs. Indeed the projected ensemble constructed in this way cannot cover all of $U(2)$, for a trivial reason: the optimal decoder always returns unitaries that are closer to $I$ than to $X$, $Y$, or $Z$; such unitaries make up exactly one quarter of $U(2)$ (see sketch in Fig.~\ref{fig:main_idea}(b), top). This prevents the formation of 1-designs, and thus $k$-designs for any $k$. 
    \item A random decoder,
    \begin{align}
        U_{{\rm rand},\v{s}}= \sigma_{{\rm rand},\v{s}} U_{L, \v{s}},
    \end{align}
    where $U_{L, \v{s}}$ is the unitary obtained from any other decoder (for example MWPM or optimal) and $\sigma_{{\rm rand},\v s}$ is a logical Pauli unitary chosen uniformly from $\{I,X,Y,Z\}$, independently and identically for each syndrome $\v s$. 
    This choice is relevant to questions about unitary designs. It is also natural for applications since logical Paulis can always be implemented transversally. With this ``Pauli scrambling'' operation, the projected ensemble forms a 1-design with high probability, while formation of $k$-designs with $k\geq 2$ relies on the intrinsic randomness of syndrome outcomes.
\end{itemize}

\subsection{Unitary design phase transition}\label{ssec:design_transition}

We now investigate convergence of the projected ensemble to a unitary $k$-design numerically for $k=2,3$, using the design distance measure $\Delta^{(k)}$ (see Sec.~\ref{sec:review_designs}). Specifically, we consider odd-distance rotated surface codes, see Fig.~\ref{fig:surfacecode}, with the simplest choice of $U$ that meets the requirements of Theorem~\ref{theo:logicalunitary}: $U=\prod_j u_j$ with each $u_j$ a single-qubit gate acting on qubit $j$. 
To tune the strength of $U$ across the error-correction threshold, we sample each $u_j$ from the Haar measure on $U(2)$ with probability $p$, and we set $u_j = \mathbb{I}$ otherwise. We denote the resulting measure over the many-body unitary group $U(2^n)$ by $\mu_p$. The parameter $p\in[0,1]$ is what tunes the strength of $U$.

The polynomial-time classical algorithm in Ref.~\cite{Bravyi2018correcting}, based on a free Majorana fermion representation, relies on having only $Z$ rotations on the physical qubits, and fails for generic $u_j$ rotations. Here we propose a tensor network based algorithm that can substantially extend the accessible system sizes compared to other existing methods, allowing us to study the unitary design phase transition numerically. We will elaborate on the algorithm and investigate its computational complexity in Section~\ref{sec:staircase}.

In Fig.~\ref{fig:design_transition}(a-d), we show a finite-size scaling analysis of the unitary $k$-design distance for $k=2,3, 4$ 
\begin{align}
    \Delta^{(k)}
    & =\underset{U \sim \mu_p }{\mathbb{E}}\left\Vert\sum_{\v{s}}p(\v{s}) 
    \mathcal{U}_{{\rm dec},\v{s}}^{\otimes k} 
    -\Phi^{(k)}_{\rm Haar}\right\Vert_\Diamond, \\
    \mathcal{U}_{{\rm dec},\v{s}}(\cdot) 
    & = U_{{\rm dec},\v{s}} [\cdot] U_{{\rm dec},\v{s}}^\dagger,
    \label{eq:designdist_diamond}
\end{align}
where `dec' stands for the choice of decoder; here we use the random decoder.
Data is shown as a function of $p$ (probability of having a nontrivial physical unitary on each qubit), with varying code distance $d$ from $7$ to $15$ (always odd). We compute the diamond norm by using the Qiskit package~\cite{qiskit2024}, where the diamond norm is formulated as a semidefinite program based on Ref.~\cite{watrous_simplersemidefiniteprogramscompletely_2012}. 
The existence of finite-size crossing points for $\Delta^{(k)}$ indicates a unitary design phase transition in the projected ensemble, controlled by $p$: 
for $p>p_{c}$, $\Delta^{(k)}$ decreases in the code distance $d$, while for $p<p_{c}$, $\Delta^{(k)}$ increases in the code distance $d$.  
The former behavior denotes convergence to a unitary $k$-design in the thermodynamic limit $d\to\infty$. (See Appendix~\ref{ap:distribution} for additional data on the distribution of logical unitaries in the two phases, which supports convergence toward the Haar distribution in the design phase.) 
We refer to the $p > p_{c}$ regime as a ``unitary design phase''. 

We use the scaling ansatz
\begin{align}\label{eq:ansatz}
    \Delta^{(k)}(p, d)=F_{k}\[(p-p_{c})d^{1/\nu}\],
\end{align}
to determine the location the critical point $p_c$ and correlation length critical exponent $\nu$. From the data collapse in Fig.~\ref{fig:design_transition}(b,d,f), we locate $p_{c}\approx0.87$ and $\nu\approx1.6$ for all $k=2$, $k=3$ and $k=4$, consistent with a single unitary design phase transition for different moments $k$. We conjecture that the same critical point should control the convergence of all moments $k$, thus giving the Haar distribution for $d\to\infty$ when $p > p_{c}$. We note that the cardinality of the projected ensemble ($2^{n-1}$, the number of possible syndrome outcomes) allows the formation of unitary designs up to moments $k \leq 2^{n/2}$~\cite{Brandao2016logical}.  

\begin{figure}
    \centering
    \includegraphics[width=1.0\linewidth]{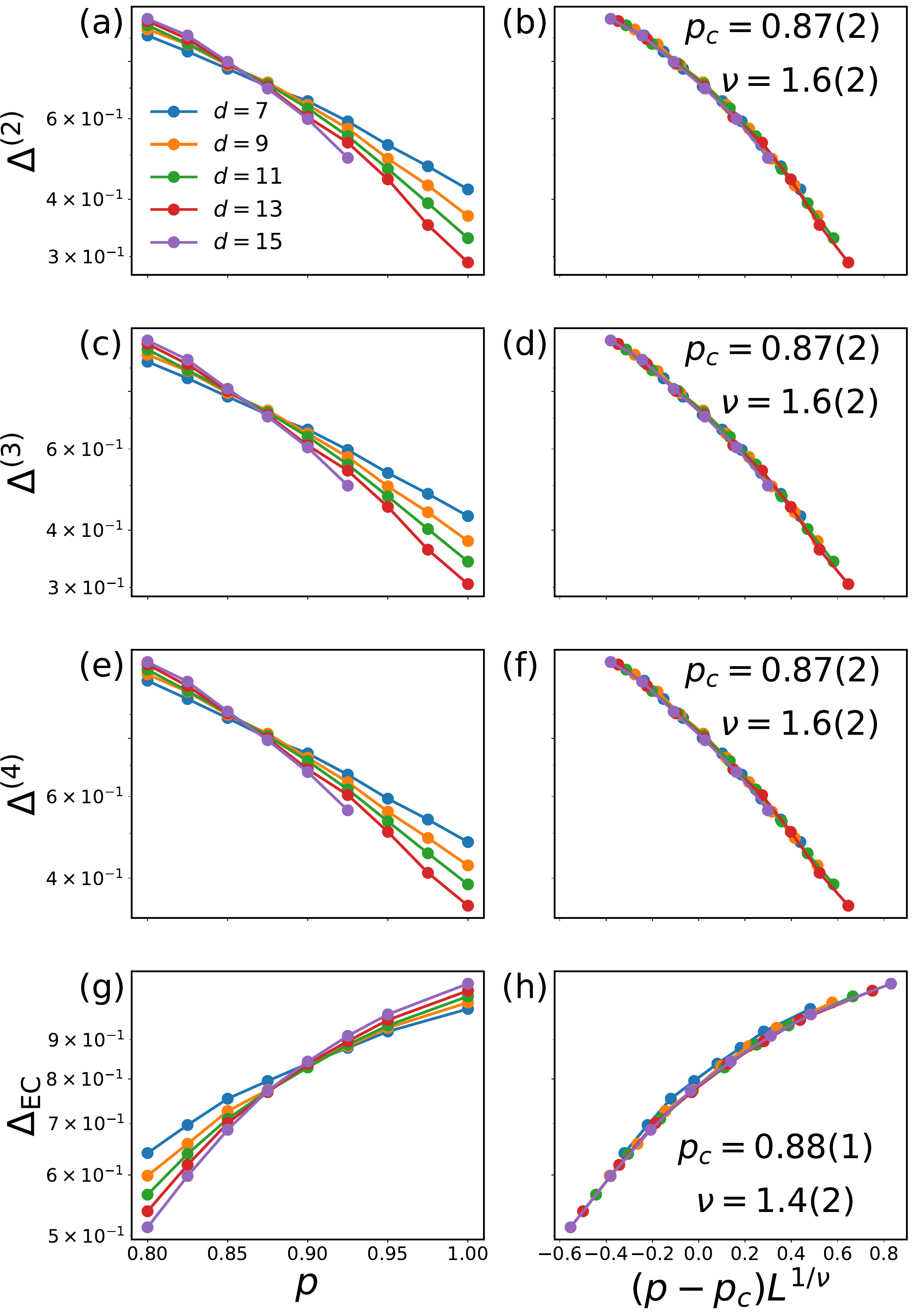}
    \caption{(a)(c)(e) Unitary $k$-design distance $\Delta^{(k)}$ for $k=2, 3, 4$ with the random decoder and (g) logical error rate $\Delta_{\rm EC}$  with the optimal decoder as a function of the probability of applying random single-qubit unitaries $p$. (b)(d)(f)(h) Corresponding scaling collapse of the data. The data are averaged over $384-1280$ realizations of random single-qubit unitaries.}
    \label{fig:design_transition}
\end{figure}

\subsection{Error correction threshold}
Using the same simulation method, we can also investigate the optimal error correction threshold via the logical error rate defined as
\begin{align}
    \Delta_{\rm EC} = \mathop{\mathbb{E}}_{U\sim \mu_p} \[\sum_{\v{s}}p(\v{s})\left\Vert \mathcal{U}_{{\rm opt},\v{s}}-\mathcal{I} \right\Vert_\Diamond\]
\end{align}
where $\mathcal{I}$ is the identity channel, and $\mu_p$ is the measure over the unitary group $U(2^n)$ induced by our prescription for sampling the physical gates $\{u_j\}$ ($u_j$ Haar-random on $U(2)$ with probability $p$, $u_j=\mathbb{I}$ otherwise). Note the difference with the design distance measures $\Delta^{(k)}$ [Eq.~\eqref{eq:designdist_diamond}]: here the diamond norm is inside the syndrome average, so $\Delta_{\rm EC}$ is small if and only if each logical gate $U_{L,\v s}$ is close to the identity.

The numerical results are presented in Fig.~\ref{fig:design_transition}(e,f). The data suggest the logical error rate $\Delta_{\rm EC}$ decreases in the code distance $d$ when $p<p_{c,\mathrm{EC}}$, where $p_{c,\mathrm{EC}}$ can be viewed as the coherent error threshold.  Similar to the unitary design, we locate the critical point $p_{c,\mathrm{EC}}$ and
correlation length exponent $\nu$ by finding a data collapse on the scaling ansatz
\begin{align}
    \Delta_{\rm EC}(p, d)=F\[(p-p_{c, \mathrm{EC}})d^{1/\nu}\].
\end{align}
The result is $p_{c,\mathrm{EC}}= 0.88(1)$ and $\nu=1.4(2)$, both of which fall in the tolerance interval of the results for the unitary design phase transition. 
This provides numerical evidence of a unitary design phase transition for $k\geq2$ that coincides with the optimal error correction threshold. 

As we argued above, the optimal error correction threshold poses a lower bound for the formation of unitary designs. Assuming that a unitary design phase exists, one could {\it a priori} have two distinct thresholds $p_{c,\rm EC} < p_{c}$ with an intervening `non-EC, non-design' phase. Our results indicate that this is not the case. In particular they suggest that the projected ensemble can flow toward two distinct fixed points as $d\to\infty$: the random Pauli ensemble\footnote{
This assumes a random decoder. With a specific decoder like MWPM we will have a trivial `identity' ensemble $\mathcal{E}_{\mathbb I} = \{1,\mathbb{I}\}$ below a decoder-dependent threshold $p < p_{c,{\rm dec}}$, as shown in Fig.~\ref{fig:main_idea}(b). 
} 
$\mathcal{E}_\text{Pauli} = \{1/4,\mathbb{I}; 1/4,X; 1/4,Y; 1/4,Z\}$ when $p<p_c$, and the Haar random ensemble $\mathcal{E}_{\rm Haar}$ when $p>p_c$. 

\subsection{Phase transitions under the MWPM decoder}

\begin{figure}
    \centering
    \includegraphics[width=1.0\linewidth]{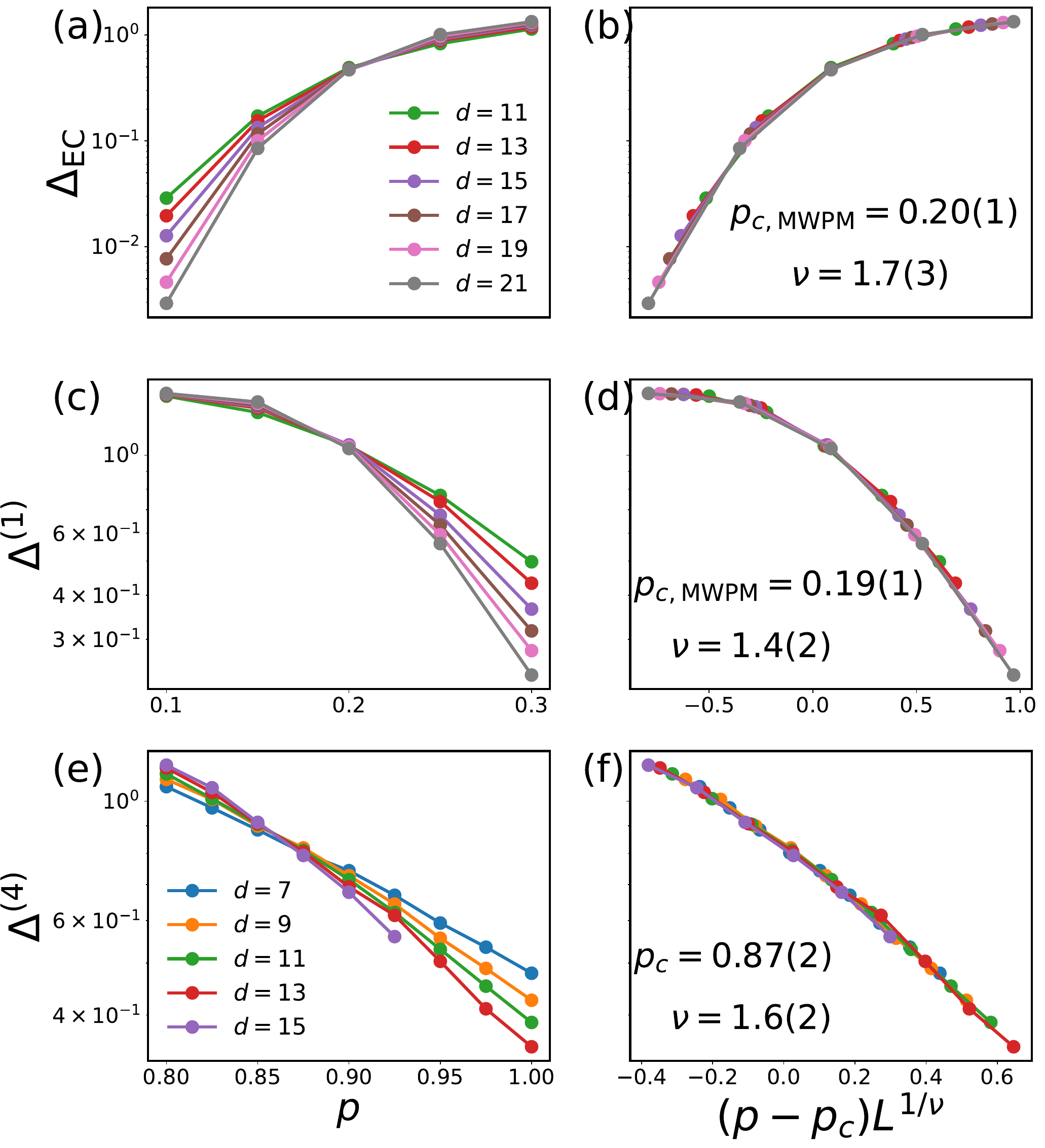}
    \caption{(a) (b) Logical error rate $\Delta_{\rm EC}$, (c)(d) $1$-design distance $\Delta^{(1)}$ and (e)(f) $4$-design distance $\Delta^{(4)}$  under the MWPM decoder as a function of the probability of applying random single-qubit unitaries $p$.  Scaling collapse of the data. The data are averaged over $384-1280$ realizations of random single-qubit unitaries.}
    \label{fig:mwmd}
\end{figure}

Finally we study the phase transitions in error correction and the formation of unitary $1$-design under a specific decoder (neither random nor optimal), such as MWPM. For MWPM decoder, one incurs logical Pauli errors due to the error degeneracy~\cite{stace2010error}, i.e. the ambiguity in the choice of matching. In this case, since the extra logical Pauli errors can form a logical 1-design (but not a 2-design or higher), the error correction threshold and $1$-design transition will be separated from the $(k\geq2)$-design transition, with $p_{c, \mathrm{MWPM}}\approx0.20$, much smaller than the $(k\geq 2)$-design threshold $p_c\approx0.87$, as shown in Fig.~\ref{fig:mwmd}. The wide discrepancy in thresholds between the MWPM decoder and the optimal decoder arises due to the different information available to the two decoders: our MWPM uses only syndrome data and no prior information on the noise, while the optimal decoder uses full knowledge of the configuration of coherent errors. 

\section{Staircase circuit mapping and computational phase transition}\label{sec:staircase}

In this section we present a classical algorithm to decode the surface code under general single-qubit coherent errors, which we used to simulate the projected ensemble of logical unitaries in Sec.~\ref{sec:emergent_unitary_design}. Notably, our algorithm undergoes a complexity phase transition at the same value of $p_c$ identified above as the optimal error correction and unitary design threshold, with efficient simulation below threshold and inefficient simulation above threshold. This complexity transition is due to a measurement induced entanglement transition in an effective $(1+1)$-dimensional monitored dynamics that appears in the algorithm (see Sec.~\ref{sec:review_sebd}). 
We stress that our algorithm is of independent interest and could be applied more broadly. It may also be possible to generalize it to the case of incoherent errors~\cite{behrends_surfacecode_2025}, which is an interesting direction to pursue in future work.

\begin{figure}
    \centering
    \includegraphics[width=1.0\linewidth]{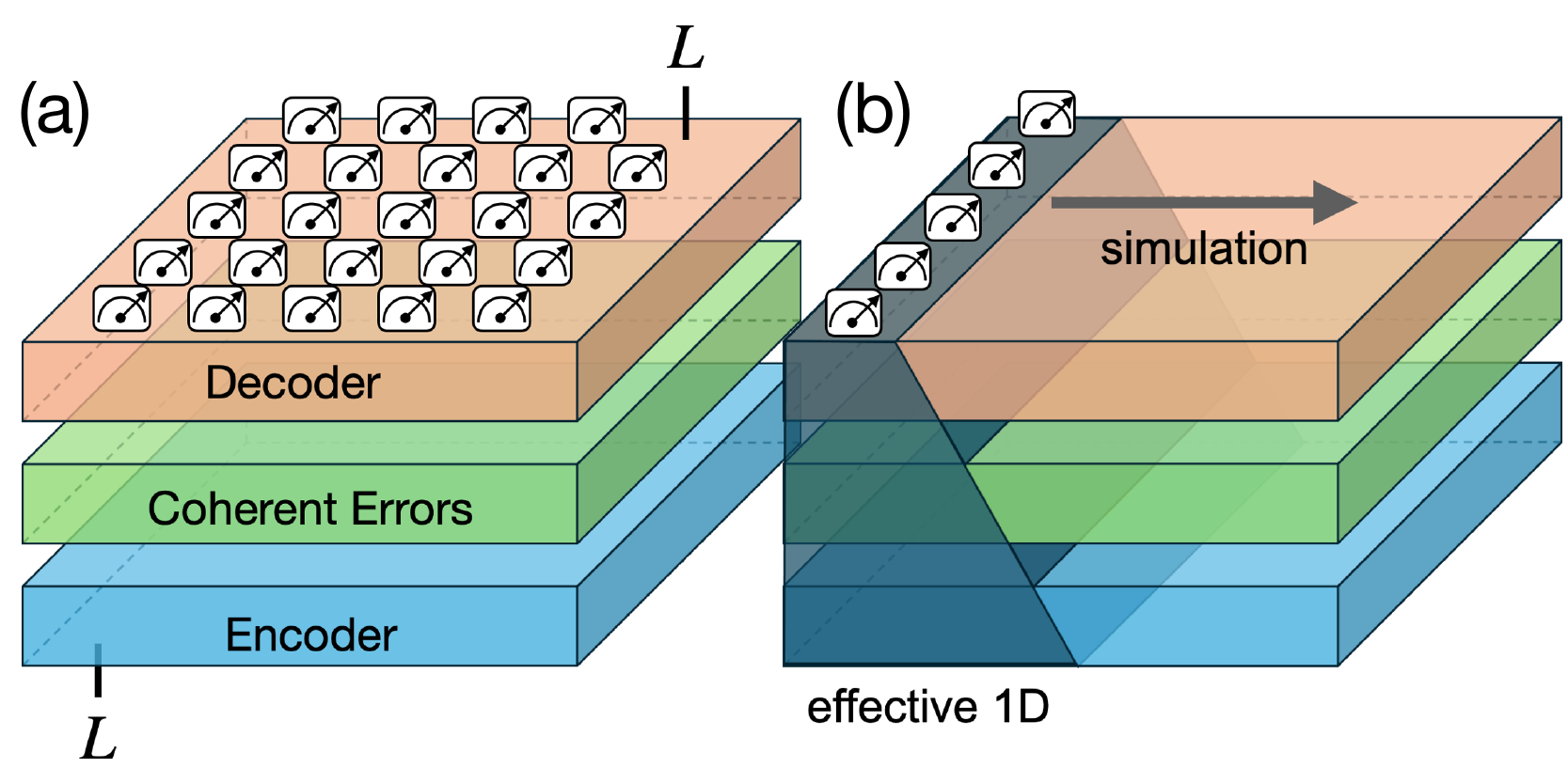}
    \caption{(a) Schematic of the encoding-decoding circuits with coherent errors. (b) The staircase representation enables the mapping of 2D circuits to effective 1D circuits along the spatial direction, similar to the space-evolving block decimation~\cite{Napp2022efficient,Cheng2023efficient}.}
    \label{fig:staircase}
\end{figure}

Ref.~\cite{suzuki_efficient_2017,Bravyi2018correcting} introduced an efficient simulation algorithm based on a mapping of the surface code with coherent $Z$-errors to a system of free Majorana fermions. This is key to the ability to efficiently decode the problem and thus successfully correct errors below threshold. In Sec.~\ref{sec:emergent_unitary_design} we have generalized the construction of unitary projected ensembles beyond $Z$ errors, allowing for arbitrary on-site unitaries as well as a wide class of many-body unitaries. In such cases the mapping to free Majorana fermions breaks down, and a new approach is needed.

Our algorithm is based on the ``staircase'' circuit representation of the encoder (decoder) of surface codes. Surface code states possess long-range (topological) entanglement, and so cannot be prepared by shallow local unitary circuits. Staircase circuits sequentially apply gates across the system, building up the long-range entanglement step by step, and achieve the optimal linear circuit depth lower bound~\cite{Satzinger2021realizing,Higgott2021optimallocalunitary}.
While in practice measurement-based circuits may have an advantage for preparing surface code states~\cite{iqbal_topological_2024,foss-feig_experimental_2023}, at the level of classical simulation it will prove useful to focus on staircase encoder circuits. Similarly, staircase decoder circuits can be used to turn the two- and four-qubit measurements of surface code stabilizers into single-qubit measurements. 
The sequential nature of staircase encoder/decoder circuits, combined with the repeated measurement of individual qubits, allows us to effectively trade a space dimension for a time dimension. This turns the 2D problem with final measurements into a $(1+1)$D dynamics with mid-circuit measurements, Fig.~\ref{fig:staircase}. This general strategy is briefly reviewed in Sec.~\ref{sec:review_sebd} and is the basis of a simulation algorithm called `space-evolving block decimation' (SEBD)~\cite{Napp2022efficient,Cheng2023efficient}.

\subsection{Warm-up: repetition code}

Before diving into the algorithm for surface codes, we first consider the repetition code as a toy example. 
The encoder circuit $E$ of the repetition code is a sequence of CNOT gates arranged in a staircase pattern. 
The first input qubit contains the state to be encoded, $\ket{\psi}$, while the other $n-1$ qubtis store a syndrome string $\v s\in \{0,1\}^{n-1}$; $E$ acts on these input states as
\begin{equation}
    E [(\alpha \ket{0} + \beta\ket{1}) \otimes \ket{\v s}] = \alpha \ket{0,\mathbf{s}}+ \beta \ket{1,\bar{\mathbf s}}
\end{equation}
where $\bar{\mathbf s}$ is the bit-wise negation of bitstring $\mathbf s$. In words, the output state is a ``Schr\"odinger's cat'' of two bitstrings with bond checks $g_i = Z_i Z_{i+1} = (-1)^{s_i}$. The trivial syndrome $\v s = \v 0$ yields repetition code logical states, $\alpha \ket{0}^{\otimes n} + \beta\ket{1}^{\otimes 1} = \ket{\psi_L}$.
It follows that, when we eventually want to measure the parity checks $g_i = Z_i Z_{i+1}$ after acting with a physical unitary $U$, we may instead apply the decoder circuit $E^\dagger$ followed by computational basis measurements on qubits $i=2,\dots n$, which after application of $E^\dagger$ store the syndrome bits. In all, the protocol is given by $(\mathbb{I} \otimes \ketbra{\v s}) E^\dagger U E (\ket{\psi}\otimes\ket{0}^{\otimes n-1})$, where $\mathbb{I} \otimes \ketbra{\v s}$ projects the syndrome qubits $i=2,\dots n$ onto the syndrome outcome and acts trivially on the logical qubit $i=1$. This describes a left-to-right sweep of CNOT gates ($E$), followed by single-qubit gates ($U$) and then a right-to-left sweep of CNOT gates ($E^\dagger$) and measurements.

It is advantageous to write the whole protocol as a staircase circuit. We can do so by choosing a different decoder circuit, $D = \tilde{E}^\dagger$, where $\tilde{E}$ is simply the mirror reflection of $E$. This way both $E$ and $D$ proceed left-to-right. The whole protocol $(\ketbra{\v s}\otimes I) D U E (\ket{\psi}\otimes \ket{\v 0})$ is shown diagrammatically in the top left panel of Fig.~\ref{fig:encoder}. It takes in a logical input state in the leftmost qubit, $i=1$, and returns a logical output state on the rightmost qubit, $i=n$. The tensor network connecting input to output may be seen as the iteration of a unit cell made of only three qubits. At the end of each unit cell, a qubit is measured out and a new auxiliary qubit in the $\ket{0}$ state is introduced, keeping the number of qubits constant and finite. This way we have effectively turned the simulation of a 1D repetition code onto a $(0+1)$D monitored evolution. 

\subsection{Surface code}

\begin{figure*}
    \centering
    \includegraphics[width=1.0\linewidth]{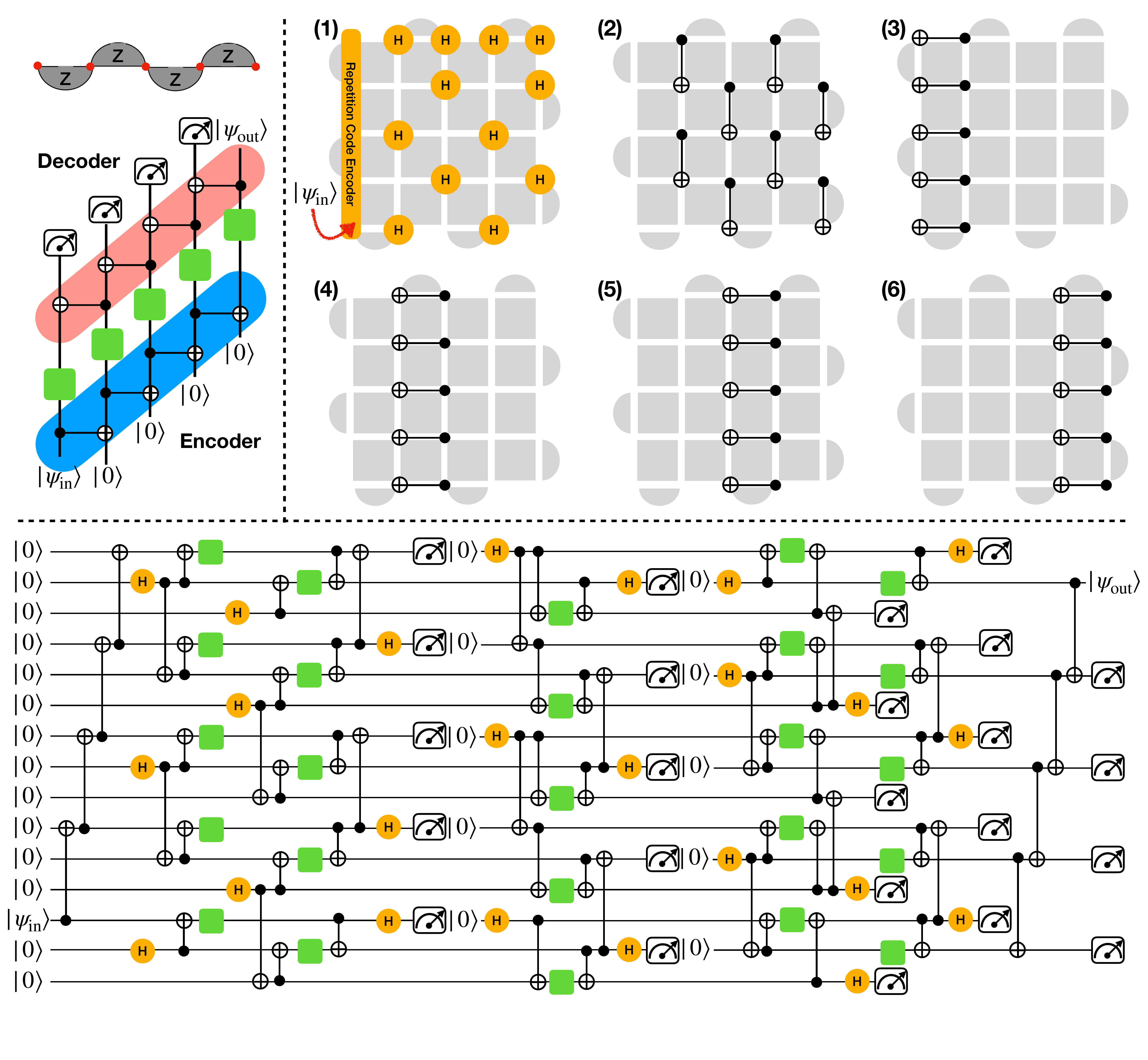}
    \caption{Top left panel: repetition code with $d=5$ and its staircase encoding(blue)-decoding(red) circuit with on-site coherent errors(green). Top right panel: steps of the staircase encoder for $d=5$ surface code. Step (1) includes a depth-$d$ staircase circuit (repetition code encoder) while all the other steps have depth $1$. Steps (3)--(6) constitute a staircase of CNOT gates along the horizontal direction, which has total depth $d-1$. Bottom panel: effective (1+1)D monitored circuit on $3\times5$ qubits for the encoding-decoding circuit of the $d=5$ surface code.}
    \label{fig:encoder}
\end{figure*}

The surface code, as a hypergraph product of two repetition codes~\cite{Tillich2014quantum}, also features a staircase structure. In the top right panel of Fig.~\ref{fig:encoder}, we explicitly show a circuit to encode a $d=5$ surface code with a staircase of CNOT gates along the horizontal direction. The first step is to encode the input qubit (bottom of the top left panel in Fig.~\ref{fig:encoder}) into a 1D repetition code along an edge, with a 1D staircase circuit that sweeps the vertical direction, bottom-to-top in this case. Then a 2D staircase circuit sweeping the system left-to-right turns this 1D repetition code logical state into a 2D surface code logical state. The circuit as a whole has linear depth. 
In analogy to the repetition code, we can choose a decoder circuit $D = \tilde{E}^\dagger$ where $\tilde{E}$ is obtained from $E$ by a spatial inversion, so that the total encoding-decoding circuit takes a staircase form.

Having recast the whole protocol as a staircase circuit followed by one-qubit measurements, we can now view it as a $(1+1)$D monitored evolution, in analogy with the repetition code but with an additional transverse spatial dimension. This is a version of a `holographic' or `space-time dual' mapping, which leverage the equivalence between space and time in the presence of measurements~\cite{foss-feig_holographic_2021,ippoliti_fractal_2022,lu_spacetime_2021,hoke_measurement-induced_2023,anand_holographic_2023} to trade a space dimension for a time dimension (see Fig.~\ref{fig:staircase}(b)). The idea is similar to the space-evolving block decimation method for sampling 2D shallow circuits~\cite{Napp2022efficient,Cheng2023efficient}, see Sec.~\ref{sec:review_sebd}.

As a result of this mapping we have an effective 1D description of the problem, which paves the way for matrix product state (MPS) algorithms. The relevant MPS has $N=3d$ qubits (corresponding to a $3\times d$ strip of qubits, with the factor of $3$ coming from the geometry of the past lightcone in Fig.~\ref{fig:staircase}(b)). See bottom panel of Fig.~\ref{fig:encoder} for an explicit example for $d=5$ surface code. Notably, the final (syndrome) measurements are converted to midcircuit measurements and resets in the dynamics of the MPS. This is especially interesting since mid-circuit measurements can lower the state's entanglement entropy and even drive a transition to an area-law-entangled phase (see Sec.~\ref{sec:review_sebd}). The accuracy of the MPS representation relies on the entanglement area law~\cite{Schuch2008entropy}. 

Similar to the random circuit sampling problem studied in Ref.~\cite{Napp2022efficient,Cheng2023efficient}, the unitary gates comprising encoder $E$, ``coherent errors'' $U$, and decoder $D$ are competing with the syndrome measurements; this raises the possibility of an entanglement phase transition. Since the number of measurements is fixed by the circuit's geometry, the competition is driven by the strength of coherent errors---the same parameter that, as we saw, drives the unitary design and error correction transition. It is natural to conjecture an entanglement phase transition in this problem, and thus a complexity transition for the associated classical algorithm, at the same threshold $p_c$ of coherent errors.
In the following we investigate this question both numerically and analytically.

\subsection{Entanglement phase transition \label{sec:mipt}}

In this section we numerically study the entanglement phase transition in the effective 1D system that is used in the above algorithm (Fig.~\ref{fig:encoder}, bottom). The model we consider is the same as in Section~\ref{ssec:design_transition}, where coherent errors are on-site Haar random unitaries with density $p$. Due to the structured geometry of this circuit, the tripartite mutual information (the clearest and most direct diagnostic of the entanglement phase transition in one dimension~\cite{Zabalo2020critical}) suffers from large finite-size effects. 
Therefore, we resort to the coding perspective on the measurement-induced entanglement phase transition~\cite{Gullans2020dynamical,Gullans2020scalable}: we ask whether encoded information can be destroyed by measurements. Concretely one can consider the behavior of a reference qubit $R$ initially entangled with the monitored system~\cite{Gullans2020scalable}. Let $S_R(t)$ be the entanglement entropy of the reference qubit as a function time $t$ (in our case in the depth of the monitored circuit). At late time one generally has $S_R(t)\sim \exp({-t/\tau})$, with a purification time $\tau$. The behavior of $\tau$ changes sharply at the transition: from $\tau\sim\exp(n)$ in the volume-law phase to $\tau\sim O(1)$ in the area-law phase. At the critical point, $\tau$ diverges algebraically as $\tau\sim L^{z}$ with $L$ the linear size of the system and $z$ a dynamical critical exponent, which is usually $1$~\cite{Zabalo2020critical}.
Therefore, in monitored random circuits, one can use $\tau/L$ as a diagnostic for the entanglement phase transition.

However, the general picture we just reviewed becomes more subtle in our case. This is due to the fact that our monitored circuits have a special structure, being formed by encoder and decoder circuits of the surface code: the initial state naturally factors into a logical part and a syndrome part. 
If the reference qubit is entangled with the code's logical qubit, its entanglement with the system will be persistently maximal (one bit), since based on Theorem~\ref{theo:logicalunitary}, the syndrome measurements result in logical unitaries and thus never read out the logical information.

To circumvent this issue, we instead consider entangling the reference qubit to logical degrees of freedom, then running the monitored evolution over time and eventually measuring out the logical qubit while leaving some syndrome degrees of freedom unmeasured. 
This setup can be naturally constructed in the MPS representation of the encoding-decoding circuit. We encode a $(d_x+1)\times d_y$ surface code\footnote{Here $d_x$, $d_y$ refer to the linear dimensions of the surface code.}, then treat the first $d_x\times d_y$ qubits (leaving out the last column) as if they formed a surface code patch. We apply coherent errors on this $d_x\times d_y$ surface code, decode it, and measure out the corresponding logical qubit (See Fig.~\ref{fig:GHZ}(a)). The remaining column of $d_y$ qubits retains some syndrome information from the larger [$(d_x+1)\times d_y$] surface code, and thus can remain entangled with the reference qubit (See Fig.~\ref{fig:staircase}(b)). 
This circuit is closely related to the first $d_x$ time steps in the $(1+1)$D monitored circuit. It differs only by the application of a repetition code decoder and a logical measurement; these are artificial extra operations that do not alter the phase transition (See Appendix~\ref{ap:clifford}) but have the advantage of offering an analytical understanding of the connection between unitary design phase transition and entanglement phase transition (See Section~\ref{sec:connection}). 

\begin{figure}
    \centering
    \includegraphics[width=1.0\linewidth]{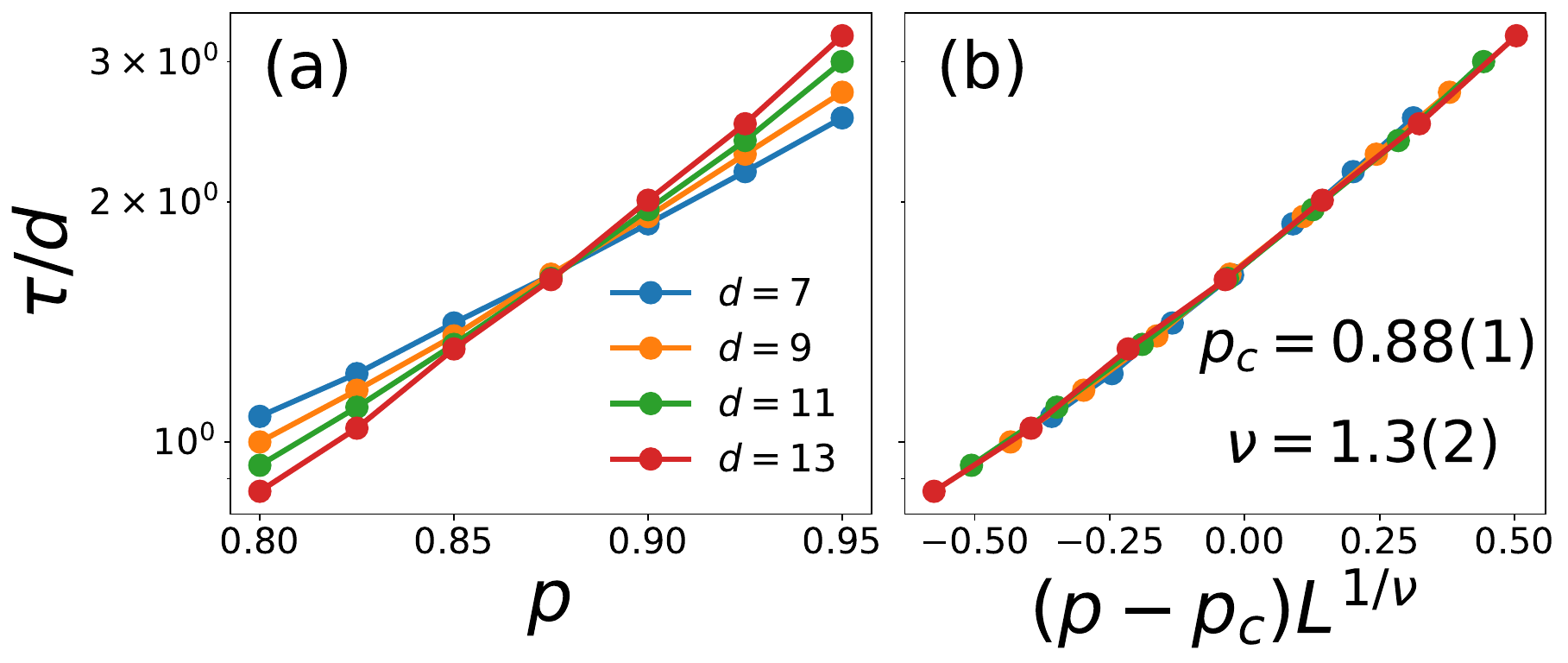}
    \caption{(a) Purification time scale $\tau$  as a function of the probability $p$ of random single-qubit physical unitaries. (b) Scaling collapse of the data. The data are averaged over $100-1000$ realizations of both the random coherent errors and syndrome sampling.}
    \label{fig:purification}
\end{figure}

This construction allows us to study the purification of the reference qubit as we sweep across the staircase along one direction---here we choose the horizontal ($x$) direction. Setting $d_y=d$ and viewing $d_x$ as a variable time scale, the averaged entropy of the reference qubit takes the form $S_R(d_x)\sim\exp(-d_x/\tau)$.

In Fig.~\ref{fig:purification}, we show a finite-size scaling analysis of $\tau/d$ for model considered in Section~\ref{sec:emergent_unitary_design}, obtained from the MPS simulation of the spacewise dynamics up to $d_x=3d$ (only $d\leq d_x\leq 3d$ are used in the fit to avoid early-time transient effects). The existence of a finite-size crossing point of $\tau/d$ in Fig.~\ref{fig:purification}(a) indicates $z=1$ ($\tau \propto d^1$ at criticality). Therefore, we can locate the critical point $p_c$ and correlation length exponent length $\nu$ by using a similar scaling ansatz in Eq.~\ref{eq:ansatz}. From the data collapse in Fig.~\ref{fig:purification}(b), we have $p_c=0.88(1)$ and $\nu=1.3(2)$, which coincide with the critical point of the unitary design phase transition within tolerance. This indicates the a deeper connection between the unitary design phase transition and the entanglement phase transition.

In Appendix~\ref{ap:clifford} we numerically verify that the purification transition of $S_R$ coincides with the entanglement transition in the 1D MPS in Clifford circuits (replacing Haar random single qubit coherent errors by random Clifford coherent errors). 
There, by using stabilizer simulation, we can compute the tripartite mutual information $I_3$ in large codes ($d> 100$) and verify that the purification of a reference qubit does signal the entanglement phase transition.

\subsection{Connection between unitary design phase transition and entanglement phase transition}\label{sec:connection}

\begin{figure}
    \centering
    \includegraphics[width=1.0\linewidth]{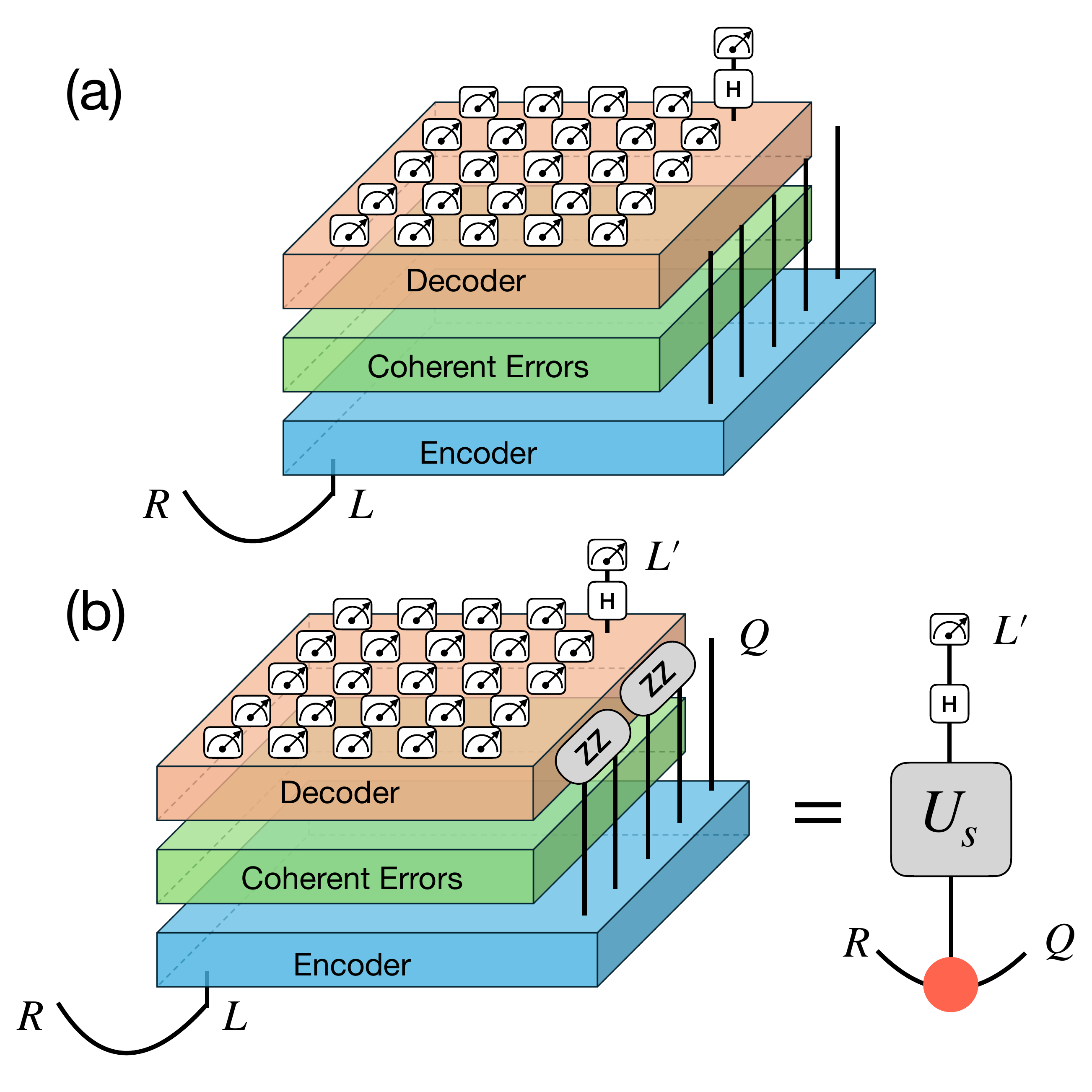}
    \caption{(a) Schematic of the circuit used in Sec.~\ref{sec:mipt} to study dynamical purification of a reference qubit. The reference qubit $R$ is entangled with the logical qubit $L$ of a $(d_x+1)\times d_y$ surface code, then coherent errors and syndrome measurements are applied only to a $d_x\times d_y$ surface code patch, leaving out a column of qubits on the right. We also measure the logical qubit $L'$ of the $d_x\times d_y$ surface code. 
    (b) Mapping of the circuit to a three-qubit GHZ state. Pairwise $ZZ$ measurements on the qubits in the rightmost column leave out a single qubit $Q$. $R$, $L'$ and $Q$ form a GHZ state. The coherent errors and syndrome measurements result in a unitary $U_{L,\v s}$ on $L'$ prior to its measurement. 
    }
    \label{fig:GHZ}
\end{figure}

So far we have made a numerical observation that the entanglement phase transition in our $(1+1)$D monitored dynamics associated to the surface code simulation has the same critical point $p_c$ as the optimal error correction and unitary design thresholds uncovered in Sec.~\ref{sec:emergent_unitary_design}. We now derive an analytical understanding of this connection. 

The key technical observation is a mapping of the many-body circuit of Fig.~\ref{fig:GHZ}(a) onto a three-qubit one, where the first qubit represents the reference $R$, the second is the logical qubit $L'$ (which is being measured in our setup), and the third, $Q$, represents the unmeasured syndrome degrees of freedom. The three qubits are in an $X$-basis GHZ state, $\ket{+}^{\otimes 3} + \ket{-}^{\otimes 3}$, see Fig.~\ref{fig:GHZ}(b); the $L'$ qubit is being acted upon by the logical gate $U_{L,\v{s}}$ and then measured in the $X$ basis. 
Clearly if the gate is trivial, $U_{L,\v s}\approx \mathbb{I}$, then the GHZ state disentangles into $\ket{\pm }_R\ket{\pm}_Q$ after the measurement, and thus the reference purifies. On the other hand if the gate is far from the identity, it effectively rotates the measurement away from the $X$ basis and allows for some entanglement between $R$ and $Q$ to survive. This draws a precise connection between the projected ensemble $\mathcal{E}_{U,L} = \{U_{L,\v s}\}$ and the entanglement of the reference qubit, and draws an analytical connection between the two transitions.

Next we explain the mapping in detail. We return to the setup discussed in Sec.~\ref{sec:mipt}, where we have a reference qubit entangled with the syndrome degrees of freedom of a $(d_x+1)\times d_y$ surface code, and we aim to separate that into a $d_x\times d_y$ surface code patch and a $1\times d_y$ repetition code. 
It can be shown (see Appendix~\ref{ap:decomposition}) that, by performing two-body $ZZ$ measurements on all but one of the qubits in the last column (See Fig.~\ref{fig:GHZ}(b), recall $d_y$ is odd), one indeed obtains a $d_x\times d_y$ surface code and a $d_y$-repetition code. Calling the repetition code logical qubit $Q$ and the $d_x\times d_y$ surface code logical qubit $L'$, one can show that the reference $R$ forms a GHZ state with $Q$ and $L'$ at this stage: $\frac{1}{\sqrt 2} (\ket{+}_R \ket{+}_{L'} \ket{+}_Q + \ket{-}_R \ket{-}_{L'} \ket{-}_Q)$.  
At this point we apply coherent errors to the $d_x\times d_y$ surface code, followed by decoding and measurement of the syndrome qubits; this acts on logical qubit $L'$ as a unitary drawn from the appropriate projected ensemble $\mathcal{E}_{L,U}$. 
Finally, we measure out $L'$ in the $X$ basis. This concludes the mapping illustrated in Fig.~\ref{fig:GHZ}(b).

We note that, although the reference qubit's entanglement in Fig.~\ref{fig:GHZ}(a) will generally be altered by the $ZZ$ measurements in Fig.~\ref{fig:GHZ}(b), the qualitative behavior across the phase transition remains the same.

From Section~\ref{sec:emergent_unitary_design}, we have the projected ensemble of logical unitaries flows to the random Pauli ensemble $\mathcal{E}_\text{Pauli}$ (or the identity ensemble $\mathcal{E}_{\mathbb I}$ depending on the decoder) in the thermodynamic limit $d\rightarrow\infty$ for $p<p_c$. In this case the logical qubit $L'$ is being measured in the $X$ basis, regardless of the possible applied Pauli unitary; this measurement disentangles the $X$-basis GHZ state and thus purifies the reference qubit.  
When $p>p_c$, the projected ensemble flows instead to $\mathcal{E}_{\rm Haar}$, and application of a random unitary on $L'$ effectively rotates the measurement away from the $X$-basis, leaving some entanglement between $R$ and $Q$. This can be quantified by the average purity of the reference qubit:
\begin{align}
    \langle \mathcal{P} \rangle_R=\frac{2}{3}
\end{align}
as can be obtained by Weingarten calculus. 

This argument can be adapted to the setting of Fig.~\ref{fig:GHZ}(a), where we do not apply the extra $ZZ$ measurements on the final column of qubits. The GHZ state on $RL'Q$ can be replaced by a superposition over different $ZZ$ measurement outcomes,
\begin{align}\label{eq:RLQM}
    |\psi_{\v{s},l}\rangle_{RQM} 
    = \sum_{\v{m}} \sqrt{2p_{\v{m}|\v{s}}} \langle l_{L'}|U_{\v{m},\v{s}}|{\rm GHZ}\rangle_{RL'Q} |\v m\rangle_M,
\end{align}
where $l$, $\v s$, and $\v m$ are the measurement outcomes of $L'$, syndrome, and $ZZ$ operators respectively ($\v m\in\{0,1\}^{(d_y-1)/2}$), we introduced a label $M$ for the latter subsystem, and we used the fact that $p_l=\frac{1}{2}$. 
In this case, when $p<p_c$, the reference qubit is purified on average in the thermodynamic limit for the same reason. 
For $p>p_c$, we show in Appendix~\ref{ap:upperbound} that the averaged purity in the reference qubit satisfies
\begin{align}\label{eq:upperbound}
    \langle \mathcal{P} \rangle_R\leq\frac{2}{3}.
\end{align}
Intuitively, the $ZZ$ measurements on average tend to purify $R$, so if $R$ remains entangled to $Q$, then it was likely to be entangled to $QM$ before $M$ was measured. 

Under this perspective,  whether the reference qubit is entangled with the rest of the system can be directly related to whether the projected ensemble flows to the identity/Pauli or to the Haar distribution. Note that the entanglement/purification phase transition is independent of the decoder, thus generally coincides with the $k\geq2$-design transition.

\section{Discussion} \label{sec:discussion}

\subsection{Summary}

We have introduced a protocol to apply random unitary gates directly on encoded qubits by leveraging the intrinsic randomness of measurement in quantum mechanics. 
Our protocol is based on the application of unitary transformations on the physical qubits followed by syndrome measurements and correction. Each of these steps can be implemented efficiently on quantum processors, and requires no auxiliary qubits (unlike e.g. magic state distillation). 

We have provided criteria on the quantum error correcting code and on the physical unitary transformation (Theorem~\ref{theo:logicalunitary}) that ensure unitarity of the random logical gates, significantly generalizing previously known results and allowing the realization of arbitrary logical gates. 
We have then investigated, both numerically and analytically, the emergence of unitary designs in this protocol, uncovering a sharp threshold in the strength $p$ of the applied physical unitaries. When $p>p_c$, we have a ``unitary design phase'' where, upon taking the code distance $d\to\infty$, the distribution of logical unitaries approaches the Haar random distribution; when $p < p_c$, all gates in the distribution converge toward the identity, which can be seen as successful correction of coherent errors, or to random Pauli unitaries, which correspond to unsuccessful error correction due to a suboptimal decoder.
Finally, we have presented a matrix product state algorithm to simulate this protocol, leveraging ``staircase'' encoder and decoder circuits for the surface code, by which we map the original 2D problem onto a $(1+1)$D monitored dynamics. The same threshold $p_c$ that controls the optimal error correction threshold and unitary design phases also emerges as an entanglement phase transition in this effective monitored dynamics, and thus gives a transition in computational complexity for our MPS algorithm.

\subsection{Outlook}

Our results open a number of directions for future investigation and pave the way to several applications. We outline some of them next. 

\subsubsection{Applications of logical unitary designs}

Unitary designs have many applications in quantum information science. Our work presents an efficient route to port these applications to the setting of encoded logical qubits, which is especially interesting as we enter the era of error-corrected quantum processors.

First, we note that some applications such as ``Pauli-twirling''~\cite{dur_standard_2005,cai_constructing_2019} only require a unitary $k$-design for $k=1$. On encoded qubits, that can be achieved directly by transversal Pauli unitaries.  In these cases the higher level of randomness afforded by our protocol is unnecessary. 
Further, many protocols~\cite{knill_randomized_2008,dankert_exact_2009,huang_predicting_2020,elben_randomized_2023} require a unitary $k$-design with $k=2$ or 3, which is realized by the Clifford group. Some error correcting codes where our protocol is applicable, such as the odd-distance color code (on the 2D honeycomb lattice), also admit a transversal implementation of the Hadamard and phase gate, which together generate the Clifford group and thus offer a straightforward route to transversal 3-designs.

Our protocol becomes useful for applications requiring $k=2$ or 3-designs on codes that do not admit a transversal implementation of the Clifford group, such as the surface code~\cite{brown_poking_2017}; or for any any applications requiring a 4-design or higher, which entails non-Clifford (or ``magic'') gates. 
Below we list some such applications.

{\bf Classical shadow tomography} is an efficient protocol to learn many properties of quantum states via randomized measurements. One requires a 2-design to get correct results on average, and a 3-design to bound statistical fluctuations~\cite{huang_predicting_2020,elben_randomized_2023}. Classical shadow tomography is usually formulated for physical qubits, but it is natural to envision applications on fully- or partially-error corrected quantum processors where one would need randomized {\it logical} measurements. On the surface code these could not be implemented transversally, so our method could prove advantageous. A key aspect of the protocol is its division into a ``data acquisition'' phase, where random unitary gates drawn from a 3-design are applied to the processor before measurement, and a separate ``classical postprocessing'' phase where the same unitaries are simulated on a classical computer. The feasibility of that simulation is usually guaranteed by the use of Clifford gates. In our case we instead rely on the MPS algorithm presented in Sec.~\ref{sec:staircase}. This highlights a trade-off between randomness and computational complexity (the MPS algorithm becomes asymptotically inefficient above threshold, in the unitary design phase).

{\bf Randomized benchmarking} is a family of practical approaches to estimate the fidelity of unitary gates on a noisy quantum computer~\cite{knill_randomized_2008}. It is based on the application of a long sequence of random gates drawn from a 2-design, followed by a measurement. Comparison of the actual mesurement statistics with the ideal one expected under noiseless gates yields an accurate estimate of the fidelity. With our protocol one could implement randomized benchmarking directly at the logical level. This is interesting in practice since logical gates comprise many physical gates, so learning their fidelity or error models is more challenging. 
Additionally, {\it higher-order randomized benchmarking} was proposed in Ref.~\cite{nakata_quantum_2021} as a generalization of the aforementioned protocol that makes use of $k$-designs with $k > 2$ to learn more detailed features of the noise model beyond the fidelity, for example whether or not the noise channel is self-adjoint. The application proposed in Ref.~\cite{nakata_quantum_2021} requires exact rather than approximate designs, but we note that in the design phase the error $\epsilon$ is suppressed exponentially in the distance $d$, so issues of approximation can likely be avoided by simply scaling up. Finally we note that, since our protocol implements logical unitaries nondeterministically, it is not directly compatible with standard randomized benchmarking schemes, which rely on deterministic inversion of gates. Instead, our protocol is suitable for inverse-free randomized benchmarking frameworks, such as binary randomized benchmarking~\cite{Hines_fully_2024} and cross-entropy benchmarking~\cite{Boixo_characterizing_2018,arute_quantum_2019}, where inversion gates are replaced by classical postprocessing based on the known gate configuration and measurement outcomes.

{\bf Quantum cryptography applications} have been a major driver for the introduction of unitary designs~\cite{ambainis_quantum_2007,lancien_weak_2020}. Our approach could thus facilitate cryptographic protocols on encoded qubits. Furthermore, since our protocol is not restricted to a specific value of $k$, it would be interesting to investigate the generation of pseudorandom unitaries~\cite{schuster_random_2024}. These are unitary ensembles that are computationally indistinguishable from the Haar ensemble even with access to many copies, and are of great interest in quantum cryptography~\cite{ji_pseudorandom_2018}. 

{\bf Random circuit sampling experiments} are a key benchmark of computational supremacy in present-day quantum computers~\cite{arute_quantum_2019} and crucially require high levels of randomness that goes beyond the Clifford group (3-design). Even non-Clifford circuits may be classically simulated by so-called Clifford perturbation theory~\cite{begusic_real-time_2024} (or Pauli path~\cite{gonzalez-garcia_pauli_2024}) methods, if the gates are insufficiently far from Clifford. To remove any loopholes that may be exploited by classical algorithms, it is imperative to build circuits with gates that are as close as possible to genuinely random.
Our protocol could find application in next-generation random circuit sampling experiments on partially error-corrected computers. Circuits that combine our highly random single-qubit logical gates with transversal Clifford gates on two logical qubits (e.g. iSWAP, which scrambles very quickly~\cite{mi_information_2021}) appear especially promising.

\subsubsection{Effect of noise}

In this work we have focused on the ideal setting where unitary control and measurements are perfect. Where we discussed an error correction threshold, it was in reference to ``coherent errors'' (the physical unitaries $\{u_j\}$) that are {\it applied intentionally} and known by the experimentalist. 
The performance of this protocol in the presence of real errors is a natural and crucial question. While a thorough investigation is left to future work, our present results already provide some insight on this. 

The ``unitary design phase'' is defined by the strong randomness in the post-measurement ensemble of unitaries $\{U_{L,\v s}\}$. It is reasonable to expect in that phase that additional (uncontrolled) coherent errors would cause significant damage. The application of unitary designs, as outlined above, requires a classical post-processing step where data from the noisy experiment is cross-correlated with a simulation of the protocol. If we aim to implement physical unitaries $\{u_j\}$ but, due e.g. to systematic over- or under-rotation, we implement some other set of gates $\{u'_j\}$, this will lead to a discrepancy between the real and simulated projected ensembles, $\{U_{L,\v s}^{\rm (exp.)} \}$ and $\{U_{L,\v s}^{\rm (sim.)} \}$. 
Below threshold both ensembles should flow to the trivial `identity ensemble' $\mathcal{E}_{\mathbb{I}}$, thus there is no discrepancy---a signature of the error-correcting phase. But above threshold one expects a strong sensitivity to the choice of $\{u_j\}$ and thus a significant difference between experimental and simulated ensembles. 

The analysis of incoherent errors is more complex (as it involves a projected ensemble of {\it channels}) but one can expect similar behavior. Measurement error is the simplest to analyze, as it just leads to a permutation of the labels $\v s$ in the projected ensemble; again we would expect this to cause significant logical errors above threshold. 

Based on these arguments one should not expect the protocol as stated here to be fault-tolerant above threshold. There is thus a trade-off between the amount of randomness in the generated unitary ensemble and the fidelity of the operation (and also the classical computational complexity, see Sec.~\ref{sec:staircase}). It is interesting to conjecture that operation at or near the threshold, with multiple repeated rounds of the protocol, might give best results in practice and possibly even fault-tolerance. These questions are left to future investigation. 

Another important consequence of noise is on the classical simulation of the protocol (necessary for any application of the random logical gates). 
In the absence of noise, the ensemble of logical gates is independent of how the code is implemented, so in our simulation we are free to use the staircase encoder/decoder circuits which enable our mapping to one dimension. 
With noise, the specific encoding and syndrome measurement circuits would affect error propagation and would thus need to be included in the simulation, limiting the applicability of our approach.
However, to a first approximation, state preparation errors can be treated as channels acting on the input logical state, while measurement errors can be treated as bit flips on the syndrome labels (permutations of the ensemble). 
In the middle of the protocol, coherent errors can be included straightforwardly. Incoherent errors on the other hand are more challenging, since our MPS method relies on area-law entanglement scaling in pure states. While methods based on matrix product operator representations of the density matrix exist~\cite{darmawan_tnsurfacecode_2017,darmawan_optimal_2024}, there is no guarantee of an area-law scaling of (operator) entanglement in that setting. The possibility of leveraging entanglement transitions under incoherent noise remains an outstanding question. 

\subsubsection{Extension to general LDPC codes and multiple logical qubits}

Our Theorem~\ref{theo:logicalunitary} establishes the unitarity of the projected ensemble for an important set of codes: CSS codes with one logical qubit, odd distance, and whose stabilizers have even Pauli weight. These include the surface code and color code which are practically very relevant to near-term architectures, particularly superconductig qubits where spatial locality is required. 
It would be interesting to extend our results to more general low-density particy check (LDPC) codes~\cite{breuckmann_ldpc_2021}. These are codes with parity checks that are few-body, but not necessarily geometrically local. Relaxing geometric locality can be beneficial in terms of scaling of code rate and distance~\cite{bravyi_no-go_2009}, and is practically relevant to platforms such as trapped ions and neutral atom tweezer arrays. 

A challenge lies in the fact that our Theorem~\ref{theo:logicalunitary} applies only to codes with one logical qubit ($k=1$), whereas a key motivation for moving to general LDPC codes is the larger number of logical qubits, potentially even finite encoding rate ($k\propto n$ as $n\to\infty$). It would be interesting to identify variations of our protocol that can give logical unitaries on codes with $k>1$. Product constructions of LDPC codes~\cite{rakovszky_physics_2024} could be a promising route toward this goal, since our understanding of the protocol for repetition and surface codes may be leveraged as a building block to address more complex LDPC codes.

{\it Note added.} During completion of this manuscript we became aware of related works studying the same or related transitions in the surface code. Ref.~\cite{behrends_statistical_2024} maps the surface code under coherent single-qubit errors to a non-unitary transfer matrix which undergoes an entanglement phase transition.  Ref.~\cite{bao_phases_2024} uses a staircase circuit similar to ours and also uncovers an entanglement phase transition. Neither work discusses unitarity of the logical operation or emergence of unitary designs. 

\acknowledgments
This material is based upon work supported by the Defense Advanced Research Projects Agency (DARPA) under Agreement No. HR00112490357.
M.J.G. acknowledges support from the NSF QLCI award no. OMA-2120757. The views, opinions and/or findings expressed are those of the authors and should not be interpreted as representing the official views or policies of the Department of Defense or the U.S. Government. Numerical simulations were performed in part on HPC resources provided by the Texas Advanced Computing Center (TACC) at the University of Texas at Austin.

\appendix

\section{Proof of Theorem \ref{theo:logicalunitary}}\label{ap:theo1proof}
Here we present a proof of Theorem~\ref{theo:logicalunitary}. 
We start by writing the logical density matrix of our $[\![n,1,d]\!]$ stabilizer code in the form
\begin{align}
    \rho_L=\frac{\mathbb{I}+\(\v{n}\cdot\v{\sigma}_L\)}{2}\Pi_{\boldsymbol{0}},
\end{align}
where $\v{\sigma}_{L}=\(X_L, Y_L, Z_L\)$ is the vector of logical operators. 
We consider the anti-unitary time-reversal operator $\mathcal{T} = \mathcal{K} \otimes_{j=1}^n\(iY_j\)$, where $\mathcal K$ is complex conjugation. This anti-unitary transformation acts by conjugation on Pauli strings $P$ as
\begin{align}
    \mathcal{T} P \mathcal{T}^{-1} = (-1)^{|P|}P,
\end{align}
where $|P|$ is the Pauli weight of the Pauli string (number of non-identity Pauli matrices in the string).

It is easy to show using conditions (1) and (2) that any representative of any logical operator must have odd weight. It is helpful to prove the following simple fact:
\begin{fact} \label{fact:weight}
    Consider two Pauli operators $A$ and $B$.
    If $AB = BA$, then $C = AB$ is a Hermitian Pauli operator of weight $|C| = |A| + |B| \mod 2$.
    If $AB = -BA$, then $C = AB/i$ is a Hermitian Pauli operator of weight $|C| = |A| + |B| + 1\mod 2$. 
\end{fact}
\begin{proof}
    This follows straightforwardly from $\mathcal{T} AB\mathcal{T}^{-1} = (-1)^{|A| + |B|} AB = \mathcal{T} C \mathcal{T}^{-1} = (-1)^{|C|}C$ in the first case, and $\mathcal{T} AB\mathcal{T}^{-1} = (-1)^{|A| + |B|} AB = \mathcal{T} iC\mathcal{T}^{-1} = -(-1)^{|C|}iC$ in the second (using anti-unitarity of $\mathcal T$).
\end{proof}
Since stabilizers commute with each logical operation, it follows that for any stabilizer $g$ and any logical $\sigma_L$ we have $|\sigma_L g| = |\sigma_L| + |g| \mod 2$; using $|g| = 0\mod 2$ (condition (1)) this gives $|\sigma_L g| = |\sigma_L| \mod 2$. So all representatives of a given logical have the same Pauli weight modulo 2. 
Additionally, since $d_x$ and $d_z$ are both odd, we have that $|X_L|$ and $|Z_L|$ are odd, and since $X_L$ and $Z_L$ anticommute, Fact~\ref{fact:weight} gives $|Y_L| = |X_L| + |Z_L| + 1 \mod 2$ and thus $|Y_L|$ is odd as well. We conclude that all nontrivial logical operators have odd Pauli weight.

Based on this fact we can decompose $\rho_L$ as
\begin{align}
    \rho_L=\rho_++\rho_-,
\end{align}
with $\rho_+ = \frac{1}{2} \Pi_0 $ made exclusively of even-weight Pauli operators and $\rho_-= \frac{1}{2} \Pi_0\(\v{n}\cdot\v{\sigma}_L\)$ made exclusively of odd-weigth Pauli operators: $\mathcal{T} \rho_\pm \mathcal{T}^{-1} = \pm \rho_\pm$. 

Condition (3) states that $[U,\mathcal{T}] = 0$, therefore 
the state after application of $U$ becomes
\begin{align}
U\rho_LU^\dagger
= U \rho_+ U^\dagger + U \rho_- U^\dagger 
= \rho_+' + \rho_-',
\end{align}
with $\rho'_{\pm} = U\rho_\pm U^\dagger$ such that $\mathcal{T} \rho'_\pm \mathcal{T}^{-1} = \pm \rho'_\pm$: $U$ does not mix the time-reversal-even and time-reversal-odd sectors. 
Now the probability of measuring syndrome $\v s$ on state $U\rho_LU^\dagger$ is
\begin{align}
    p(\v{s}) & = \Tr\(\Pi_{\v{s}} U\rho_L U^\dagger \)
    = \Tr(\Pi_{\v s} \rho'_+) + \Tr(\Pi_{\v s} \rho'_-),
\end{align}
with $\Pi_{\v s}$ the projector on syndrome $\v s$. It can be easily seen that
\begin{align}
    \Pi_{\v s} 
    = \prod_{i=1}^{n-1} \frac{\mathbb{I} + (-1)^{s_i} g_i}{2} 
\end{align}
contains only even-weight operators (stabilizers), therefore $\rho'_- \Pi_{\v s}$ contains only odd-weight operators (logicals). As the latter are traceless, we get 
\begin{align}
    p(\v{s})
    & = \Tr(\Pi_{\v s} \rho'_+) 
    = \frac{1}{2} \Tr(\Pi_{\v 0} U\Pi_{\v s} U^\dagger), 
\end{align} 
which is independent of the input logical state. 
In words, all the logical information is stored in odd-weight operators, but the syndrome measurements only probe even-weight operators, and so learn no logical information. 
At this point Theorem~\ref{theo:logicalunitary} follows from application of Fact~\ref{fact:unitarity}. 

\section{Additional data on the distribution of logical unitaries}\label{ap:distribution}

In this appendix we provide more fine-grained data on the distributions of logical unitaries in both phases. This gives additional evidence for our claim of unitary designs above the error correction threshold. 

To simplify our analysis, we consider {deterministic, uniform} coherent errors $U_i=\exp(i\alpha\v{n}\cdot\v{\sigma}_i)$, where $\v{n}=(1, 1, 1)/\sqrt{3}$ and $\v{\sigma}_i=(X_i, Y_i, Z_i)$. 
The angle $\alpha$ plays a similar role to the density of coherent errors $p$ used in the main text and tunes a design/error-correction phase transition in the projected ensemble. 
Data shown in Fig.~\ref{fig:D4_111} on the 4-design distance $\Delta^{(4)}$ indicates a design phase above $\alpha \gtrsim 0.18\pi$. Note that for both $\alpha = 0$ and $\alpha = \pi/3$ we have Clifford coherent errors, which cannot form more than 1-designs (see App.~\ref{ap:clifford}), so an upper limit to the design phase is expected before $\alpha = \pi/3$, hence the non-monotonic behavior of $\Delta^{(4)}$ in Fig.~\ref{fig:D4_111}. 

\begin{figure}
    \centering
    \includegraphics[width=1.0\linewidth]{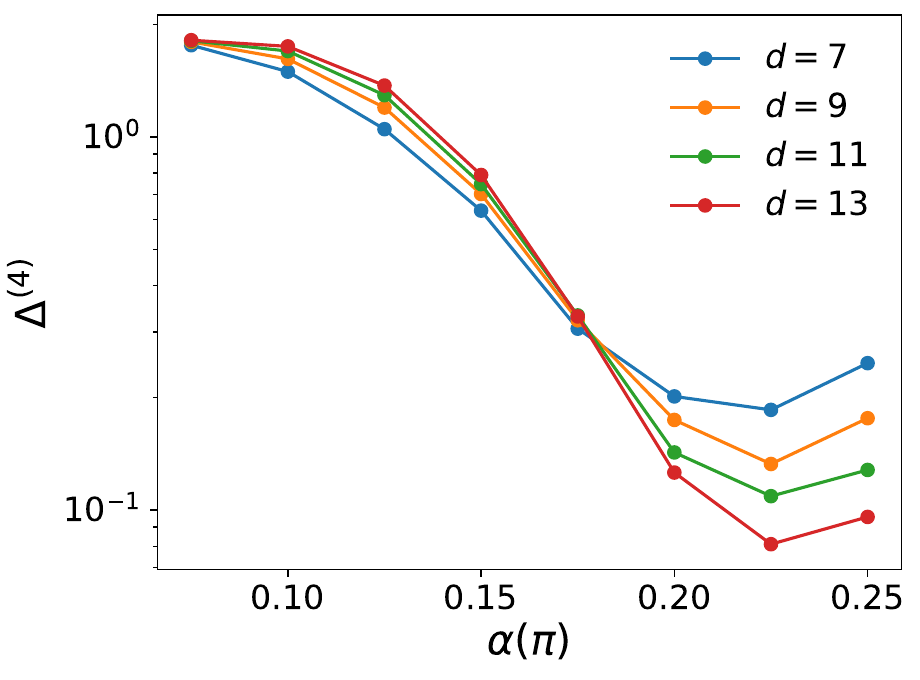}
    \caption{Unitary $4$-design distance $\Delta^{(4)}$ with random decoder for transversal coherent error $U_i=\exp(i\alpha\v{n}\cdot\v{\sigma}_i)$, with $\v{n}=(1, 1, 1)/\sqrt{3}$. Data are sampled over 10000 instances.}
    \label{fig:D4_111}
\end{figure}

To investigate the unitary distributions in both phases, we choose $\alpha=0.15\pi$ as representative of the error-correcting phase and $\alpha=0.225\pi$ as representative of the unitary-design phase. 
We parametrize the logical unitary in terms of three angles $\beta_{\v s}$, $\theta_{\v s}$ and $\phi_{\v s}$ as
\begin{align}\label{eq:U_para}
    U(\v{s}) 
    & = \exp\(i\beta_{\v{s}} \v{n}_{\v s} \cdot \boldsymbol{\sigma}_L\), \\ 
    \v{n}_{\v s} 
    & = \(\sin\theta_{\v{s}}\cos\varphi_{\v{s}}, \sin\theta_{\v{s}}\sin\varphi_{\v{s}}, \cos\theta_{\v{s}}\). 
\end{align}
For the Haar distribution, the probability density function of $(\beta,\theta,\varphi)$ is given by
\begin{align}
    p(\beta, \theta, \varphi) \propto \sin^2\beta\sin\theta. \label{eq:angle_distribution_haar}
\end{align}

In Fig.~\ref{fig:dist_111}, we compare the marginal distributions of $\beta$, $\theta$, and $\varphi$ between the Haar ensemble and the projected ensemble at $\alpha=0.15\pi$ and $\alpha=0.225\pi$. The former clearly shows deviations from Haar in the form of sharp peaks around multiples of $\pi/2$ for each angle, which are characteristic of Pauli unitaries; the peaks become sharper as the distance $d$ increases form 7 to 13, consistent with convergence toward a random Pauli distribution.
On the contrary, for $\alpha = 0.225\pi$ we see close agreement with the Haar prediction already at $d = 7$, with further improvement (visible especially in $\theta$) as we scale up to $d = 13$. 
This provides fine-grained evidence of the approach to the Haar distribution in the design phase. 

\begin{figure*}
    \centering
    \includegraphics[width=1.0\linewidth]{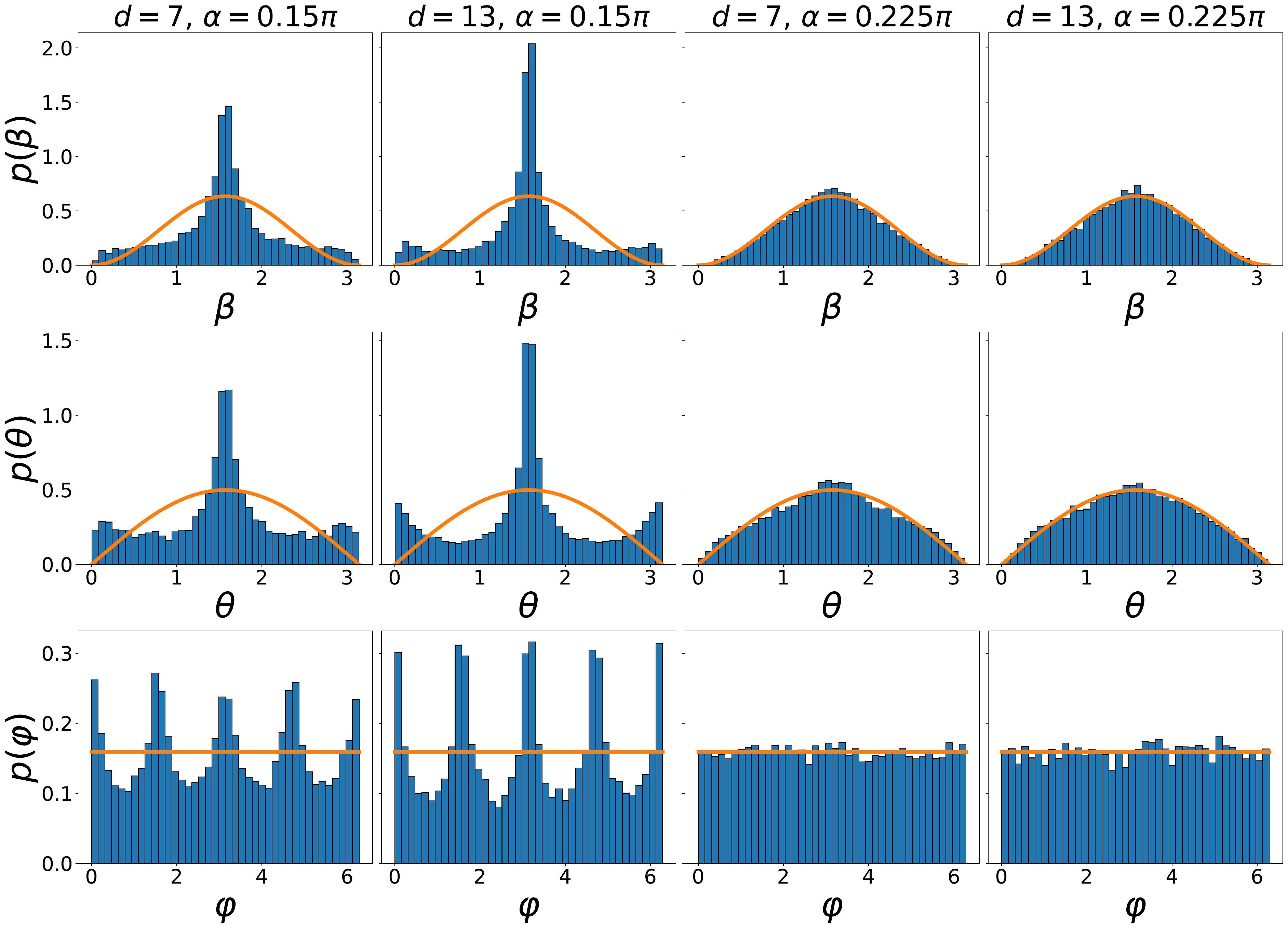}
    \caption{Distributions of unitary parameters $\beta,\theta,\phi$ in Eq.\eqref{eq:U_para} in the projected ensembles for code distances $d=7$ and $13$, and for coherent error strengths $\alpha=0.15\pi$ (below threshold) and $0.225\pi$ (above threshold). 
    The probability density functions for the Haar distribution, Eq.~\eqref{eq:angle_distribution_haar}, are shown by orange lines. Data are sampled over 10000 syndromes.}
    \label{fig:dist_111}
\end{figure*}

\section{Clifford simulation of error correction threshold and entanglement transition}\label{ap:clifford}

In this appendix, we numerically investigate the logical unitary design phase transition and entanglement phase transition in Clifford circuits, where we replace the Haar-random coherent errors by random Clifford ones. Since the encoding and decoding circuits are themselves Clifford, the overall circuit is amenable to large-scale classical simulation via stabilizer formalism~\cite{Aaronson2004}. Part of our simulation uses the \textit{QuantumClifford.jl} package~\cite{Krastanov2024}. 

We compute the tripartite mutual information $I_3$ as a diagnostic of the entanglement phase transition~\cite{Zabalo2020critical} , with $I_3$ defined as
\begin{align}
    I_3=&S(A)+S(B)+S(C)-S(A\cup B)\nonumber\\
    &-S(A\cup C)-S(B\cup C)+S(A\cup B\cup C), 
\end{align} 
where $S(X)$ is the von Neumann entropy of region $X$ and $A, B, C$ are contiguous subsystems of linear size $(d-1)/4$.

An important feature of the projected ensemble of logical unitaries in the Clifford case is that it can form at most a $1$-design. This follows from its limited cardinality: at most four distinct elements belong to $\mathcal{E}_L$. 
One can show this by writing the Choi state associated to a logical unitary $U_{L,\v{s}}$:
\begin{align}
    |U(\v{s})\rangle=\frac{1}{\sqrt{2}}\(|0_R\rangle U_L(\v{s})|0_L\rangle+|1_R\rangle U_L(\v{s})|1_L\rangle\).
\end{align}
This is a stabilizer state of two qubits and as such can be represented by two stabilizer generators $g_1=P_{1,R}P_{1,L}$, $g_2=P_{2,R}P_{2,L}$ where $P_{1/2, R/L}$ are Pauli matrices on the two qubits that we call $L$ and $R$. Unitarity of $U_L$ implies that this state is maximally entangled, i.e. $P_{1, R/L}\neq P_{2, R/L}$. For stabilizer states, the outcome of projective measurements can only affect the ``phase bits'' in the stabilizer tableau, i.e. the $\pm 1$ signs in front of each stabilizer generator. So for any given syndrome $\v s$ we have
\begin{align}
    |U(\v{s})\rangle\langle U(\v{s})| \in \left\{ \frac{\mathbb{I}\pm g_1}{2}\frac{\mathbb{I}\pm g_2}{2} \right\}
\end{align}
which as claimed contains only four elements.
More specifically we have 
\begin{align}
    U(\v{s})=CP(\v{s}),
\end{align}
where $C$ is some Clifford gate independent of $\v s$ and $P(\v{s})$ is a Pauli gate that may depend on $\v s$. As a result, the ensemble consisting of $CP(\v{s})$ can only form a unitary $1$-design.

We summarize our numerical results for Clifford simulation in Fig.~\ref{fig:clifford} on purification time $\tau$, tripartite mutual information $I_3$, logical error rate $\Delta_{\rm EC}$. Their critical points are all around $p_c\approx0.945$, indicating the coincidence of entanglement, purification, error correction, and unitary design phase transitions.

\begin{figure}
    \centering
    \includegraphics[width=1.0\linewidth]{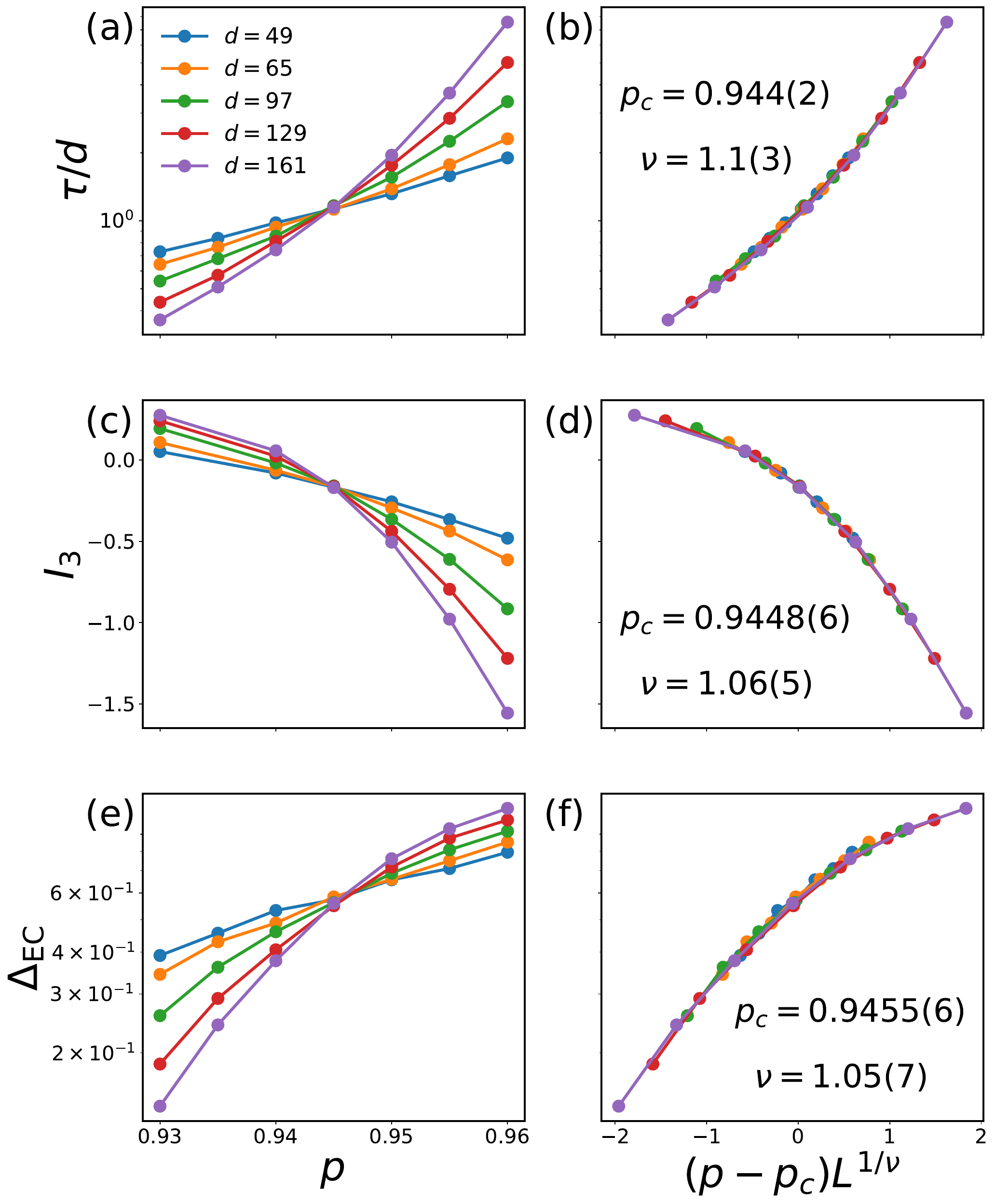}
    \caption{(a,b) Purification times $\tau$, (c,d) tripartite mutual information $I_3$, (e,f) logical error rate $\Delta_{\rm EC}$ under the optimal decoder as a function of the probability $p$ of applying random single-qubit Clifford operators. (b,d,f) represent scaling collapse of the data in (a,c,e). The data are averaged over $5120$ realizations.}
    \label{fig:clifford}
\end{figure}

\section{Decomposition of surface codes}\label{ap:decomposition}

\begin{figure}
    \centering
    \includegraphics[width=1.0\linewidth]{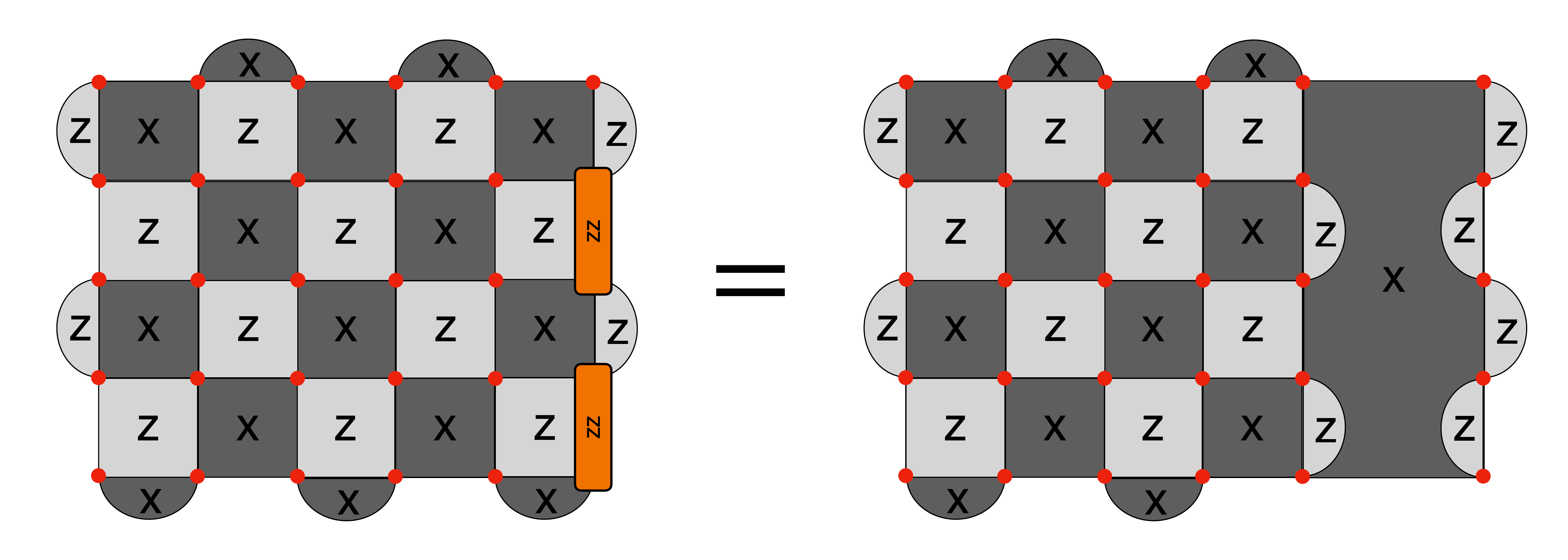}
    \caption{Decomposition of a $6\times5$ surface code into a $5\times5$ surface code and $d=5$ repetition code after some edge $ZZ$ measurements.}
    \label{fig:decomposition}
\end{figure}

Here we explain the decomposition of surface codes we used in Sec.~\ref{sec:connection} and Fig.~\ref{fig:GHZ}(b) to connect unitary design and entanglement phases. The main statement is that, after some $ZZ$ measurements (or $XX$ depending the orientation of the surface code) on an edge, the stabilizers of a $(d_x+1)\times d_y$ surface code can be decomposed into stabilizers of a $d_x\times d_y$ surface code and a $1\times d_y$ repetition code. Different measurement outcomes correspond to different signs of stabilizers at the edge. 

In Fig.~\ref{fig:decomposition}, we show this explicitly for the case $d_x=d_y=5$. At the same time, if we start from a Bell pair state between the logical $L$ of the $(d_x+1)\times d_y$ surface code and a reference qubit $R$
\begin{align}
    \rho_{RL} \propto \frac{\mathbb{I}+X_RX_L}{2}\frac{\mathbb{I}+Z_RZ_L}{2},
\end{align}
after the decomposition, we have a three qubit $X$-basis GHZ state over the reference $R$, the logical $L'$ of the $d_x\times d_y$ surface code, and the logical $Q$ of the $1\times d_y$ repetition code:
\begin{align}\label{eq:rho_RLQ}
    \rho_{RL'Q}\propto \frac{\mathbb{I}+X_RX_{L'}}{2}\frac{\mathbb{I}+Z_RZ_{L'}Z_Q}{2}\frac{\mathbb{I}+X_{L'}X_Q}{2},
\end{align}
where we used the fact that $X_L=X_{L'}$, $Z_L=Z_{L'}Z_Q$ and the last $X_{L'}X_Q$ stabilizer comes from the large $X$ plaquette in Fig.~\ref{fig:decomposition} (that originates by merging smaller $X$ plaquettes due to the $ZZ$ measurements).

\section{Upper bound of averged purity in the reference qubit}\label{ap:upperbound}
Here we prove Eq.\eqref{eq:upperbound} for the state in Fig.~\ref{fig:GHZ}(a) in the design phase.  The state is explicitly given by Eq.\eqref{eq:RLQM}:
\begin{align}
    |\psi_{\v{s}l}\rangle_{RQM} = \sum_{\v m}\sqrt{2p_{\v{m}|\v{s}}}\langle l_{L'}|U_{\v{m},\v{s}}|{\rm GHZ}\rangle_{RL'Q}|\v m\rangle_M,
\end{align}
where $l$, $\v s$, $\v m$ are the measurement outcomes of $L'$, syndrome, and $ZZ$ operators on the last column respectively, and $p_{\v{m}|\v{s}}$ is the conditional proability of $\v m$ given $\v s$. 

The density matrix of the reference qubit $R$ is 
\begin{align}
    \rho(\v{s},l)_R 
    & = {\rm Tr}_{QM} (\ket{\psi_{\v{s}l}} \bra{\psi_{\v{s} l}} ) \nonumber \\
    & = \sum_{l'=\pm } \ketbra{l'}_R \left(\sum_{\v m} p_{\v{m} | \v{s}} |\bra{l} U_{\v{m},\v{s}}\ket{l'}|^2 \right).
\end{align}
The averaged purity thus reads
\begin{align}\label{eq:upperbound_derive}
    \langle \mathcal{P} \rangle_R =& \frac{1}{2}\sum_{\v s,l, l'}p_{\v s} \(\sum_{\v m}p_{\v m|\v s}|\langle l|U_{\v m, \v s}|l'\rangle|^2\)^2\nonumber\\
    \leq& \frac{1}{2}\sum_{\v s,l, l'}p_{\v s} \sum_{\v m} p_{\v m|\v s} |\langle l|U_{\v m, \v s}|l'\rangle|^4\nonumber\\
    =&\frac{1}{2}\sum_{\v m,l, l'}p_{\v m} \sum_{\v s}p_{\v s|\v m}|\langle l|U_{\v m, \v s}|l'\rangle|^4,
\end{align}
where we have used the fact $p_l=\frac{1}{2}$ and the Cauchy–Schwarz inequality $(\sum_{\v m} a_{\v m} b_{\v m})^2 \leq (\sum_{\v m} a_{\v m}^2) (\sum_{\v m} b_{\v m}^2)$ for vectors $a_{\v m} = \sqrt{p_{\v m|\v s}}$ and $b_{\v m} = \sqrt{p_{\v m|\v s}} |\langle l|U_{\v m, \v s}|l'\rangle|^2 $. We have also used the definition of conditional probability rule $p_{\v m|\v s} p_{\v s} = p_{\v s|\v m} p_{\v m}$. 

In the design phase, for all $\v m$ the ensemble $\mathcal{E}_{\v m} \equiv \{p_{\v s|\v m}, U_{\v m, \v s}\}$ flows to the Haar distribution, since different $\v m$'s differ only by the sign of certain stabilizers in the state $\rho_{RL'Q}$, Eq.~\eqref{eq:rho_RLQ}, which affects the ensemble by at most an overall Pauli rotation. 
Replacing $\sum_{\v s}p_{\v s|\v m} |\langle l|U_{\v m, \v s}|l'\rangle|^4 \to \mathbb{E}_{U\sim {\rm Haar}} [|\langle l|U|l'\rangle|^4] = \frac{1}{3}$ in the right hand side of Eq.\eqref{eq:upperbound_derive} yields $\langle \mathcal{P} \rangle_R  \leq \frac{2}{3}$.

\bibliography{logical_unitary_design}

\begin{thebibliography}{89}%
\makeatletter
\providecommand \@ifxundefined [1]{%
 \@ifx{#1\undefined}
}%
\providecommand \@ifnum [1]{%
 \ifnum #1\expandafter \@firstoftwo
 \else \expandafter \@secondoftwo
 \fi
}%
\providecommand \@ifx [1]{%
 \ifx #1\expandafter \@firstoftwo
 \else \expandafter \@secondoftwo
 \fi
}%
\providecommand \natexlab [1]{#1}%
\providecommand \enquote  [1]{``#1''}%
\providecommand \bibnamefont  [1]{#1}%
\providecommand \bibfnamefont [1]{#1}%
\providecommand \citenamefont [1]{#1}%
\providecommand \href@noop [0]{\@secondoftwo}%
\providecommand \href [0]{\begingroup \@sanitize@url \@href}%
\providecommand \@href[1]{\@@startlink{#1}\@@href}%
\providecommand \@@href[1]{\endgroup#1\@@endlink}%
\providecommand \@sanitize@url [0]{\catcode `\\12\catcode `\$12\catcode
  `\&12\catcode `\#12\catcode `\^12\catcode `\_12\catcode `\%12\relax}%
\providecommand \@@startlink[1]{}%
\providecommand \@@endlink[0]{}%
\providecommand \url  [0]{\begingroup\@sanitize@url \@url }%
\providecommand \@url [1]{\endgroup\@href {#1}{\urlprefix }}%
\providecommand \urlprefix  [0]{URL }%
\providecommand \Eprint [0]{\href }%
\providecommand \doibase [0]{https://doi.org/}%
\providecommand \selectlanguage [0]{\@gobble}%
\providecommand \bibinfo  [0]{\@secondoftwo}%
\providecommand \bibfield  [0]{\@secondoftwo}%
\providecommand \translation [1]{[#1]}%
\providecommand \BibitemOpen [0]{}%
\providecommand \bibitemStop [0]{}%
\providecommand \bibitemNoStop [0]{.\EOS\space}%
\providecommand \EOS [0]{\spacefactor3000\relax}%
\providecommand \BibitemShut  [1]{\csname bibitem#1\endcsname}%
\let\auto@bib@innerbib\@empty
\bibitem [{\citenamefont {Eastin}(2009)}]{eastin_restrictions_2009}%
  \BibitemOpen
  \bibfield  {author} {\bibinfo {author} {\bibfnamefont {B.}~\bibnamefont
  {Eastin}},\ }\bibfield  {title} {\bibinfo {title} {Restrictions on
  {Transversal} {Encoded} {Quantum} {Gate} {Sets}},\ }\bibfield  {journal}
  {\bibinfo  {journal} {Physical Review Letters}\ }\textbf {\bibinfo {volume}
  {102}},\ \href {https://doi.org/10.1103/PhysRevLett.102.110502}
  {10.1103/PhysRevLett.102.110502} (\bibinfo {year} {2009})\BibitemShut
  {NoStop}%
\bibitem [{\citenamefont {Knill}(2004)}]{knill_fault-tolerant_2004}%
  \BibitemOpen
  \bibfield  {author} {\bibinfo {author} {\bibfnamefont {E.}~\bibnamefont
  {Knill}},\ }\bibfield  {title} {\bibinfo {title} {Fault-{Tolerant}
  {Postselected} {Quantum} {Computation}: {Schemes}},\ }\bibfield  {journal}
  {\bibinfo  {journal} {arXiv:0402171}\ }\href
  {https://doi.org/10.48550/arXiv.quant-ph/0402171}
  {10.48550/arXiv.quant-ph/0402171} (\bibinfo {year} {2004})\BibitemShut
  {NoStop}%
\bibitem [{\citenamefont {Bravyi}(2005)}]{bravyi_universal_2005}%
  \BibitemOpen
  \bibfield  {author} {\bibinfo {author} {\bibfnamefont {S.}~\bibnamefont
  {Bravyi}},\ }\bibfield  {title} {\bibinfo {title} {Universal quantum
  computation with ideal {Clifford} gates and noisy ancillas},\ }\bibfield
  {journal} {\bibinfo  {journal} {Physical Review A}\ }\textbf {\bibinfo
  {volume} {71}},\ \href {https://doi.org/10.1103/PhysRevA.71.022316}
  {10.1103/PhysRevA.71.022316} (\bibinfo {year} {2005})\BibitemShut {NoStop}%
\bibitem [{\citenamefont {Hastings}(2018)}]{hastings_distillation_2018}%
  \BibitemOpen
  \bibfield  {author} {\bibinfo {author} {\bibfnamefont {M.~B.}\ \bibnamefont
  {Hastings}},\ }\bibfield  {title} {\bibinfo {title} {Distillation with
  {Sublogarithmic} {Overhead}},\ }\bibfield  {journal} {\bibinfo  {journal}
  {Physical Review Letters}\ }\textbf {\bibinfo {volume} {120}},\ \href
  {https://doi.org/10.1103/PhysRevLett.120.050504}
  {10.1103/PhysRevLett.120.050504} (\bibinfo {year} {2018})\BibitemShut
  {NoStop}%
\bibitem [{\citenamefont {Wills}\ \emph {et~al.}(2024)\citenamefont {Wills},
  \citenamefont {Hsieh},\ and\ \citenamefont
  {Yamasaki}}]{wills_constant-overhead_2024}%
  \BibitemOpen
  \bibfield  {author} {\bibinfo {author} {\bibfnamefont {A.}~\bibnamefont
  {Wills}}, \bibinfo {author} {\bibfnamefont {M.-H.}\ \bibnamefont {Hsieh}},\
  and\ \bibinfo {author} {\bibfnamefont {H.}~\bibnamefont {Yamasaki}},\
  }\bibfield  {title} {\bibinfo {title} {Constant-{Overhead} {Magic} {State}
  {Distillation}},\ }\bibfield  {journal} {\bibinfo  {journal}
  {arXiv:2408.07764}\ }\href {https://doi.org/10.48550/arXiv.2408.07764}
  {10.48550/arXiv.2408.07764} (\bibinfo {year} {2024})\BibitemShut {NoStop}%
\bibitem [{\citenamefont {Choi}\ \emph {et~al.}(2023)\citenamefont {Choi},
  \citenamefont {Shaw}, \citenamefont {Madjarov}, \citenamefont {Xie},
  \citenamefont {Finkelstein}, \citenamefont {Covey}, \citenamefont {Cotler},
  \citenamefont {Mark}, \citenamefont {Huang}, \citenamefont {Kale},
  \citenamefont {Pichler}, \citenamefont {Brand{\~a}o}, \citenamefont {Choi},\
  and\ \citenamefont {Endres}}]{Choi2023preparing}%
  \BibitemOpen
  \bibfield  {author} {\bibinfo {author} {\bibfnamefont {J.}~\bibnamefont
  {Choi}}, \bibinfo {author} {\bibfnamefont {A.~L.}\ \bibnamefont {Shaw}},
  \bibinfo {author} {\bibfnamefont {I.~S.}\ \bibnamefont {Madjarov}}, \bibinfo
  {author} {\bibfnamefont {X.}~\bibnamefont {Xie}}, \bibinfo {author}
  {\bibfnamefont {R.}~\bibnamefont {Finkelstein}}, \bibinfo {author}
  {\bibfnamefont {J.~P.}\ \bibnamefont {Covey}}, \bibinfo {author}
  {\bibfnamefont {J.~S.}\ \bibnamefont {Cotler}}, \bibinfo {author}
  {\bibfnamefont {D.~K.}\ \bibnamefont {Mark}}, \bibinfo {author}
  {\bibfnamefont {H.-Y.}\ \bibnamefont {Huang}}, \bibinfo {author}
  {\bibfnamefont {A.}~\bibnamefont {Kale}}, \bibinfo {author} {\bibfnamefont
  {H.}~\bibnamefont {Pichler}}, \bibinfo {author} {\bibfnamefont {F.~G. S.~L.}\
  \bibnamefont {Brand{\~a}o}}, \bibinfo {author} {\bibfnamefont
  {S.}~\bibnamefont {Choi}},\ and\ \bibinfo {author} {\bibfnamefont
  {M.}~\bibnamefont {Endres}},\ }\bibfield  {title} {\bibinfo {title}
  {Preparing random states and benchmarking with many-body quantum chaos},\
  }\href {https://doi.org/10.1038/s41586-022-05442-1} {\bibfield  {journal}
  {\bibinfo  {journal} {Nature}\ }\textbf {\bibinfo {volume} {613}},\ \bibinfo
  {pages} {468} (\bibinfo {year} {2023})}\BibitemShut {NoStop}%
\bibitem [{\citenamefont {Cotler}\ \emph {et~al.}(2023)\citenamefont {Cotler},
  \citenamefont {Mark}, \citenamefont {Huang}, \citenamefont {Hern\'andez},
  \citenamefont {Choi}, \citenamefont {Shaw}, \citenamefont {Endres},\ and\
  \citenamefont {Choi}}]{Cotler2023emergent}%
  \BibitemOpen
  \bibfield  {author} {\bibinfo {author} {\bibfnamefont {J.~S.}\ \bibnamefont
  {Cotler}}, \bibinfo {author} {\bibfnamefont {D.~K.}\ \bibnamefont {Mark}},
  \bibinfo {author} {\bibfnamefont {H.-Y.}\ \bibnamefont {Huang}}, \bibinfo
  {author} {\bibfnamefont {F.}~\bibnamefont {Hern\'andez}}, \bibinfo {author}
  {\bibfnamefont {J.}~\bibnamefont {Choi}}, \bibinfo {author} {\bibfnamefont
  {A.~L.}\ \bibnamefont {Shaw}}, \bibinfo {author} {\bibfnamefont
  {M.}~\bibnamefont {Endres}},\ and\ \bibinfo {author} {\bibfnamefont
  {S.}~\bibnamefont {Choi}},\ }\bibfield  {title} {\bibinfo {title} {Emergent
  quantum state designs from individual many-body wave functions},\ }\href
  {https://doi.org/10.1103/PRXQuantum.4.010311} {\bibfield  {journal} {\bibinfo
   {journal} {PRX Quantum}\ }\textbf {\bibinfo {volume} {4}},\ \bibinfo {pages}
  {010311} (\bibinfo {year} {2023})}\BibitemShut {NoStop}%
\bibitem [{\citenamefont {Ho}\ and\ \citenamefont
  {Choi}(2022)}]{ho_exact_2022}%
  \BibitemOpen
  \bibfield  {author} {\bibinfo {author} {\bibfnamefont {W.~W.}\ \bibnamefont
  {Ho}}\ and\ \bibinfo {author} {\bibfnamefont {S.}~\bibnamefont {Choi}},\
  }\bibfield  {title} {\bibinfo {title} {Exact {Emergent} {Quantum} {State}
  {Designs} from {Quantum} {Chaotic} {Dynamics}},\ }\href
  {https://doi.org/10.1103/PhysRevLett.128.060601} {\bibfield  {journal}
  {\bibinfo  {journal} {Physical Review Letters}\ }\textbf {\bibinfo {volume}
  {128}},\ \bibinfo {pages} {060601} (\bibinfo {year} {2022})}\BibitemShut
  {NoStop}%
\bibitem [{\citenamefont {Ippoliti}\ and\ \citenamefont
  {Ho}(2022)}]{ippoliti_solvable_2022}%
  \BibitemOpen
  \bibfield  {author} {\bibinfo {author} {\bibfnamefont {M.}~\bibnamefont
  {Ippoliti}}\ and\ \bibinfo {author} {\bibfnamefont {W.~W.}\ \bibnamefont
  {Ho}},\ }\bibfield  {title} {\bibinfo {title} {Solvable model of deep
  thermalization with distinct design times},\ }\href
  {https://doi.org/10.22331/q-2022-12-29-886} {\bibfield  {journal} {\bibinfo
  {journal} {Quantum}\ }\textbf {\bibinfo {volume} {6}},\ \bibinfo {pages}
  {886} (\bibinfo {year} {2022})}\BibitemShut {NoStop}%
\bibitem [{\citenamefont {Turkeshi}\ and\ \citenamefont
  {Sierant}(2024)}]{Turkeshi2024error}%
  \BibitemOpen
  \bibfield  {author} {\bibinfo {author} {\bibfnamefont {X.}~\bibnamefont
  {Turkeshi}}\ and\ \bibinfo {author} {\bibfnamefont {P.}~\bibnamefont
  {Sierant}},\ }\bibfield  {title} {\bibinfo {title} {Error-resilience phase
  transitions in encoding-decoding quantum circuits},\ }\href
  {https://doi.org/10.1103/PhysRevLett.132.140401} {\bibfield  {journal}
  {\bibinfo  {journal} {Phys. Rev. Lett.}\ }\textbf {\bibinfo {volume} {132}},\
  \bibinfo {pages} {140401} (\bibinfo {year} {2024})}\BibitemShut {NoStop}%
\bibitem [{\citenamefont {Niroula}\ \emph {et~al.}(2024)\citenamefont
  {Niroula}, \citenamefont {White}, \citenamefont {Wang}, \citenamefont
  {Johri}, \citenamefont {Zhu}, \citenamefont {Monroe}, \citenamefont {Noel},\
  and\ \citenamefont {Gullans}}]{niroula_phase_2024}%
  \BibitemOpen
  \bibfield  {author} {\bibinfo {author} {\bibfnamefont {P.}~\bibnamefont
  {Niroula}}, \bibinfo {author} {\bibfnamefont {C.~D.}\ \bibnamefont {White}},
  \bibinfo {author} {\bibfnamefont {Q.}~\bibnamefont {Wang}}, \bibinfo {author}
  {\bibfnamefont {S.}~\bibnamefont {Johri}}, \bibinfo {author} {\bibfnamefont
  {D.}~\bibnamefont {Zhu}}, \bibinfo {author} {\bibfnamefont {C.}~\bibnamefont
  {Monroe}}, \bibinfo {author} {\bibfnamefont {C.}~\bibnamefont {Noel}},\ and\
  \bibinfo {author} {\bibfnamefont {M.~J.}\ \bibnamefont {Gullans}},\
  }\bibfield  {title} {\bibinfo {title} {Phase transition in magic with random
  quantum circuits},\ }\href {https://doi.org/10.1038/s41567-024-02637-3}
  {\bibfield  {journal} {\bibinfo  {journal} {Nature Physics}\ }\textbf
  {\bibinfo {volume} {20}},\ \bibinfo {pages} {1786} (\bibinfo {year}
  {2024})}\BibitemShut {NoStop}%
\bibitem [{\citenamefont {Behrends}\ \emph {et~al.}(2024)\citenamefont
  {Behrends}, \citenamefont {Venn},\ and\ \citenamefont
  {B\'eri}}]{behrends_surface_2024}%
  \BibitemOpen
  \bibfield  {author} {\bibinfo {author} {\bibfnamefont {J.}~\bibnamefont
  {Behrends}}, \bibinfo {author} {\bibfnamefont {F.}~\bibnamefont {Venn}},\
  and\ \bibinfo {author} {\bibfnamefont {B.}~\bibnamefont {B\'eri}},\
  }\bibfield  {title} {\bibinfo {title} {Surface codes, quantum circuits, and
  entanglement phases},\ }\href
  {https://doi.org/10.1103/PhysRevResearch.6.013137} {\bibfield  {journal}
  {\bibinfo  {journal} {Phys. Rev. Res.}\ }\textbf {\bibinfo {volume} {6}},\
  \bibinfo {pages} {013137} (\bibinfo {year} {2024})}\BibitemShut {NoStop}%
\bibitem [{\citenamefont {Venn}\ \emph {et~al.}(2023)\citenamefont {Venn},
  \citenamefont {Behrends},\ and\ \citenamefont {B\'eri}}]{Venn_coherent_2023}%
  \BibitemOpen
  \bibfield  {author} {\bibinfo {author} {\bibfnamefont {F.}~\bibnamefont
  {Venn}}, \bibinfo {author} {\bibfnamefont {J.}~\bibnamefont {Behrends}},\
  and\ \bibinfo {author} {\bibfnamefont {B.}~\bibnamefont {B\'eri}},\
  }\bibfield  {title} {\bibinfo {title} {Coherent-error threshold for surface
  codes from majorana delocalization},\ }\href
  {https://doi.org/10.1103/PhysRevLett.131.060603} {\bibfield  {journal}
  {\bibinfo  {journal} {Phys. Rev. Lett.}\ }\textbf {\bibinfo {volume} {131}},\
  \bibinfo {pages} {060603} (\bibinfo {year} {2023})}\BibitemShut {NoStop}%
\bibitem [{\citenamefont {Knill}(2008)}]{knill_randomized_2008}%
  \BibitemOpen
  \bibfield  {author} {\bibinfo {author} {\bibfnamefont {E.}~\bibnamefont
  {Knill}},\ }\bibfield  {title} {\bibinfo {title} {Randomized benchmarking of
  quantum gates},\ }\bibfield  {journal} {\bibinfo  {journal} {Physical Review
  A}\ }\textbf {\bibinfo {volume} {77}},\ \href
  {https://doi.org/10.1103/PhysRevA.77.012307} {10.1103/PhysRevA.77.012307}
  (\bibinfo {year} {2008})\BibitemShut {NoStop}%
\bibitem [{\citenamefont {McClean}\ \emph {et~al.}(2018)\citenamefont
  {McClean}, \citenamefont {Boixo}, \citenamefont {Smelyanskiy}, \citenamefont
  {Babbush},\ and\ \citenamefont {Neven}}]{mcclean_barren_2018}%
  \BibitemOpen
  \bibfield  {author} {\bibinfo {author} {\bibfnamefont {J.~R.}\ \bibnamefont
  {McClean}}, \bibinfo {author} {\bibfnamefont {S.}~\bibnamefont {Boixo}},
  \bibinfo {author} {\bibfnamefont {V.~N.}\ \bibnamefont {Smelyanskiy}},
  \bibinfo {author} {\bibfnamefont {R.}~\bibnamefont {Babbush}},\ and\ \bibinfo
  {author} {\bibfnamefont {H.}~\bibnamefont {Neven}},\ }\bibfield  {title}
  {\bibinfo {title} {Barren plateaus in quantum neural network training
  landscapes},\ }\href {https://doi.org/10.1038/s41467-018-07090-4} {\bibfield
  {journal} {\bibinfo  {journal} {Nature Communications}\ }\textbf {\bibinfo
  {volume} {9}},\ \bibinfo {pages} {4812} (\bibinfo {year} {2018})}\BibitemShut
  {NoStop}%
\bibitem [{\citenamefont {Huang}\ \emph {et~al.}(2020)\citenamefont {Huang},
  \citenamefont {Kueng},\ and\ \citenamefont
  {Preskill}}]{huang_predicting_2020}%
  \BibitemOpen
  \bibfield  {author} {\bibinfo {author} {\bibfnamefont {H.-Y.}\ \bibnamefont
  {Huang}}, \bibinfo {author} {\bibfnamefont {R.}~\bibnamefont {Kueng}},\ and\
  \bibinfo {author} {\bibfnamefont {J.}~\bibnamefont {Preskill}},\ }\bibfield
  {title} {\bibinfo {title} {Predicting many properties of a quantum system
  from very few measurements},\ }\href
  {https://doi.org/10.1038/s41567-020-0932-7} {\bibfield  {journal} {\bibinfo
  {journal} {Nature Physics}\ }\textbf {\bibinfo {volume} {16}},\ \bibinfo
  {pages} {1050} (\bibinfo {year} {2020})}\BibitemShut {NoStop}%
\bibitem [{\citenamefont {Elben}\ \emph {et~al.}(2023)\citenamefont {Elben},
  \citenamefont {Flammia}, \citenamefont {Huang}, \citenamefont {Kueng},
  \citenamefont {Preskill}, \citenamefont {Vermersch},\ and\ \citenamefont
  {Zoller}}]{elben_randomized_2023}%
  \BibitemOpen
  \bibfield  {author} {\bibinfo {author} {\bibfnamefont {A.}~\bibnamefont
  {Elben}}, \bibinfo {author} {\bibfnamefont {S.~T.}\ \bibnamefont {Flammia}},
  \bibinfo {author} {\bibfnamefont {H.-Y.}\ \bibnamefont {Huang}}, \bibinfo
  {author} {\bibfnamefont {R.}~\bibnamefont {Kueng}}, \bibinfo {author}
  {\bibfnamefont {J.}~\bibnamefont {Preskill}}, \bibinfo {author}
  {\bibfnamefont {B.}~\bibnamefont {Vermersch}},\ and\ \bibinfo {author}
  {\bibfnamefont {P.}~\bibnamefont {Zoller}},\ }\bibfield  {title} {\bibinfo
  {title} {The randomized measurement toolbox},\ }\href
  {https://doi.org/10.1038/s42254-022-00535-2} {\bibfield  {journal} {\bibinfo
  {journal} {Nature Reviews Physics}\ }\textbf {\bibinfo {volume} {5}},\
  \bibinfo {pages} {9} (\bibinfo {year} {2023})}\BibitemShut {NoStop}%
\bibitem [{\citenamefont {Mele}(2024)}]{mele_introduction_2024}%
  \BibitemOpen
  \bibfield  {author} {\bibinfo {author} {\bibfnamefont {A.~A.}\ \bibnamefont
  {Mele}},\ }\bibfield  {title} {\bibinfo {title} {Introduction to {Haar}
  {Measure} {Tools} in {Quantum} {Information}: {A} {Beginner}'s {Tutorial}},\
  }\href {https://doi.org/10.22331/q-2024-05-08-1340} {\bibfield  {journal}
  {\bibinfo  {journal} {Quantum}\ }\textbf {\bibinfo {volume} {8}},\ \bibinfo
  {pages} {1340} (\bibinfo {year} {2024})}\BibitemShut {NoStop}%
\bibitem [{\citenamefont {Ambainis}\ and\ \citenamefont
  {Emerson}(2007)}]{ambainis_quantum_2007}%
  \BibitemOpen
  \bibfield  {author} {\bibinfo {author} {\bibfnamefont {A.}~\bibnamefont
  {Ambainis}}\ and\ \bibinfo {author} {\bibfnamefont {J.}~\bibnamefont
  {Emerson}},\ }\bibfield  {title} {\bibinfo {title} {Quantum t-designs: t-wise
  {Independence} in the {Quantum} {World}},\ }in\ \href
  {https://doi.org/10.1109/CCC.2007.26} {\emph {\bibinfo {booktitle}
  {Twenty-{Second} {Annual} {IEEE} {Conference} on {Computational} {Complexity}
  ({CCC}'07)}}}\ (\bibinfo  {publisher} {IEEE Computer Society},\ \bibinfo
  {address} {Los Alamitos, CA, USA},\ \bibinfo {year} {2007})\ pp.\ \bibinfo
  {pages} {129--140}\BibitemShut {NoStop}%
\bibitem [{\citenamefont {Dankert}(2009)}]{dankert_exact_2009}%
  \BibitemOpen
  \bibfield  {author} {\bibinfo {author} {\bibfnamefont {C.}~\bibnamefont
  {Dankert}},\ }\bibfield  {title} {\bibinfo {title} {Exact and approximate
  unitary 2-designs and their application to fidelity estimation},\ }\bibfield
  {journal} {\bibinfo  {journal} {Physical Review A}\ }\textbf {\bibinfo
  {volume} {80}},\ \href {https://doi.org/10.1103/PhysRevA.80.012304}
  {10.1103/PhysRevA.80.012304} (\bibinfo {year} {2009})\BibitemShut {NoStop}%
\bibitem [{\citenamefont {Roberts}\ and\ \citenamefont
  {Yoshida}(2017)}]{roberts_chaos_2017}%
  \BibitemOpen
  \bibfield  {author} {\bibinfo {author} {\bibfnamefont {D.~A.}\ \bibnamefont
  {Roberts}}\ and\ \bibinfo {author} {\bibfnamefont {B.}~\bibnamefont
  {Yoshida}},\ }\bibfield  {title} {\bibinfo {title} {Chaos and complexity by
  design},\ }\href {https://doi.org/10.1007/JHEP04(2017)121} {\bibfield
  {journal} {\bibinfo  {journal} {arXiv:1610.04903}\ }\textbf {\bibinfo
  {volume} {2017}},\ \bibinfo {pages} {121} (\bibinfo {year}
  {2017})}\BibitemShut {NoStop}%
\bibitem [{\citenamefont {Hunter-Jones}(2019)}]{hunter-jones_unitary_2019}%
  \BibitemOpen
  \bibfield  {author} {\bibinfo {author} {\bibfnamefont {N.}~\bibnamefont
  {Hunter-Jones}},\ }\bibfield  {title} {\bibinfo {title} {Unitary designs from
  statistical mechanics in random quantum circuits},\ }\bibfield  {journal}
  {\bibinfo  {journal} {arXiv:1905.12053}\ }\href
  {https://doi.org/10.48550/arXiv.1905.12053} {10.48550/arXiv.1905.12053}
  (\bibinfo {year} {2019})\BibitemShut {NoStop}%
\bibitem [{\citenamefont {Schuster}\ \emph {et~al.}(2025)\citenamefont
  {Schuster}, \citenamefont {Haferkamp},\ and\ \citenamefont
  {Huang}}]{schuster_random_2024}%
  \BibitemOpen
  \bibfield  {author} {\bibinfo {author} {\bibfnamefont {T.}~\bibnamefont
  {Schuster}}, \bibinfo {author} {\bibfnamefont {J.}~\bibnamefont
  {Haferkamp}},\ and\ \bibinfo {author} {\bibfnamefont {H.-Y.}\ \bibnamefont
  {Huang}},\ }\bibfield  {title} {\bibinfo {title} {Random unitaries in
  extremely low depth},\ }\href {https://doi.org/10.1126/science.adv8590}
  {\bibfield  {journal} {\bibinfo  {journal} {Science}\ }\textbf {\bibinfo
  {volume} {389}},\ \bibinfo {pages} {92} (\bibinfo {year} {2025})}\BibitemShut
  {NoStop}%
\bibitem [{\citenamefont {Ippoliti}\ and\ \citenamefont
  {Ho}(2023)}]{ippoliti_dynamical_2023}%
  \BibitemOpen
  \bibfield  {author} {\bibinfo {author} {\bibfnamefont {M.}~\bibnamefont
  {Ippoliti}}\ and\ \bibinfo {author} {\bibfnamefont {W.~W.}\ \bibnamefont
  {Ho}},\ }\bibfield  {title} {\bibinfo {title} {Dynamical {Purification} and
  the {Emergence} of {Quantum} {State} {Designs} from the {Projected}
  {Ensemble}},\ }\href {https://doi.org/10.1103/PRXQuantum.4.030322} {\bibfield
   {journal} {\bibinfo  {journal} {PRX Quantum}\ }\textbf {\bibinfo {volume}
  {4}},\ \bibinfo {pages} {030322} (\bibinfo {year} {2023})}\BibitemShut
  {NoStop}%
\bibitem [{\citenamefont {Bhore}\ \emph {et~al.}(2023)\citenamefont {Bhore},
  \citenamefont {Desaules},\ and\ \citenamefont {Papic}}]{bhore_deep_2023}%
  \BibitemOpen
  \bibfield  {author} {\bibinfo {author} {\bibfnamefont {T.}~\bibnamefont
  {Bhore}}, \bibinfo {author} {\bibfnamefont {J.-Y.}\ \bibnamefont
  {Desaules}},\ and\ \bibinfo {author} {\bibfnamefont {Z.}~\bibnamefont
  {Papic}},\ }\bibfield  {title} {\bibinfo {title} {Deep thermalization in
  constrained quantum systems},\ }\href
  {https://doi.org/10.1103/PhysRevB.108.104317} {\bibfield  {journal} {\bibinfo
   {journal} {Physical Review B}\ }\textbf {\bibinfo {volume} {108}},\ \bibinfo
  {pages} {104317} (\bibinfo {year} {2023})}\BibitemShut {NoStop}%
\bibitem [{\citenamefont {Mark}\ \emph {et~al.}(2024)\citenamefont {Mark},
  \citenamefont {Surace}, \citenamefont {Elben}, \citenamefont {Shaw},
  \citenamefont {Choi}, \citenamefont {Refael}, \citenamefont {Endres},\ and\
  \citenamefont {Choi}}]{mark_maximum_2024}%
  \BibitemOpen
  \bibfield  {author} {\bibinfo {author} {\bibfnamefont {D.~K.}\ \bibnamefont
  {Mark}}, \bibinfo {author} {\bibfnamefont {F.}~\bibnamefont {Surace}},
  \bibinfo {author} {\bibfnamefont {A.}~\bibnamefont {Elben}}, \bibinfo
  {author} {\bibfnamefont {A.~L.}\ \bibnamefont {Shaw}}, \bibinfo {author}
  {\bibfnamefont {J.}~\bibnamefont {Choi}}, \bibinfo {author} {\bibfnamefont
  {G.}~\bibnamefont {Refael}}, \bibinfo {author} {\bibfnamefont
  {M.}~\bibnamefont {Endres}},\ and\ \bibinfo {author} {\bibfnamefont
  {S.}~\bibnamefont {Choi}},\ }\bibfield  {title} {\bibinfo {title} {Maximum
  entropy principle in deep thermalization and in hilbert-space ergodicity},\
  }\href {https://doi.org/10.1103/PhysRevX.14.041051} {\bibfield  {journal}
  {\bibinfo  {journal} {Phys. Rev. X}\ }\textbf {\bibinfo {volume} {14}},\
  \bibinfo {pages} {041051} (\bibinfo {year} {2024})}\BibitemShut {NoStop}%
\bibitem [{\citenamefont {Chan}\ and\ \citenamefont
  {De~Luca}(2024)}]{chan_projected_2024}%
  \BibitemOpen
  \bibfield  {author} {\bibinfo {author} {\bibfnamefont {A.}~\bibnamefont
  {Chan}}\ and\ \bibinfo {author} {\bibfnamefont {A.}~\bibnamefont {De~Luca}},\
  }\bibfield  {title} {\bibinfo {title} {Projected state ensemble of a generic
  model of many-body quantum chaos},\ }\href
  {https://doi.org/10.1088/1751-8121/ad7211} {\bibfield  {journal} {\bibinfo
  {journal} {Journal of Physics A: Mathematical and Theoretical}\ }\textbf
  {\bibinfo {volume} {57}},\ \bibinfo {pages} {405001} (\bibinfo {year}
  {2024})}\BibitemShut {NoStop}%
\bibitem [{\citenamefont {Lucas}\ \emph {et~al.}(2023)\citenamefont {Lucas},
  \citenamefont {Piroli}, \citenamefont {De~Nardis},\ and\ \citenamefont
  {De~Luca}}]{lucas_generalized_2023}%
  \BibitemOpen
  \bibfield  {author} {\bibinfo {author} {\bibfnamefont {M.}~\bibnamefont
  {Lucas}}, \bibinfo {author} {\bibfnamefont {L.}~\bibnamefont {Piroli}},
  \bibinfo {author} {\bibfnamefont {J.}~\bibnamefont {De~Nardis}},\ and\
  \bibinfo {author} {\bibfnamefont {A.}~\bibnamefont {De~Luca}},\ }\bibfield
  {title} {\bibinfo {title} {Generalized deep thermalization for free
  fermions},\ }\href {https://doi.org/10.1103/PhysRevA.107.032215} {\bibfield
  {journal} {\bibinfo  {journal} {Physical Review A}\ }\textbf {\bibinfo
  {volume} {107}},\ \bibinfo {pages} {032215} (\bibinfo {year}
  {2023})}\BibitemShut {NoStop}%
\bibitem [{\citenamefont {Liu}\ \emph {et~al.}(2024)\citenamefont {Liu},
  \citenamefont {Huang},\ and\ \citenamefont {Ho}}]{liu_deep_2024}%
  \BibitemOpen
  \bibfield  {author} {\bibinfo {author} {\bibfnamefont {C.}~\bibnamefont
  {Liu}}, \bibinfo {author} {\bibfnamefont {Q.~C.}\ \bibnamefont {Huang}},\
  and\ \bibinfo {author} {\bibfnamefont {W.~W.}\ \bibnamefont {Ho}},\
  }\bibfield  {title} {\bibinfo {title} {Deep thermalization in gaussian
  continuous-variable quantum systems},\ }\href
  {https://doi.org/10.1103/PhysRevLett.133.260401} {\bibfield  {journal}
  {\bibinfo  {journal} {Phys. Rev. Lett.}\ }\textbf {\bibinfo {volume} {133}},\
  \bibinfo {pages} {260401} (\bibinfo {year} {2024})}\BibitemShut {NoStop}%
\bibitem [{\citenamefont {Chang}\ \emph {et~al.}(2025)\citenamefont {Chang},
  \citenamefont {Shrotriya}, \citenamefont {Ho},\ and\ \citenamefont
  {Ippoliti}}]{chang_deep_2024}%
  \BibitemOpen
  \bibfield  {author} {\bibinfo {author} {\bibfnamefont {R.-A.}\ \bibnamefont
  {Chang}}, \bibinfo {author} {\bibfnamefont {H.}~\bibnamefont {Shrotriya}},
  \bibinfo {author} {\bibfnamefont {W.~W.}\ \bibnamefont {Ho}},\ and\ \bibinfo
  {author} {\bibfnamefont {M.}~\bibnamefont {Ippoliti}},\ }\bibfield  {title}
  {\bibinfo {title} {Deep thermalization under charge-conserving quantum
  dynamics},\ }\href {https://doi.org/10.1103/PRXQuantum.6.020343} {\bibfield
  {journal} {\bibinfo  {journal} {PRX Quantum}\ }\textbf {\bibinfo {volume}
  {6}},\ \bibinfo {pages} {020343} (\bibinfo {year} {2025})}\BibitemShut
  {NoStop}%
\bibitem [{\citenamefont {Bravyi}\ \emph {et~al.}(2018)\citenamefont {Bravyi},
  \citenamefont {Englbrecht}, \citenamefont {K{\"o}nig},\ and\ \citenamefont
  {Peard}}]{Bravyi2018correcting}%
  \BibitemOpen
  \bibfield  {author} {\bibinfo {author} {\bibfnamefont {S.}~\bibnamefont
  {Bravyi}}, \bibinfo {author} {\bibfnamefont {M.}~\bibnamefont {Englbrecht}},
  \bibinfo {author} {\bibfnamefont {R.}~\bibnamefont {K{\"o}nig}},\ and\
  \bibinfo {author} {\bibfnamefont {N.}~\bibnamefont {Peard}},\ }\bibfield
  {title} {\bibinfo {title} {Correcting coherent errors with surface codes},\
  }\href {https://doi.org/10.1038/s41534-018-0106-y} {\bibfield  {journal}
  {\bibinfo  {journal} {npj Quantum Information}\ }\textbf {\bibinfo {volume}
  {4}},\ \bibinfo {pages} {55} (\bibinfo {year} {2018})}\BibitemShut {NoStop}%
\bibitem [{\citenamefont {Fisher}\ \emph {et~al.}(2023)\citenamefont {Fisher},
  \citenamefont {Khemani}, \citenamefont {Nahum},\ and\ \citenamefont
  {Vijay}}]{fisher_random_2023}%
  \BibitemOpen
  \bibfield  {author} {\bibinfo {author} {\bibfnamefont {M.~P.~A.}\
  \bibnamefont {Fisher}}, \bibinfo {author} {\bibfnamefont {V.}~\bibnamefont
  {Khemani}}, \bibinfo {author} {\bibfnamefont {A.}~\bibnamefont {Nahum}},\
  and\ \bibinfo {author} {\bibfnamefont {S.}~\bibnamefont {Vijay}},\ }\bibfield
   {title} {\bibinfo {title} {Random {Quantum} {Circuits}},\ }\href
  {https://doi.org/10.1146/annurev-conmatphys-031720-030658} {\bibfield
  {journal} {\bibinfo  {journal} {Annual Review of Condensed Matter Physics}\
  }\textbf {\bibinfo {volume} {14}},\ \bibinfo {pages} {335} (\bibinfo {year}
  {2023})}\BibitemShut {NoStop}%
\bibitem [{\citenamefont {Potter}\ and\ \citenamefont
  {Vasseur}(2022)}]{Potter2022entanglement}%
  \BibitemOpen
  \bibfield  {author} {\bibinfo {author} {\bibfnamefont {A.~C.}\ \bibnamefont
  {Potter}}\ and\ \bibinfo {author} {\bibfnamefont {R.}~\bibnamefont
  {Vasseur}},\ }\bibfield  {title} {\bibinfo {title} {Entanglement {Dynamics}
  in {Hybrid} {Quantum} {Circuits}},\ }in\ \href
  {https://doi.org/10.1007/978-3-031-03998-0_9} {\emph {\bibinfo {booktitle}
  {Entanglement in {Spin} {Chains}: {From} {Theory} to {Quantum} {Technology}
  {Applications}}}},\ \bibinfo {series and number} {Quantum {Science} and
  {Technology}},\ \bibinfo {editor} {edited by\ \bibinfo {editor}
  {\bibfnamefont {A.}~\bibnamefont {Bayat}}, \bibinfo {editor} {\bibfnamefont
  {S.}~\bibnamefont {Bose}},\ and\ \bibinfo {editor} {\bibfnamefont
  {H.}~\bibnamefont {Johannesson}}}\ (\bibinfo {address} {Cham},\ \bibinfo
  {year} {2022})\ pp.\ \bibinfo {pages} {211--249}\BibitemShut {NoStop}%
\bibitem [{\citenamefont {Skinner}\ \emph {et~al.}(2019)\citenamefont
  {Skinner}, \citenamefont {Ruhman},\ and\ \citenamefont
  {Nahum}}]{Skinner2019measurement}%
  \BibitemOpen
  \bibfield  {author} {\bibinfo {author} {\bibfnamefont {B.}~\bibnamefont
  {Skinner}}, \bibinfo {author} {\bibfnamefont {J.}~\bibnamefont {Ruhman}},\
  and\ \bibinfo {author} {\bibfnamefont {A.}~\bibnamefont {Nahum}},\ }\bibfield
   {title} {\bibinfo {title} {Measurement-induced phase transitions in the
  dynamics of entanglement},\ }\href
  {https://doi.org/10.1103/PhysRevX.9.031009} {\bibfield  {journal} {\bibinfo
  {journal} {Phys. Rev. X}\ }\textbf {\bibinfo {volume} {9}},\ \bibinfo {pages}
  {031009} (\bibinfo {year} {2019})}\BibitemShut {NoStop}%
\bibitem [{\citenamefont {Li}\ \emph {et~al.}(2018)\citenamefont {Li},
  \citenamefont {Chen},\ and\ \citenamefont {Fisher}}]{li_quantum_2018}%
  \BibitemOpen
  \bibfield  {author} {\bibinfo {author} {\bibfnamefont {Y.}~\bibnamefont
  {Li}}, \bibinfo {author} {\bibfnamefont {X.}~\bibnamefont {Chen}},\ and\
  \bibinfo {author} {\bibfnamefont {M.~P.~A.}\ \bibnamefont {Fisher}},\
  }\bibfield  {title} {\bibinfo {title} {Quantum {Zeno} effect and the
  many-body entanglement transition},\ }\href
  {https://doi.org/10.1103/PhysRevB.98.205136} {\bibfield  {journal} {\bibinfo
  {journal} {Physical Review B}\ }\textbf {\bibinfo {volume} {98}},\ \bibinfo
  {pages} {205136} (\bibinfo {year} {2018})}\BibitemShut {NoStop}%
\bibitem [{\citenamefont {Choi}\ \emph {et~al.}(2020)\citenamefont {Choi},
  \citenamefont {Bao}, \citenamefont {Qi},\ and\ \citenamefont
  {Altman}}]{Choi2020quantum}%
  \BibitemOpen
  \bibfield  {author} {\bibinfo {author} {\bibfnamefont {S.}~\bibnamefont
  {Choi}}, \bibinfo {author} {\bibfnamefont {Y.}~\bibnamefont {Bao}}, \bibinfo
  {author} {\bibfnamefont {X.-L.}\ \bibnamefont {Qi}},\ and\ \bibinfo {author}
  {\bibfnamefont {E.}~\bibnamefont {Altman}},\ }\bibfield  {title} {\bibinfo
  {title} {Quantum error correction in scrambling dynamics and
  measurement-induced phase transition},\ }\href
  {https://doi.org/10.1103/PhysRevLett.125.030505} {\bibfield  {journal}
  {\bibinfo  {journal} {Phys. Rev. Lett.}\ }\textbf {\bibinfo {volume} {125}},\
  \bibinfo {pages} {030505} (\bibinfo {year} {2020})}\BibitemShut {NoStop}%
\bibitem [{\citenamefont {Bao}\ \emph {et~al.}(2020)\citenamefont {Bao},
  \citenamefont {Choi},\ and\ \citenamefont {Altman}}]{Bao2020theory}%
  \BibitemOpen
  \bibfield  {author} {\bibinfo {author} {\bibfnamefont {Y.}~\bibnamefont
  {Bao}}, \bibinfo {author} {\bibfnamefont {S.}~\bibnamefont {Choi}},\ and\
  \bibinfo {author} {\bibfnamefont {E.}~\bibnamefont {Altman}},\ }\bibfield
  {title} {\bibinfo {title} {Theory of the phase transition in random unitary
  circuits with measurements},\ }\href
  {https://doi.org/10.1103/PhysRevB.101.104301} {\bibfield  {journal} {\bibinfo
   {journal} {Phys. Rev. B}\ }\textbf {\bibinfo {volume} {101}},\ \bibinfo
  {pages} {104301} (\bibinfo {year} {2020})}\BibitemShut {NoStop}%
\bibitem [{\citenamefont {Gullans}\ and\ \citenamefont
  {Huse}(2020{\natexlab{a}})}]{Gullans2020dynamical}%
  \BibitemOpen
  \bibfield  {author} {\bibinfo {author} {\bibfnamefont {M.~J.}\ \bibnamefont
  {Gullans}}\ and\ \bibinfo {author} {\bibfnamefont {D.~A.}\ \bibnamefont
  {Huse}},\ }\bibfield  {title} {\bibinfo {title} {Dynamical purification phase
  transition induced by quantum measurements},\ }\href
  {https://doi.org/10.1103/PhysRevX.10.041020} {\bibfield  {journal} {\bibinfo
  {journal} {Phys. Rev. X}\ }\textbf {\bibinfo {volume} {10}},\ \bibinfo
  {pages} {041020} (\bibinfo {year} {2020}{\natexlab{a}})}\BibitemShut
  {NoStop}%
\bibitem [{\citenamefont {Gullans}\ and\ \citenamefont
  {Huse}(2020{\natexlab{b}})}]{Gullans2020scalable}%
  \BibitemOpen
  \bibfield  {author} {\bibinfo {author} {\bibfnamefont {M.~J.}\ \bibnamefont
  {Gullans}}\ and\ \bibinfo {author} {\bibfnamefont {D.~A.}\ \bibnamefont
  {Huse}},\ }\bibfield  {title} {\bibinfo {title} {Scalable probes of
  measurement-induced criticality},\ }\href
  {https://doi.org/10.1103/PhysRevLett.125.070606} {\bibfield  {journal}
  {\bibinfo  {journal} {Phys. Rev. Lett.}\ }\textbf {\bibinfo {volume} {125}},\
  \bibinfo {pages} {070606} (\bibinfo {year} {2020}{\natexlab{b}})}\BibitemShut
  {NoStop}%
\bibitem [{\citenamefont {Fan}\ \emph {et~al.}(2021)\citenamefont {Fan},
  \citenamefont {Vijay}, \citenamefont {Vishwanath},\ and\ \citenamefont
  {You}}]{fan_self-organized_2021}%
  \BibitemOpen
  \bibfield  {author} {\bibinfo {author} {\bibfnamefont {R.}~\bibnamefont
  {Fan}}, \bibinfo {author} {\bibfnamefont {S.}~\bibnamefont {Vijay}}, \bibinfo
  {author} {\bibfnamefont {A.}~\bibnamefont {Vishwanath}},\ and\ \bibinfo
  {author} {\bibfnamefont {Y.-Z.}\ \bibnamefont {You}},\ }\bibfield  {title}
  {\bibinfo {title} {Self-organized error correction in random unitary circuits
  with measurement},\ }\href {https://doi.org/10.1103/PhysRevB.103.174309}
  {\bibfield  {journal} {\bibinfo  {journal} {Phys. Rev. B}\ }\textbf {\bibinfo
  {volume} {103}},\ \bibinfo {pages} {174309} (\bibinfo {year}
  {2021})}\BibitemShut {NoStop}%
\bibitem [{\citenamefont {Li}\ and\ \citenamefont
  {Fisher}(2021)}]{li_statistical_2021}%
  \BibitemOpen
  \bibfield  {author} {\bibinfo {author} {\bibfnamefont {Y.}~\bibnamefont
  {Li}}\ and\ \bibinfo {author} {\bibfnamefont {M.~P.~A.}\ \bibnamefont
  {Fisher}},\ }\bibfield  {title} {\bibinfo {title} {Statistical mechanics of
  quantum error correcting codes},\ }\href
  {https://doi.org/10.1103/PhysRevB.103.104306} {\bibfield  {journal} {\bibinfo
   {journal} {Physical Review B}\ }\textbf {\bibinfo {volume} {103}},\ \bibinfo
  {pages} {104306} (\bibinfo {year} {2021})}\BibitemShut {NoStop}%
\bibitem [{\citenamefont {Li}\ \emph {et~al.}(2023)\citenamefont {Li},
  \citenamefont {Zou}, \citenamefont {Glorioso}, \citenamefont {Altman},\ and\
  \citenamefont {Fisher}}]{li_cross_2023}%
  \BibitemOpen
  \bibfield  {author} {\bibinfo {author} {\bibfnamefont {Y.}~\bibnamefont
  {Li}}, \bibinfo {author} {\bibfnamefont {Y.}~\bibnamefont {Zou}}, \bibinfo
  {author} {\bibfnamefont {P.}~\bibnamefont {Glorioso}}, \bibinfo {author}
  {\bibfnamefont {E.}~\bibnamefont {Altman}},\ and\ \bibinfo {author}
  {\bibfnamefont {M.~P.~A.}\ \bibnamefont {Fisher}},\ }\bibfield  {title}
  {\bibinfo {title} {Cross {Entropy} {Benchmark} for {Measurement}-{Induced}
  {Phase} {Transitions}},\ }\href
  {https://doi.org/10.1103/PhysRevLett.130.220404} {\bibfield  {journal}
  {\bibinfo  {journal} {Physical Review Letters}\ }\textbf {\bibinfo {volume}
  {130}},\ \bibinfo {pages} {220404} (\bibinfo {year} {2023})}\BibitemShut
  {NoStop}%
\bibitem [{\citenamefont {Noel}\ \emph {et~al.}(2022)\citenamefont {Noel},
  \citenamefont {Niroula}, \citenamefont {Zhu}, \citenamefont {Risinger},
  \citenamefont {Egan}, \citenamefont {Biswas}, \citenamefont {Cetina},
  \citenamefont {Gorshkov}, \citenamefont {Gullans}, \citenamefont {Huse},\
  and\ \citenamefont {Monroe}}]{noel_measurement-induced_2022}%
  \BibitemOpen
  \bibfield  {author} {\bibinfo {author} {\bibfnamefont {C.}~\bibnamefont
  {Noel}}, \bibinfo {author} {\bibfnamefont {P.}~\bibnamefont {Niroula}},
  \bibinfo {author} {\bibfnamefont {D.}~\bibnamefont {Zhu}}, \bibinfo {author}
  {\bibfnamefont {A.}~\bibnamefont {Risinger}}, \bibinfo {author}
  {\bibfnamefont {L.}~\bibnamefont {Egan}}, \bibinfo {author} {\bibfnamefont
  {D.}~\bibnamefont {Biswas}}, \bibinfo {author} {\bibfnamefont
  {M.}~\bibnamefont {Cetina}}, \bibinfo {author} {\bibfnamefont {A.~V.}\
  \bibnamefont {Gorshkov}}, \bibinfo {author} {\bibfnamefont {M.~J.}\
  \bibnamefont {Gullans}}, \bibinfo {author} {\bibfnamefont {D.~A.}\
  \bibnamefont {Huse}},\ and\ \bibinfo {author} {\bibfnamefont
  {C.}~\bibnamefont {Monroe}},\ }\bibfield  {title} {\bibinfo {title}
  {Measurement-induced quantum phases realized in a trapped-ion quantum
  computer},\ }\href {https://doi.org/10.1038/s41567-022-01619-7} {\bibfield
  {journal} {\bibinfo  {journal} {Nature Physics}\ }\textbf {\bibinfo {volume}
  {18}},\ \bibinfo {pages} {760} (\bibinfo {year} {2022})}\BibitemShut
  {NoStop}%
\bibitem [{\citenamefont {Hoke}\ \emph {et~al.}(2023)\citenamefont {Hoke},
  \citenamefont {Ippoliti}, \citenamefont {Rosenberg}, \citenamefont {Abanin},
  \citenamefont {Acharya}, \citenamefont {Andersen}, \citenamefont {Ansmann},
  \citenamefont {Arute}, \citenamefont {Arya}, \citenamefont {Asfaw},
  \citenamefont {Atalaya}, \citenamefont {Bardin}, \citenamefont {Bengtsson},
  \citenamefont {Bortoli}, \citenamefont {Bourassa}, \citenamefont {Bovaird}
  \emph {et~al.}}]{hoke_measurement-induced_2023}%
  \BibitemOpen
  \bibfield  {author} {\bibinfo {author} {\bibfnamefont {J.~C.}\ \bibnamefont
  {Hoke}}, \bibinfo {author} {\bibfnamefont {M.}~\bibnamefont {Ippoliti}},
  \bibinfo {author} {\bibfnamefont {E.}~\bibnamefont {Rosenberg}}, \bibinfo
  {author} {\bibfnamefont {D.}~\bibnamefont {Abanin}}, \bibinfo {author}
  {\bibfnamefont {R.}~\bibnamefont {Acharya}}, \bibinfo {author} {\bibfnamefont
  {T.~I.}\ \bibnamefont {Andersen}}, \bibinfo {author} {\bibfnamefont
  {M.}~\bibnamefont {Ansmann}}, \bibinfo {author} {\bibfnamefont
  {F.}~\bibnamefont {Arute}}, \bibinfo {author} {\bibfnamefont
  {K.}~\bibnamefont {Arya}}, \bibinfo {author} {\bibfnamefont {A.}~\bibnamefont
  {Asfaw}}, \bibinfo {author} {\bibfnamefont {J.}~\bibnamefont {Atalaya}},
  \bibinfo {author} {\bibfnamefont {J.~C.}\ \bibnamefont {Bardin}}, \bibinfo
  {author} {\bibfnamefont {A.}~\bibnamefont {Bengtsson}}, \bibinfo {author}
  {\bibfnamefont {G.}~\bibnamefont {Bortoli}}, \bibinfo {author} {\bibfnamefont
  {A.}~\bibnamefont {Bourassa}}, \bibinfo {author} {\bibfnamefont
  {J.}~\bibnamefont {Bovaird}}, \emph {et~al.},\ }\bibfield  {title} {\bibinfo
  {title} {Measurement-induced entanglement and teleportation on a noisy
  quantum processor},\ }\href {https://doi.org/10.1038/s41586-023-06505-7}
  {\bibfield  {journal} {\bibinfo  {journal} {Nature}\ }\textbf {\bibinfo
  {volume} {622}},\ \bibinfo {pages} {481} (\bibinfo {year}
  {2023})}\BibitemShut {NoStop}%
\bibitem [{\citenamefont {Vovk}\ and\ \citenamefont
  {Pichler}(2022)}]{vovk_entanglement-optimal_2022}%
  \BibitemOpen
  \bibfield  {author} {\bibinfo {author} {\bibfnamefont {T.}~\bibnamefont
  {Vovk}}\ and\ \bibinfo {author} {\bibfnamefont {H.}~\bibnamefont {Pichler}},\
  }\bibfield  {title} {\bibinfo {title} {Entanglement-optimal trajectories of
  many-body quantum markov processes},\ }\href
  {https://doi.org/10.1103/PhysRevLett.128.243601} {\bibfield  {journal}
  {\bibinfo  {journal} {Phys. Rev. Lett.}\ }\textbf {\bibinfo {volume} {128}},\
  \bibinfo {pages} {243601} (\bibinfo {year} {2022})}\BibitemShut {NoStop}%
\bibitem [{\citenamefont {Cheng}\ and\ \citenamefont
  {Ippoliti}(2023)}]{Cheng2023efficient}%
  \BibitemOpen
  \bibfield  {author} {\bibinfo {author} {\bibfnamefont {Z.}~\bibnamefont
  {Cheng}}\ and\ \bibinfo {author} {\bibfnamefont {M.}~\bibnamefont
  {Ippoliti}},\ }\bibfield  {title} {\bibinfo {title} {Efficient sampling of
  noisy shallow circuits via monitored unraveling},\ }\href
  {https://doi.org/10.1103/PRXQuantum.4.040326} {\bibfield  {journal} {\bibinfo
   {journal} {PRX Quantum}\ }\textbf {\bibinfo {volume} {4}},\ \bibinfo {pages}
  {040326} (\bibinfo {year} {2023})}\BibitemShut {NoStop}%
\bibitem [{\citenamefont {Chen}\ \emph {et~al.}(2024)\citenamefont {Chen},
  \citenamefont {Bao},\ and\ \citenamefont {Choi}}]{chen_optimized_2023}%
  \BibitemOpen
  \bibfield  {author} {\bibinfo {author} {\bibfnamefont {Z.}~\bibnamefont
  {Chen}}, \bibinfo {author} {\bibfnamefont {Y.}~\bibnamefont {Bao}},\ and\
  \bibinfo {author} {\bibfnamefont {S.}~\bibnamefont {Choi}},\ }\bibfield
  {title} {\bibinfo {title} {Optimized trajectory unraveling for classical
  simulation of noisy quantum dynamics},\ }\href
  {https://doi.org/10.1103/PhysRevLett.133.230403} {\bibfield  {journal}
  {\bibinfo  {journal} {Phys. Rev. Lett.}\ }\textbf {\bibinfo {volume} {133}},\
  \bibinfo {pages} {230403} (\bibinfo {year} {2024})}\BibitemShut {NoStop}%
\bibitem [{\citenamefont {Napp}\ \emph {et~al.}(2022)\citenamefont {Napp},
  \citenamefont {La~Placa}, \citenamefont {Dalzell}, \citenamefont
  {Brand\~ao},\ and\ \citenamefont {Harrow}}]{Napp2022efficient}%
  \BibitemOpen
  \bibfield  {author} {\bibinfo {author} {\bibfnamefont {J.~C.}\ \bibnamefont
  {Napp}}, \bibinfo {author} {\bibfnamefont {R.~L.}\ \bibnamefont {La~Placa}},
  \bibinfo {author} {\bibfnamefont {A.~M.}\ \bibnamefont {Dalzell}}, \bibinfo
  {author} {\bibfnamefont {F.~G. S.~L.}\ \bibnamefont {Brand\~ao}},\ and\
  \bibinfo {author} {\bibfnamefont {A.~W.}\ \bibnamefont {Harrow}},\ }\bibfield
   {title} {\bibinfo {title} {Efficient classical simulation of random shallow
  2d quantum circuits},\ }\href {https://doi.org/10.1103/PhysRevX.12.021021}
  {\bibfield  {journal} {\bibinfo  {journal} {Phys. Rev. X}\ }\textbf {\bibinfo
  {volume} {12}},\ \bibinfo {pages} {021021} (\bibinfo {year}
  {2022})}\BibitemShut {NoStop}%
\bibitem [{\citenamefont {Venn}\ and\ \citenamefont
  {B\'eri}(2020)}]{Venn2020error-correction}%
  \BibitemOpen
  \bibfield  {author} {\bibinfo {author} {\bibfnamefont {F.}~\bibnamefont
  {Venn}}\ and\ \bibinfo {author} {\bibfnamefont {B.}~\bibnamefont {B\'eri}},\
  }\bibfield  {title} {\bibinfo {title} {Error-correction and noise-decoherence
  thresholds for coherent errors in planar-graph surface codes},\ }\href
  {https://doi.org/10.1103/PhysRevResearch.2.043412} {\bibfield  {journal}
  {\bibinfo  {journal} {Phys. Rev. Res.}\ }\textbf {\bibinfo {volume} {2}},\
  \bibinfo {pages} {043412} (\bibinfo {year} {2020})}\BibitemShut {NoStop}%
\bibitem [{\citenamefont {Darmawan}(2024)}]{darmawan_optimal_2024}%
  \BibitemOpen
  \bibfield  {author} {\bibinfo {author} {\bibfnamefont {A.~S.}\ \bibnamefont
  {Darmawan}},\ }\href {https://arxiv.org/abs/2403.08706} {\bibinfo {title}
  {Optimal adaptation of surface-code decoders to local noise}} (\bibinfo
  {year} {2024}),\ \Eprint {https://arxiv.org/abs/2403.08706} {arXiv:2403.08706
  [quant-ph]} \BibitemShut {NoStop}%
\bibitem [{\citenamefont {Dennis}\ \emph {et~al.}(2002)\citenamefont {Dennis},
  \citenamefont {Kitaev}, \citenamefont {Landahl},\ and\ \citenamefont
  {Preskill}}]{dennis_topological_2002}%
  \BibitemOpen
  \bibfield  {author} {\bibinfo {author} {\bibfnamefont {E.}~\bibnamefont
  {Dennis}}, \bibinfo {author} {\bibfnamefont {A.}~\bibnamefont {Kitaev}},
  \bibinfo {author} {\bibfnamefont {A.}~\bibnamefont {Landahl}},\ and\ \bibinfo
  {author} {\bibfnamefont {J.}~\bibnamefont {Preskill}},\ }\bibfield  {title}
  {\bibinfo {title} {Topological quantum memory},\ }\href
  {https://doi.org/10.1063/1.1499754} {\bibfield  {journal} {\bibinfo
  {journal} {Journal of Mathematical Physics}\ }\textbf {\bibinfo {volume}
  {43}},\ \bibinfo {pages} {4452} (\bibinfo {year} {2002})}\BibitemShut
  {NoStop}%
\bibitem [{\citenamefont {Fowler}(2012)}]{fowler_towards_2012}%
  \BibitemOpen
  \bibfield  {author} {\bibinfo {author} {\bibfnamefont {A.~G.}\ \bibnamefont
  {Fowler}},\ }\bibfield  {title} {\bibinfo {title} {Towards {Practical}
  {Classical} {Processing} for the {Surface} {Code}},\ }\bibfield  {journal}
  {\bibinfo  {journal} {Physical Review Letters}\ }\textbf {\bibinfo {volume}
  {108}},\ \href {https://doi.org/10.1103/PhysRevLett.108.180501}
  {10.1103/PhysRevLett.108.180501} (\bibinfo {year} {2012})\BibitemShut
  {NoStop}%
\bibitem [{\citenamefont {Javadi-Abhari}\ \emph {et~al.}(2024)\citenamefont
  {Javadi-Abhari}, \citenamefont {Treinish}, \citenamefont {Krsulich},
  \citenamefont {Wood}, \citenamefont {Lishman}, \citenamefont {Gacon},
  \citenamefont {Martiel}, \citenamefont {Nation}, \citenamefont {Bishop},
  \citenamefont {Cross}, \citenamefont {Johnson},\ and\ \citenamefont
  {Gambetta}}]{qiskit2024}%
  \BibitemOpen
  \bibfield  {author} {\bibinfo {author} {\bibfnamefont {A.}~\bibnamefont
  {Javadi-Abhari}}, \bibinfo {author} {\bibfnamefont {M.}~\bibnamefont
  {Treinish}}, \bibinfo {author} {\bibfnamefont {K.}~\bibnamefont {Krsulich}},
  \bibinfo {author} {\bibfnamefont {C.~J.}\ \bibnamefont {Wood}}, \bibinfo
  {author} {\bibfnamefont {J.}~\bibnamefont {Lishman}}, \bibinfo {author}
  {\bibfnamefont {J.}~\bibnamefont {Gacon}}, \bibinfo {author} {\bibfnamefont
  {S.}~\bibnamefont {Martiel}}, \bibinfo {author} {\bibfnamefont {P.~D.}\
  \bibnamefont {Nation}}, \bibinfo {author} {\bibfnamefont {L.~S.}\
  \bibnamefont {Bishop}}, \bibinfo {author} {\bibfnamefont {A.~W.}\
  \bibnamefont {Cross}}, \bibinfo {author} {\bibfnamefont {B.~R.}\ \bibnamefont
  {Johnson}},\ and\ \bibinfo {author} {\bibfnamefont {J.~M.}\ \bibnamefont
  {Gambetta}},\ }\href {https://doi.org/10.48550/arXiv.2405.08810} {\bibinfo
  {title} {Quantum computing with {Q}iskit}} (\bibinfo {year} {2024}),\ \Eprint
  {https://arxiv.org/abs/2405.08810} {arXiv:2405.08810 [quant-ph]} \BibitemShut
  {NoStop}%
\bibitem [{\citenamefont
  {Watrous}(2012)}]{watrous_simplersemidefiniteprogramscompletely_2012}%
  \BibitemOpen
  \bibfield  {author} {\bibinfo {author} {\bibfnamefont {J.}~\bibnamefont
  {Watrous}},\ }\href {https://arxiv.org/abs/1207.5726} {\bibinfo {title}
  {Simpler semidefinite programs for completely bounded norms}} (\bibinfo
  {year} {2012}),\ \Eprint {https://arxiv.org/abs/1207.5726} {arXiv:1207.5726
  [quant-ph]} \BibitemShut {NoStop}%
\bibitem [{\citenamefont {Brand{\~a}o}\ \emph {et~al.}(2016)\citenamefont
  {Brand{\~a}o}, \citenamefont {Harrow},\ and\ \citenamefont
  {Horodecki}}]{Brandao2016logical}%
  \BibitemOpen
  \bibfield  {author} {\bibinfo {author} {\bibfnamefont {F.~G. S.~L.}\
  \bibnamefont {Brand{\~a}o}}, \bibinfo {author} {\bibfnamefont {A.~W.}\
  \bibnamefont {Harrow}},\ and\ \bibinfo {author} {\bibfnamefont
  {M.}~\bibnamefont {Horodecki}},\ }\bibfield  {title} {\bibinfo {title} {Local
  random quantum circuits are approximate polynomial-designs},\ }\href
  {https://doi.org/10.1007/s00220-016-2706-8} {\bibfield  {journal} {\bibinfo
  {journal} {Communications in Mathematical Physics}\ }\textbf {\bibinfo
  {volume} {346}},\ \bibinfo {pages} {397} (\bibinfo {year}
  {2016})}\BibitemShut {NoStop}%
\bibitem [{\citenamefont {Stace}\ and\ \citenamefont
  {Barrett}(2010)}]{stace2010error}%
  \BibitemOpen
  \bibfield  {author} {\bibinfo {author} {\bibfnamefont {T.~M.}\ \bibnamefont
  {Stace}}\ and\ \bibinfo {author} {\bibfnamefont {S.~D.}\ \bibnamefont
  {Barrett}},\ }\bibfield  {title} {\bibinfo {title} {Error correction and
  degeneracy in surface codes suffering loss},\ }\href
  {https://doi.org/10.1103/PhysRevA.81.022317} {\bibfield  {journal} {\bibinfo
  {journal} {Phys. Rev. A}\ }\textbf {\bibinfo {volume} {81}},\ \bibinfo
  {pages} {022317} (\bibinfo {year} {2010})}\BibitemShut {NoStop}%
\bibitem [{\citenamefont {Behrends}\ and\ \citenamefont
  {Béri}(2025)}]{behrends_surfacecode_2025}%
  \BibitemOpen
  \bibfield  {author} {\bibinfo {author} {\bibfnamefont {J.}~\bibnamefont
  {Behrends}}\ and\ \bibinfo {author} {\bibfnamefont {B.}~\bibnamefont
  {Béri}},\ }\href {https://arxiv.org/abs/2412.21055} {\bibinfo {title} {The
  surface code beyond pauli channels: Logical noise coherence,
  information-theoretic measures, and errorfield-double phenomenology}}
  (\bibinfo {year} {2025}),\ \Eprint {https://arxiv.org/abs/2412.21055}
  {arXiv:2412.21055 [quant-ph]} \BibitemShut {NoStop}%
\bibitem [{\citenamefont {Suzuki}\ \emph {et~al.}(2017)\citenamefont {Suzuki},
  \citenamefont {Fujii},\ and\ \citenamefont {Koashi}}]{suzuki_efficient_2017}%
  \BibitemOpen
  \bibfield  {author} {\bibinfo {author} {\bibfnamefont {Y.}~\bibnamefont
  {Suzuki}}, \bibinfo {author} {\bibfnamefont {K.}~\bibnamefont {Fujii}},\ and\
  \bibinfo {author} {\bibfnamefont {M.}~\bibnamefont {Koashi}},\ }\bibfield
  {title} {\bibinfo {title} {Efficient simulation of quantum error correction
  under coherent error based on the nonunitary free-fermionic formalism},\
  }\href {https://doi.org/10.1103/PhysRevLett.119.190503} {\bibfield  {journal}
  {\bibinfo  {journal} {Phys. Rev. Lett.}\ }\textbf {\bibinfo {volume} {119}},\
  \bibinfo {pages} {190503} (\bibinfo {year} {2017})}\BibitemShut {NoStop}%
\bibitem [{\citenamefont {Satzinger}\ \emph {et~al.}(2021)\citenamefont
  {Satzinger}, \citenamefont {Liu}, \citenamefont {Smith}, \citenamefont
  {Knapp}, \citenamefont {Newman}, \citenamefont {Jones}, \citenamefont {Chen},
  \citenamefont {Quintana}, \citenamefont {Mi}, \citenamefont {Dunsworth},
  \citenamefont {Gidney}, \citenamefont {Aleiner}, \citenamefont {Arute},
  \citenamefont {Arya}, \citenamefont {Atalaya}, \citenamefont {Babbush} \emph
  {et~al.}}]{Satzinger2021realizing}%
  \BibitemOpen
  \bibfield  {author} {\bibinfo {author} {\bibfnamefont {K.~J.}\ \bibnamefont
  {Satzinger}}, \bibinfo {author} {\bibfnamefont {Y.-J.}\ \bibnamefont {Liu}},
  \bibinfo {author} {\bibfnamefont {A.}~\bibnamefont {Smith}}, \bibinfo
  {author} {\bibfnamefont {C.}~\bibnamefont {Knapp}}, \bibinfo {author}
  {\bibfnamefont {M.}~\bibnamefont {Newman}}, \bibinfo {author} {\bibfnamefont
  {C.}~\bibnamefont {Jones}}, \bibinfo {author} {\bibfnamefont
  {Z.}~\bibnamefont {Chen}}, \bibinfo {author} {\bibfnamefont {C.}~\bibnamefont
  {Quintana}}, \bibinfo {author} {\bibfnamefont {X.}~\bibnamefont {Mi}},
  \bibinfo {author} {\bibfnamefont {A.}~\bibnamefont {Dunsworth}}, \bibinfo
  {author} {\bibfnamefont {C.}~\bibnamefont {Gidney}}, \bibinfo {author}
  {\bibfnamefont {I.}~\bibnamefont {Aleiner}}, \bibinfo {author} {\bibfnamefont
  {F.}~\bibnamefont {Arute}}, \bibinfo {author} {\bibfnamefont
  {K.}~\bibnamefont {Arya}}, \bibinfo {author} {\bibfnamefont {J.}~\bibnamefont
  {Atalaya}}, \bibinfo {author} {\bibfnamefont {R.}~\bibnamefont {Babbush}},
  \emph {et~al.},\ }\bibfield  {title} {\bibinfo {title} {Realizing
  topologically ordered states on a quantum processor},\ }\href
  {https://doi.org/10.1126/science.abi8378} {\bibfield  {journal} {\bibinfo
  {journal} {Science}\ }\textbf {\bibinfo {volume} {374}},\ \bibinfo {pages}
  {1237} (\bibinfo {year} {2021})},\ \Eprint
  {https://arxiv.org/abs/https://www.science.org/doi/pdf/10.1126/science.abi8378}
  {https://www.science.org/doi/pdf/10.1126/science.abi8378} \BibitemShut
  {NoStop}%
\bibitem [{\citenamefont {Higgott}\ \emph {et~al.}(2021)\citenamefont
  {Higgott}, \citenamefont {Wilson}, \citenamefont {Hefford}, \citenamefont
  {Dborin}, \citenamefont {Hanif}, \citenamefont {Burton},\ and\ \citenamefont
  {Browne}}]{Higgott2021optimallocalunitary}%
  \BibitemOpen
  \bibfield  {author} {\bibinfo {author} {\bibfnamefont {O.}~\bibnamefont
  {Higgott}}, \bibinfo {author} {\bibfnamefont {M.}~\bibnamefont {Wilson}},
  \bibinfo {author} {\bibfnamefont {J.}~\bibnamefont {Hefford}}, \bibinfo
  {author} {\bibfnamefont {J.}~\bibnamefont {Dborin}}, \bibinfo {author}
  {\bibfnamefont {F.}~\bibnamefont {Hanif}}, \bibinfo {author} {\bibfnamefont
  {S.}~\bibnamefont {Burton}},\ and\ \bibinfo {author} {\bibfnamefont {D.~E.}\
  \bibnamefont {Browne}},\ }\bibfield  {title} {\bibinfo {title} {Optimal local
  unitary encoding circuits for the surface code},\ }\href
  {https://doi.org/10.22331/q-2021-08-05-517} {\bibfield  {journal} {\bibinfo
  {journal} {{Quantum}}\ }\textbf {\bibinfo {volume} {5}},\ \bibinfo {pages}
  {517} (\bibinfo {year} {2021})}\BibitemShut {NoStop}%
\bibitem [{\citenamefont {Iqbal}\ \emph {et~al.}(2024)\citenamefont {Iqbal},
  \citenamefont {Tantivasadakarn}, \citenamefont {Gatterman}, \citenamefont
  {Gerber}, \citenamefont {Gilmore}, \citenamefont {Gresh}, \citenamefont
  {Hankin}, \citenamefont {Hewitt}, \citenamefont {Horst}, \citenamefont
  {Matheny}, \citenamefont {Mengle}, \citenamefont {Neyenhuis}, \citenamefont
  {Vishwanath}, \citenamefont {Foss-Feig}, \citenamefont {Verresen},\ and\
  \citenamefont {Dreyer}}]{iqbal_topological_2024}%
  \BibitemOpen
  \bibfield  {author} {\bibinfo {author} {\bibfnamefont {M.}~\bibnamefont
  {Iqbal}}, \bibinfo {author} {\bibfnamefont {N.}~\bibnamefont
  {Tantivasadakarn}}, \bibinfo {author} {\bibfnamefont {T.~M.}\ \bibnamefont
  {Gatterman}}, \bibinfo {author} {\bibfnamefont {J.~A.}\ \bibnamefont
  {Gerber}}, \bibinfo {author} {\bibfnamefont {K.}~\bibnamefont {Gilmore}},
  \bibinfo {author} {\bibfnamefont {D.}~\bibnamefont {Gresh}}, \bibinfo
  {author} {\bibfnamefont {A.}~\bibnamefont {Hankin}}, \bibinfo {author}
  {\bibfnamefont {N.}~\bibnamefont {Hewitt}}, \bibinfo {author} {\bibfnamefont
  {C.~V.}\ \bibnamefont {Horst}}, \bibinfo {author} {\bibfnamefont
  {M.}~\bibnamefont {Matheny}}, \bibinfo {author} {\bibfnamefont
  {T.}~\bibnamefont {Mengle}}, \bibinfo {author} {\bibfnamefont
  {B.}~\bibnamefont {Neyenhuis}}, \bibinfo {author} {\bibfnamefont
  {A.}~\bibnamefont {Vishwanath}}, \bibinfo {author} {\bibfnamefont
  {M.}~\bibnamefont {Foss-Feig}}, \bibinfo {author} {\bibfnamefont
  {R.}~\bibnamefont {Verresen}},\ and\ \bibinfo {author} {\bibfnamefont
  {H.}~\bibnamefont {Dreyer}},\ }\bibfield  {title} {\bibinfo {title}
  {Topological order from measurements and feed-forward on a trapped ion
  quantum computer},\ }\href {https://doi.org/10.1038/s42005-024-01698-3}
  {\bibfield  {journal} {\bibinfo  {journal} {Communications Physics}\ }\textbf
  {\bibinfo {volume} {7}},\ \bibinfo {pages} {1} (\bibinfo {year}
  {2024})}\BibitemShut {NoStop}%
\bibitem [{\citenamefont {Foss-Feig}\ \emph {et~al.}(2023)\citenamefont
  {Foss-Feig}, \citenamefont {Tikku}, \citenamefont {Lu}, \citenamefont
  {Mayer}, \citenamefont {Iqbal}, \citenamefont {Gatterman}, \citenamefont
  {Gerber}, \citenamefont {Gilmore}, \citenamefont {Gresh}, \citenamefont
  {Hankin}, \citenamefont {Hewitt}, \citenamefont {Horst}, \citenamefont
  {Matheny}, \citenamefont {Mengle}, \citenamefont {Neyenhuis}, \citenamefont
  {Dreyer} \emph {et~al.}}]{foss-feig_experimental_2023}%
  \BibitemOpen
  \bibfield  {author} {\bibinfo {author} {\bibfnamefont {M.}~\bibnamefont
  {Foss-Feig}}, \bibinfo {author} {\bibfnamefont {A.}~\bibnamefont {Tikku}},
  \bibinfo {author} {\bibfnamefont {T.-C.}\ \bibnamefont {Lu}}, \bibinfo
  {author} {\bibfnamefont {K.}~\bibnamefont {Mayer}}, \bibinfo {author}
  {\bibfnamefont {M.}~\bibnamefont {Iqbal}}, \bibinfo {author} {\bibfnamefont
  {T.~M.}\ \bibnamefont {Gatterman}}, \bibinfo {author} {\bibfnamefont {J.~A.}\
  \bibnamefont {Gerber}}, \bibinfo {author} {\bibfnamefont {K.}~\bibnamefont
  {Gilmore}}, \bibinfo {author} {\bibfnamefont {D.}~\bibnamefont {Gresh}},
  \bibinfo {author} {\bibfnamefont {A.}~\bibnamefont {Hankin}}, \bibinfo
  {author} {\bibfnamefont {N.}~\bibnamefont {Hewitt}}, \bibinfo {author}
  {\bibfnamefont {C.~V.}\ \bibnamefont {Horst}}, \bibinfo {author}
  {\bibfnamefont {M.}~\bibnamefont {Matheny}}, \bibinfo {author} {\bibfnamefont
  {T.}~\bibnamefont {Mengle}}, \bibinfo {author} {\bibfnamefont
  {B.}~\bibnamefont {Neyenhuis}}, \bibinfo {author} {\bibfnamefont
  {H.}~\bibnamefont {Dreyer}}, \emph {et~al.},\ }\bibfield  {title} {\bibinfo
  {title} {Experimental demonstration of the advantage of adaptive quantum
  circuits},\ }\bibfield  {journal} {\bibinfo  {journal} {arXiv:2302.03029}\
  }\href {https://doi.org/10.48550/arXiv.2302.03029}
  {10.48550/arXiv.2302.03029} (\bibinfo {year} {2023})\BibitemShut {NoStop}%
\bibitem [{\citenamefont {Tillich}\ and\ \citenamefont
  {Zemor}(2014)}]{Tillich2014quantum}%
  \BibitemOpen
  \bibfield  {author} {\bibinfo {author} {\bibfnamefont {J.-P.}\ \bibnamefont
  {Tillich}}\ and\ \bibinfo {author} {\bibfnamefont {G.}~\bibnamefont
  {Zemor}},\ }\bibfield  {title} {\bibinfo {title} {Quantum ldpc codes with
  positive rate and minimum distance proportional to the square root of the
  blocklength},\ }\href {https://doi.org/10.1109/TIT.2013.2292061} {\bibfield
  {journal} {\bibinfo  {journal} {IEEE Trans. Inf. Theor.}\ }\textbf {\bibinfo
  {volume} {60}},\ \bibinfo {pages} {1193–1202} (\bibinfo {year}
  {2014})}\BibitemShut {NoStop}%
\bibitem [{\citenamefont {Foss-Feig}\ \emph {et~al.}(2021)\citenamefont
  {Foss-Feig}, \citenamefont {Hayes}, \citenamefont {Dreiling}, \citenamefont
  {Figgatt}, \citenamefont {Gaebler}, \citenamefont {Moses}, \citenamefont
  {Pino},\ and\ \citenamefont {Potter}}]{foss-feig_holographic_2021}%
  \BibitemOpen
  \bibfield  {author} {\bibinfo {author} {\bibfnamefont {M.}~\bibnamefont
  {Foss-Feig}}, \bibinfo {author} {\bibfnamefont {D.}~\bibnamefont {Hayes}},
  \bibinfo {author} {\bibfnamefont {J.~M.}\ \bibnamefont {Dreiling}}, \bibinfo
  {author} {\bibfnamefont {C.}~\bibnamefont {Figgatt}}, \bibinfo {author}
  {\bibfnamefont {J.~P.}\ \bibnamefont {Gaebler}}, \bibinfo {author}
  {\bibfnamefont {S.~A.}\ \bibnamefont {Moses}}, \bibinfo {author}
  {\bibfnamefont {J.~M.}\ \bibnamefont {Pino}},\ and\ \bibinfo {author}
  {\bibfnamefont {A.~C.}\ \bibnamefont {Potter}},\ }\bibfield  {title}
  {\bibinfo {title} {Holographic quantum algorithms for simulating correlated
  spin systems},\ }\href {https://doi.org/10.1103/PhysRevResearch.3.033002}
  {\bibfield  {journal} {\bibinfo  {journal} {Physical Review Research}\
  }\textbf {\bibinfo {volume} {3}},\ \bibinfo {pages} {033002} (\bibinfo {year}
  {2021})}\BibitemShut {NoStop}%
\bibitem [{\citenamefont {Ippoliti}\ \emph {et~al.}(2022)\citenamefont
  {Ippoliti}, \citenamefont {Rakovszky},\ and\ \citenamefont
  {Khemani}}]{ippoliti_fractal_2022}%
  \BibitemOpen
  \bibfield  {author} {\bibinfo {author} {\bibfnamefont {M.}~\bibnamefont
  {Ippoliti}}, \bibinfo {author} {\bibfnamefont {T.}~\bibnamefont
  {Rakovszky}},\ and\ \bibinfo {author} {\bibfnamefont {V.}~\bibnamefont
  {Khemani}},\ }\bibfield  {title} {\bibinfo {title} {Fractal, {Logarithmic},
  and {Volume}-{Law} {Entangled} {Nonthermal} {Steady} {States} via {Spacetime}
  {Duality}},\ }\href {https://doi.org/10.1103/PhysRevX.12.011045} {\bibfield
  {journal} {\bibinfo  {journal} {Physical Review X}\ }\textbf {\bibinfo
  {volume} {12}},\ \bibinfo {pages} {011045} (\bibinfo {year}
  {2022})}\BibitemShut {NoStop}%
\bibitem [{\citenamefont {Lu}\ and\ \citenamefont
  {Grover}(2021)}]{lu_spacetime_2021}%
  \BibitemOpen
  \bibfield  {author} {\bibinfo {author} {\bibfnamefont {T.-C.}\ \bibnamefont
  {Lu}}\ and\ \bibinfo {author} {\bibfnamefont {T.}~\bibnamefont {Grover}},\
  }\bibfield  {title} {\bibinfo {title} {Spacetime duality between localization
  transitions and measurement-induced transitions},\ }\href
  {https://doi.org/10.1103/PRXQuantum.2.040319} {\bibfield  {journal} {\bibinfo
   {journal} {PRX Quantum}\ }\textbf {\bibinfo {volume} {2}},\ \bibinfo {pages}
  {040319} (\bibinfo {year} {2021})}\BibitemShut {NoStop}%
\bibitem [{\citenamefont {Anand}\ \emph {et~al.}(2023)\citenamefont {Anand},
  \citenamefont {Hauschild}, \citenamefont {Zhang}, \citenamefont {Potter},\
  and\ \citenamefont {Zaletel}}]{anand_holographic_2023}%
  \BibitemOpen
  \bibfield  {author} {\bibinfo {author} {\bibfnamefont {S.}~\bibnamefont
  {Anand}}, \bibinfo {author} {\bibfnamefont {J.}~\bibnamefont {Hauschild}},
  \bibinfo {author} {\bibfnamefont {Y.}~\bibnamefont {Zhang}}, \bibinfo
  {author} {\bibfnamefont {A.~C.}\ \bibnamefont {Potter}},\ and\ \bibinfo
  {author} {\bibfnamefont {M.~P.}\ \bibnamefont {Zaletel}},\ }\bibfield
  {title} {\bibinfo {title} {Holographic {Quantum} {Simulation} of
  {Entanglement} {Renormalization} {Circuits}},\ }\href
  {https://doi.org/10.1103/PRXQuantum.4.030334} {\bibfield  {journal} {\bibinfo
   {journal} {PRX Quantum}\ }\textbf {\bibinfo {volume} {4}},\ \bibinfo {pages}
  {030334} (\bibinfo {year} {2023})}\BibitemShut {NoStop}%
\bibitem [{\citenamefont {Schuch}\ \emph {et~al.}(2008)\citenamefont {Schuch},
  \citenamefont {Wolf}, \citenamefont {Verstraete},\ and\ \citenamefont
  {Cirac}}]{Schuch2008entropy}%
  \BibitemOpen
  \bibfield  {author} {\bibinfo {author} {\bibfnamefont {N.}~\bibnamefont
  {Schuch}}, \bibinfo {author} {\bibfnamefont {M.~M.}\ \bibnamefont {Wolf}},
  \bibinfo {author} {\bibfnamefont {F.}~\bibnamefont {Verstraete}},\ and\
  \bibinfo {author} {\bibfnamefont {J.~I.}\ \bibnamefont {Cirac}},\ }\bibfield
  {title} {\bibinfo {title} {Entropy scaling and simulability by matrix product
  states},\ }\href {https://doi.org/10.1103/PhysRevLett.100.030504} {\bibfield
  {journal} {\bibinfo  {journal} {Phys. Rev. Lett.}\ }\textbf {\bibinfo
  {volume} {100}},\ \bibinfo {pages} {030504} (\bibinfo {year}
  {2008})}\BibitemShut {NoStop}%
\bibitem [{\citenamefont {Zabalo}\ \emph {et~al.}(2020)\citenamefont {Zabalo},
  \citenamefont {Gullans}, \citenamefont {Wilson}, \citenamefont
  {Gopalakrishnan}, \citenamefont {Huse},\ and\ \citenamefont
  {Pixley}}]{Zabalo2020critical}%
  \BibitemOpen
  \bibfield  {author} {\bibinfo {author} {\bibfnamefont {A.}~\bibnamefont
  {Zabalo}}, \bibinfo {author} {\bibfnamefont {M.~J.}\ \bibnamefont {Gullans}},
  \bibinfo {author} {\bibfnamefont {J.~H.}\ \bibnamefont {Wilson}}, \bibinfo
  {author} {\bibfnamefont {S.}~\bibnamefont {Gopalakrishnan}}, \bibinfo
  {author} {\bibfnamefont {D.~A.}\ \bibnamefont {Huse}},\ and\ \bibinfo
  {author} {\bibfnamefont {J.~H.}\ \bibnamefont {Pixley}},\ }\bibfield  {title}
  {\bibinfo {title} {Critical properties of the measurement-induced transition
  in random quantum circuits},\ }\href
  {https://doi.org/10.1103/PhysRevB.101.060301} {\bibfield  {journal} {\bibinfo
   {journal} {Phys. Rev. B}\ }\textbf {\bibinfo {volume} {101}},\ \bibinfo
  {pages} {060301(R)} (\bibinfo {year} {2020})}\BibitemShut {NoStop}%
\bibitem [{\citenamefont {D\"ur}\ \emph {et~al.}(2005)\citenamefont {D\"ur},
  \citenamefont {Hein}, \citenamefont {Cirac},\ and\ \citenamefont
  {Briegel}}]{dur_standard_2005}%
  \BibitemOpen
  \bibfield  {author} {\bibinfo {author} {\bibfnamefont {W.}~\bibnamefont
  {D\"ur}}, \bibinfo {author} {\bibfnamefont {M.}~\bibnamefont {Hein}},
  \bibinfo {author} {\bibfnamefont {J.~I.}\ \bibnamefont {Cirac}},\ and\
  \bibinfo {author} {\bibfnamefont {H.-J.}\ \bibnamefont {Briegel}},\
  }\bibfield  {title} {\bibinfo {title} {Standard forms of noisy quantum
  operations via depolarization},\ }\href
  {https://doi.org/10.1103/PhysRevA.72.052326} {\bibfield  {journal} {\bibinfo
  {journal} {Phys. Rev. A}\ }\textbf {\bibinfo {volume} {72}},\ \bibinfo
  {pages} {052326} (\bibinfo {year} {2005})}\BibitemShut {NoStop}%
\bibitem [{\citenamefont {Cai}\ and\ \citenamefont
  {Benjamin}(2019)}]{cai_constructing_2019}%
  \BibitemOpen
  \bibfield  {author} {\bibinfo {author} {\bibfnamefont {Z.}~\bibnamefont
  {Cai}}\ and\ \bibinfo {author} {\bibfnamefont {S.~C.}\ \bibnamefont
  {Benjamin}},\ }\bibfield  {title} {\bibinfo {title} {Constructing {Smaller}
  {Pauli} {Twirling} {Sets} for {Arbitrary} {Error} {Channels}},\ }\href
  {https://doi.org/10.1038/s41598-019-46722-7} {\bibfield  {journal} {\bibinfo
  {journal} {Scientific Reports}\ }\textbf {\bibinfo {volume} {9}},\ \bibinfo
  {pages} {11281} (\bibinfo {year} {2019})}\BibitemShut {NoStop}%
\bibitem [{\citenamefont {Brown}\ \emph {et~al.}(2017)\citenamefont {Brown},
  \citenamefont {Laubscher}, \citenamefont {Kesselring},\ and\ \citenamefont
  {Wootton}}]{brown_poking_2017}%
  \BibitemOpen
  \bibfield  {author} {\bibinfo {author} {\bibfnamefont {B.~J.}\ \bibnamefont
  {Brown}}, \bibinfo {author} {\bibfnamefont {K.}~\bibnamefont {Laubscher}},
  \bibinfo {author} {\bibfnamefont {M.~S.}\ \bibnamefont {Kesselring}},\ and\
  \bibinfo {author} {\bibfnamefont {J.~R.}\ \bibnamefont {Wootton}},\
  }\bibfield  {title} {\bibinfo {title} {Poking holes and cutting corners to
  achieve clifford gates with the surface code},\ }\href
  {https://doi.org/10.1103/PhysRevX.7.021029} {\bibfield  {journal} {\bibinfo
  {journal} {Phys. Rev. X}\ }\textbf {\bibinfo {volume} {7}},\ \bibinfo {pages}
  {021029} (\bibinfo {year} {2017})}\BibitemShut {NoStop}%
\bibitem [{\citenamefont {Nakata}\ \emph {et~al.}(2021)\citenamefont {Nakata},
  \citenamefont {Zhao}, \citenamefont {Okuda}, \citenamefont {Bannai},
  \citenamefont {Suzuki}, \citenamefont {Tamiya}, \citenamefont {Heya},
  \citenamefont {Yan}, \citenamefont {Zuo}, \citenamefont {Tamate},
  \citenamefont {Tabuchi},\ and\ \citenamefont
  {Nakamura}}]{nakata_quantum_2021}%
  \BibitemOpen
  \bibfield  {author} {\bibinfo {author} {\bibfnamefont {Y.}~\bibnamefont
  {Nakata}}, \bibinfo {author} {\bibfnamefont {D.}~\bibnamefont {Zhao}},
  \bibinfo {author} {\bibfnamefont {T.}~\bibnamefont {Okuda}}, \bibinfo
  {author} {\bibfnamefont {E.}~\bibnamefont {Bannai}}, \bibinfo {author}
  {\bibfnamefont {Y.}~\bibnamefont {Suzuki}}, \bibinfo {author} {\bibfnamefont
  {S.}~\bibnamefont {Tamiya}}, \bibinfo {author} {\bibfnamefont
  {K.}~\bibnamefont {Heya}}, \bibinfo {author} {\bibfnamefont {Z.}~\bibnamefont
  {Yan}}, \bibinfo {author} {\bibfnamefont {K.}~\bibnamefont {Zuo}}, \bibinfo
  {author} {\bibfnamefont {S.}~\bibnamefont {Tamate}}, \bibinfo {author}
  {\bibfnamefont {Y.}~\bibnamefont {Tabuchi}},\ and\ \bibinfo {author}
  {\bibfnamefont {Y.}~\bibnamefont {Nakamura}},\ }\bibfield  {title} {\bibinfo
  {title} {Quantum circuits for exact unitary $t$-designs and applications to
  higher-order randomized benchmarking},\ }\href
  {https://doi.org/10.1103/PRXQuantum.2.030339} {\bibfield  {journal} {\bibinfo
   {journal} {PRX Quantum}\ }\textbf {\bibinfo {volume} {2}},\ \bibinfo {pages}
  {030339} (\bibinfo {year} {2021})}\BibitemShut {NoStop}%
\bibitem [{\citenamefont {Hines}\ \emph {et~al.}(2024)\citenamefont {Hines},
  \citenamefont {Hothem}, \citenamefont {Blume-Kohout}, \citenamefont
  {Whaley},\ and\ \citenamefont {Proctor}}]{Hines_fully_2024}%
  \BibitemOpen
  \bibfield  {author} {\bibinfo {author} {\bibfnamefont {J.}~\bibnamefont
  {Hines}}, \bibinfo {author} {\bibfnamefont {D.}~\bibnamefont {Hothem}},
  \bibinfo {author} {\bibfnamefont {R.}~\bibnamefont {Blume-Kohout}}, \bibinfo
  {author} {\bibfnamefont {B.}~\bibnamefont {Whaley}},\ and\ \bibinfo {author}
  {\bibfnamefont {T.}~\bibnamefont {Proctor}},\ }\bibfield  {title} {\bibinfo
  {title} {Fully scalable randomized benchmarking without motion reversal},\
  }\href {https://doi.org/10.1103/PRXQuantum.5.030334} {\bibfield  {journal}
  {\bibinfo  {journal} {PRX Quantum}\ }\textbf {\bibinfo {volume} {5}},\
  \bibinfo {pages} {030334} (\bibinfo {year} {2024})}\BibitemShut {NoStop}%
\bibitem [{\citenamefont {Boixo}\ \emph {et~al.}(2018)\citenamefont {Boixo},
  \citenamefont {Isakov}, \citenamefont {Smelyanskiy}, \citenamefont {Babbush},
  \citenamefont {Ding}, \citenamefont {Jiang}, \citenamefont {Bremner},
  \citenamefont {Martinis},\ and\ \citenamefont
  {Neven}}]{Boixo_characterizing_2018}%
  \BibitemOpen
  \bibfield  {author} {\bibinfo {author} {\bibfnamefont {S.}~\bibnamefont
  {Boixo}}, \bibinfo {author} {\bibfnamefont {S.~V.}\ \bibnamefont {Isakov}},
  \bibinfo {author} {\bibfnamefont {V.~N.}\ \bibnamefont {Smelyanskiy}},
  \bibinfo {author} {\bibfnamefont {R.}~\bibnamefont {Babbush}}, \bibinfo
  {author} {\bibfnamefont {N.}~\bibnamefont {Ding}}, \bibinfo {author}
  {\bibfnamefont {Z.}~\bibnamefont {Jiang}}, \bibinfo {author} {\bibfnamefont
  {M.~J.}\ \bibnamefont {Bremner}}, \bibinfo {author} {\bibfnamefont {J.~M.}\
  \bibnamefont {Martinis}},\ and\ \bibinfo {author} {\bibfnamefont
  {H.}~\bibnamefont {Neven}},\ }\bibfield  {title} {\bibinfo {title}
  {Characterizing quantum supremacy in near-term devices},\ }\href
  {https://doi.org/10.1038/s41567-018-0124-x} {\bibfield  {journal} {\bibinfo
  {journal} {Nature Physics}\ }\textbf {\bibinfo {volume} {14}},\ \bibinfo
  {pages} {595–600} (\bibinfo {year} {2018})}\BibitemShut {NoStop}%
\bibitem [{\citenamefont {Arute}\ \emph {et~al.}(2019)\citenamefont {Arute},
  \citenamefont {Arya}, \citenamefont {Babbush}, \citenamefont {Bacon},
  \citenamefont {Bardin}, \citenamefont {Barends}, \citenamefont {Biswas},
  \citenamefont {Boixo}, \citenamefont {Brandao}, \citenamefont {Buell},
  \citenamefont {Burkett}, \citenamefont {Chen}, \citenamefont {Chen},
  \citenamefont {Chiaro}, \citenamefont {Collins}, \citenamefont {Courtney}
  \emph {et~al.}}]{arute_quantum_2019}%
  \BibitemOpen
  \bibfield  {author} {\bibinfo {author} {\bibfnamefont {F.}~\bibnamefont
  {Arute}}, \bibinfo {author} {\bibfnamefont {K.}~\bibnamefont {Arya}},
  \bibinfo {author} {\bibfnamefont {R.}~\bibnamefont {Babbush}}, \bibinfo
  {author} {\bibfnamefont {D.}~\bibnamefont {Bacon}}, \bibinfo {author}
  {\bibfnamefont {J.~C.}\ \bibnamefont {Bardin}}, \bibinfo {author}
  {\bibfnamefont {R.}~\bibnamefont {Barends}}, \bibinfo {author} {\bibfnamefont
  {R.}~\bibnamefont {Biswas}}, \bibinfo {author} {\bibfnamefont
  {S.}~\bibnamefont {Boixo}}, \bibinfo {author} {\bibfnamefont {F.~G. S.~L.}\
  \bibnamefont {Brandao}}, \bibinfo {author} {\bibfnamefont {D.~A.}\
  \bibnamefont {Buell}}, \bibinfo {author} {\bibfnamefont {B.}~\bibnamefont
  {Burkett}}, \bibinfo {author} {\bibfnamefont {Y.}~\bibnamefont {Chen}},
  \bibinfo {author} {\bibfnamefont {Z.}~\bibnamefont {Chen}}, \bibinfo {author}
  {\bibfnamefont {B.}~\bibnamefont {Chiaro}}, \bibinfo {author} {\bibfnamefont
  {R.}~\bibnamefont {Collins}}, \bibinfo {author} {\bibfnamefont
  {W.}~\bibnamefont {Courtney}}, \emph {et~al.},\ }\bibfield  {title} {\bibinfo
  {title} {Quantum supremacy using a programmable superconducting processor},\
  }\href {https://doi.org/10.1038/s41586-019-1666-5} {\bibfield  {journal}
  {\bibinfo  {journal} {Nature}\ }\textbf {\bibinfo {volume} {574}},\ \bibinfo
  {pages} {505} (\bibinfo {year} {2019})}\BibitemShut {NoStop}%
\bibitem [{\citenamefont {Lancien}\ and\ \citenamefont
  {Majenz}(2020)}]{lancien_weak_2020}%
  \BibitemOpen
  \bibfield  {author} {\bibinfo {author} {\bibfnamefont {C.}~\bibnamefont
  {Lancien}}\ and\ \bibinfo {author} {\bibfnamefont {C.}~\bibnamefont
  {Majenz}},\ }\bibfield  {title} {\bibinfo {title} {Weak approximate unitary
  designs and applications to quantum encryption},\ }\href
  {https://doi.org/10.22331/q-2020-08-28-313} {\bibfield  {journal} {\bibinfo
  {journal} {Quantum}\ }\textbf {\bibinfo {volume} {4}},\ \bibinfo {pages}
  {313} (\bibinfo {year} {2020})}\BibitemShut {NoStop}%
\bibitem [{\citenamefont {Ji}\ \emph {et~al.}(2018)\citenamefont {Ji},
  \citenamefont {Liu},\ and\ \citenamefont {Song}}]{ji_pseudorandom_2018}%
  \BibitemOpen
  \bibfield  {author} {\bibinfo {author} {\bibfnamefont {Z.}~\bibnamefont
  {Ji}}, \bibinfo {author} {\bibfnamefont {Y.-K.}\ \bibnamefont {Liu}},\ and\
  \bibinfo {author} {\bibfnamefont {F.}~\bibnamefont {Song}},\ }\bibfield
  {title} {\bibinfo {title} {Pseudorandom {Quantum} {States}},\ }in\ \href
  {https://doi.org/10.1007/978-3-319-96878-0_5} {\emph {\bibinfo {booktitle}
  {Advances in {Cryptology} - {CRYPTO} 2018}}},\ \bibinfo {editor} {edited by\
  \bibinfo {editor} {\bibfnamefont {H.}~\bibnamefont {Shacham}}\ and\ \bibinfo
  {editor} {\bibfnamefont {A.}~\bibnamefont {Boldyreva}}}\ (\bibinfo {address}
  {Cham},\ \bibinfo {year} {2018})\ pp.\ \bibinfo {pages}
  {126--152}\BibitemShut {NoStop}%
\bibitem [{\citenamefont {Begusic}\ and\ \citenamefont
  {Chan}(2025)}]{begusic_real-time_2024}%
  \BibitemOpen
  \bibfield  {author} {\bibinfo {author} {\bibfnamefont {T.}~\bibnamefont
  {Begusic}}\ and\ \bibinfo {author} {\bibfnamefont {G.~K.-L.}\ \bibnamefont
  {Chan}},\ }\bibfield  {title} {\bibinfo {title} {Real-time operator evolution
  in two and three dimensions via sparse pauli dynamics},\ }\href
  {https://doi.org/10.1103/PRXQuantum.6.020302} {\bibfield  {journal} {\bibinfo
   {journal} {PRX Quantum}\ }\textbf {\bibinfo {volume} {6}},\ \bibinfo {pages}
  {020302} (\bibinfo {year} {2025})}\BibitemShut {NoStop}%
\bibitem [{\citenamefont {Gonz{\'{a}}lez-Garc{\'{i}}a}\ \emph
  {et~al.}(2025)\citenamefont {Gonz{\'{a}}lez-Garc{\'{i}}a}, \citenamefont
  {Cirac},\ and\ \citenamefont {Trivedi}}]{gonzalez-garcia_pauli_2024}%
  \BibitemOpen
  \bibfield  {author} {\bibinfo {author} {\bibfnamefont {G.}~\bibnamefont
  {Gonz{\'{a}}lez-Garc{\'{i}}a}}, \bibinfo {author} {\bibfnamefont {J.~I.}\
  \bibnamefont {Cirac}},\ and\ \bibinfo {author} {\bibfnamefont
  {R.}~\bibnamefont {Trivedi}},\ }\bibfield  {title} {\bibinfo {title} {Pauli
  path simulations of noisy quantum circuits beyond average case},\ }\href
  {https://doi.org/10.22331/q-2025-05-05-1730} {\bibfield  {journal} {\bibinfo
  {journal} {{Quantum}}\ }\textbf {\bibinfo {volume} {9}},\ \bibinfo {pages}
  {1730} (\bibinfo {year} {2025})}\BibitemShut {NoStop}%
\bibitem [{\citenamefont {Mi}\ \emph {et~al.}(2021)\citenamefont {Mi},
  \citenamefont {Roushan}, \citenamefont {Quintana}, \citenamefont {Mandra},
  \citenamefont {Marshall}, \citenamefont {Neill}, \citenamefont {Arute},
  \citenamefont {Arya}, \citenamefont {Atalaya}, \citenamefont {Babbush},
  \citenamefont {Bardin}, \citenamefont {Barends}, \citenamefont {Basso},
  \citenamefont {Bengtsson}, \citenamefont {Boixo}, \citenamefont {Bourassa}
  \emph {et~al.}}]{mi_information_2021}%
  \BibitemOpen
  \bibfield  {author} {\bibinfo {author} {\bibfnamefont {X.}~\bibnamefont
  {Mi}}, \bibinfo {author} {\bibfnamefont {P.}~\bibnamefont {Roushan}},
  \bibinfo {author} {\bibfnamefont {C.}~\bibnamefont {Quintana}}, \bibinfo
  {author} {\bibfnamefont {S.}~\bibnamefont {Mandra}}, \bibinfo {author}
  {\bibfnamefont {J.}~\bibnamefont {Marshall}}, \bibinfo {author}
  {\bibfnamefont {C.}~\bibnamefont {Neill}}, \bibinfo {author} {\bibfnamefont
  {F.}~\bibnamefont {Arute}}, \bibinfo {author} {\bibfnamefont
  {K.}~\bibnamefont {Arya}}, \bibinfo {author} {\bibfnamefont {J.}~\bibnamefont
  {Atalaya}}, \bibinfo {author} {\bibfnamefont {R.}~\bibnamefont {Babbush}},
  \bibinfo {author} {\bibfnamefont {J.~C.}\ \bibnamefont {Bardin}}, \bibinfo
  {author} {\bibfnamefont {R.}~\bibnamefont {Barends}}, \bibinfo {author}
  {\bibfnamefont {J.}~\bibnamefont {Basso}}, \bibinfo {author} {\bibfnamefont
  {A.}~\bibnamefont {Bengtsson}}, \bibinfo {author} {\bibfnamefont
  {S.}~\bibnamefont {Boixo}}, \bibinfo {author} {\bibfnamefont
  {A.}~\bibnamefont {Bourassa}}, \emph {et~al.},\ }\bibfield  {title} {\bibinfo
  {title} {Information scrambling in quantum circuits},\ }\href
  {https://doi.org/10.1126/science.abg5029} {\bibfield  {journal} {\bibinfo
  {journal} {Science}\ }\textbf {\bibinfo {volume} {374}},\ \bibinfo {pages}
  {1479} (\bibinfo {year} {2021})}\BibitemShut {NoStop}%
\bibitem [{\citenamefont {Darmawan}\ and\ \citenamefont
  {Poulin}(2017)}]{darmawan_tnsurfacecode_2017}%
  \BibitemOpen
  \bibfield  {author} {\bibinfo {author} {\bibfnamefont {A.~S.}\ \bibnamefont
  {Darmawan}}\ and\ \bibinfo {author} {\bibfnamefont {D.}~\bibnamefont
  {Poulin}},\ }\bibfield  {title} {\bibinfo {title} {Tensor-network simulations
  of the surface code under realistic noise},\ }\href
  {https://doi.org/10.1103/PhysRevLett.119.040502} {\bibfield  {journal}
  {\bibinfo  {journal} {Phys. Rev. Lett.}\ }\textbf {\bibinfo {volume} {119}},\
  \bibinfo {pages} {040502} (\bibinfo {year} {2017})}\BibitemShut {NoStop}%
\bibitem [{\citenamefont {Breuckmann}\ and\ \citenamefont
  {Eberhardt}(2021)}]{breuckmann_ldpc_2021}%
  \BibitemOpen
  \bibfield  {author} {\bibinfo {author} {\bibfnamefont {N.~P.}\ \bibnamefont
  {Breuckmann}}\ and\ \bibinfo {author} {\bibfnamefont {J.~N.}\ \bibnamefont
  {Eberhardt}},\ }\bibfield  {title} {\bibinfo {title} {Quantum low-density
  parity-check codes},\ }\href {https://doi.org/10.1103/PRXQuantum.2.040101}
  {\bibfield  {journal} {\bibinfo  {journal} {PRX Quantum}\ }\textbf {\bibinfo
  {volume} {2}},\ \bibinfo {pages} {040101} (\bibinfo {year}
  {2021})}\BibitemShut {NoStop}%
\bibitem [{\citenamefont {Bravyi}\ and\ \citenamefont
  {Terhal}(2009)}]{bravyi_no-go_2009}%
  \BibitemOpen
  \bibfield  {author} {\bibinfo {author} {\bibfnamefont {S.}~\bibnamefont
  {Bravyi}}\ and\ \bibinfo {author} {\bibfnamefont {B.}~\bibnamefont
  {Terhal}},\ }\bibfield  {title} {\bibinfo {title} {A no-go theorem for a
  two-dimensional self-correcting quantum memory based on stabilizer codes},\
  }\href@noop {} {\bibfield  {journal} {\bibinfo  {journal} {New J. Phys.}\
  }\textbf {\bibinfo {volume} {11}},\ \bibinfo {pages} {043029} (\bibinfo
  {year} {2009})}\BibitemShut {NoStop}%
\bibitem [{\citenamefont {Rakovszky}\ and\ \citenamefont
  {Khemani}(2024)}]{rakovszky_physics_2024}%
  \BibitemOpen
  \bibfield  {author} {\bibinfo {author} {\bibfnamefont {T.}~\bibnamefont
  {Rakovszky}}\ and\ \bibinfo {author} {\bibfnamefont {V.}~\bibnamefont
  {Khemani}},\ }\bibfield  {title} {\bibinfo {title} {The {Physics} of (good)
  {LDPC} {Codes} {II}. {Product} constructions},\ }\bibfield  {journal}
  {\bibinfo  {journal} {arXiv:2402.16831}\ }\href
  {https://doi.org/10.48550/arXiv.2402.16831} {10.48550/arXiv.2402.16831}
  (\bibinfo {year} {2024})\BibitemShut {NoStop}%
\bibitem [{\citenamefont {Behrends}\ and\ \citenamefont
  {Beri}(2024)}]{behrends_statistical_2024}%
  \BibitemOpen
  \bibfield  {author} {\bibinfo {author} {\bibfnamefont {J.}~\bibnamefont
  {Behrends}}\ and\ \bibinfo {author} {\bibfnamefont {B.}~\bibnamefont
  {Beri}},\ }\bibfield  {title} {\bibinfo {title} {Statistical mechanical
  mapping and maximum-likelihood thresholds for the surface code under generic
  single-qubit coherent errors},\ }\bibfield  {journal} {\bibinfo  {journal}
  {arXiv:2410.22436}\ }\href {https://doi.org/10.48550/arXiv.2410.22436}
  {10.48550/arXiv.2410.22436} (\bibinfo {year} {2024})\BibitemShut {NoStop}%
\bibitem [{\citenamefont {Bao}\ and\ \citenamefont
  {Anand}(2024)}]{bao_phases_2024}%
  \BibitemOpen
  \bibfield  {author} {\bibinfo {author} {\bibfnamefont {Y.}~\bibnamefont
  {Bao}}\ and\ \bibinfo {author} {\bibfnamefont {S.}~\bibnamefont {Anand}},\
  }\bibfield  {title} {\bibinfo {title} {Phases of decodability in the surface
  code with unitary errors},\ }\bibfield  {journal} {\bibinfo  {journal}
  {arXiv:2411.05785}\ }\href {https://doi.org/10.48550/arXiv.2411.05785}
  {10.48550/arXiv.2411.05785} (\bibinfo {year} {2024})\BibitemShut {NoStop}%
\bibitem [{\citenamefont {Aaronson}\ and\ \citenamefont
  {Gottesman}(2004)}]{Aaronson2004}%
  \BibitemOpen
  \bibfield  {author} {\bibinfo {author} {\bibfnamefont {S.}~\bibnamefont
  {Aaronson}}\ and\ \bibinfo {author} {\bibfnamefont {D.}~\bibnamefont
  {Gottesman}},\ }\bibfield  {title} {\bibinfo {title} {Improved simulation of
  stabilizer circuits},\ }\href {https://doi.org/10.1103/PhysRevA.70.052328}
  {\bibfield  {journal} {\bibinfo  {journal} {Phys. Rev. A}\ }\textbf {\bibinfo
  {volume} {70}},\ \bibinfo {pages} {052328} (\bibinfo {year}
  {2004})}\BibitemShut {NoStop}%
\bibitem [{\citenamefont {Krastanov}\ \emph {et~al.}(2024)\citenamefont
  {Krastanov}, \citenamefont {Ahmad}, \citenamefont {Viswanathan},
  \citenamefont {Pardis}, \citenamefont {Micciche}, \citenamefont {Lapeyre},
  \citenamefont {Rabqubit}, \citenamefont {gsommers}, \citenamefont {Hofmann},
  \citenamefont {Bhatt}, \citenamefont {Meligrana}, \citenamefont
  {Benzillaist}, \citenamefont {Zhao}, \citenamefont {IsaacP1234},
  \citenamefont {Göttgens}, \citenamefont {ShuGe-MIT}, \citenamefont {Holy},
  \citenamefont {Dang}, \citenamefont {adrianariton},\ and\ \citenamefont
  {ismoldayev}}]{Krastanov2024}%
  \BibitemOpen
  \bibfield  {author} {\bibinfo {author} {\bibfnamefont {S.}~\bibnamefont
  {Krastanov}}, \bibinfo {author} {\bibfnamefont {F.}~\bibnamefont {Ahmad}},
  \bibinfo {author} {\bibfnamefont {P.}~\bibnamefont {Viswanathan}}, \bibinfo
  {author} {\bibfnamefont {S.}~\bibnamefont {Pardis}}, \bibinfo {author}
  {\bibfnamefont {A.}~\bibnamefont {Micciche}}, \bibinfo {author}
  {\bibfnamefont {J.}~\bibnamefont {Lapeyre}}, \bibinfo {author} {\bibnamefont
  {Rabqubit}}, \bibinfo {author} {\bibnamefont {gsommers}}, \bibinfo {author}
  {\bibfnamefont {T.}~\bibnamefont {Hofmann}}, \bibinfo {author} {\bibfnamefont
  {A.}~\bibnamefont {Bhatt}}, \bibinfo {author} {\bibfnamefont
  {A.}~\bibnamefont {Meligrana}}, \bibinfo {author} {\bibnamefont
  {Benzillaist}}, \bibinfo {author} {\bibfnamefont {C.}~\bibnamefont {Zhao}},
  \bibinfo {author} {\bibnamefont {IsaacP1234}}, \bibinfo {author}
  {\bibfnamefont {L.}~\bibnamefont {Göttgens}}, \bibinfo {author}
  {\bibnamefont {ShuGe-MIT}}, \bibinfo {author} {\bibfnamefont
  {T.}~\bibnamefont {Holy}}, \bibinfo {author} {\bibfnamefont {T.}~\bibnamefont
  {Dang}}, \bibinfo {author} {\bibnamefont {adrianariton}},\ and\ \bibinfo
  {author} {\bibnamefont {ismoldayev}},\ }\href
  {https://doi.org/10.5281/zenodo.13950271} {\bibinfo {title}
  {Quantumsavory/quantumclifford.jl: v0.9.12}} (\bibinfo {year}
  {2024})\BibitemShut {NoStop}%
\end{thebibliography}%

\end{document}